\documentclass[sn-mathphys-num]{sn-jnl}

\usepackage[usenames,dvipsnames,svgnames,table]{xcolor}
\definecolor{darkblue} {rgb}{.1,.5,0.65}
\definecolor{darkgreen}{rgb}{.1,.18,.82}
\hypersetup{
	urlcolor   = blue,         
	linkcolor  = darkgreen,     
	citecolor  = darkblue,    
}

\usepackage{graphicx}
\usepackage[nointegrals]{wasysym} %
\usepackage[export]{adjustbox}
\usepackage{amsmath,amsfonts,amssymb,latexsym}
\usepackage{hhline}
\usepackage{bm}
\usepackage{verbatim}
\usepackage{enumitem}
\usepackage{mathrsfs}
\usepackage{slashed}
\usepackage{empheq}

\newcommand{\p}{\partial}

\newcommand{\lan}{\langle}
\newcommand{\ran}{\rangle}

\newcommand{\unit}{\mathbf{1}}

\newcommand{\da}{{\dagger}}

\newcommand{\ra}{\rightarrow}

\newcommand{\wt}{\widetilde}

\newcommand{\uva}{{\mathbf{\hat a}}}
\newcommand{\uvb}{{\mathbf{\hat b}}}
\newcommand{\uvc}{{\mathbf{\hat c}}}

\newcommand{\uvx}{{\mathbf{\hat x}}}
\newcommand{\uvy}{{\mathbf{\hat y}}}

\newcommand{\bfzero}{{\mathbf{0}}}

\renewcommand{\(}{\left(}
\renewcommand{\)}{\right)}
\renewcommand{\[}{\left[}
\renewcommand{\]}{\right]}
\newcommand{\mt}{\mapsto}

\newcommand{\tp}{\otimes}

\newcommand{\D}{\nabla}

\newcommand\bpm            {\begin{pmatrix}}
	\newcommand\epm           {\end{pmatrix}}

\newcommand{\ms}{\medskip}

\def\app#1#2{%
	\mathrel{%
		\setbox0=\hbox{$#1\sim$}%
		\setbox2=\hbox{%
			\rlap{\hbox{$#1\propto$}}%
			\lower1.1\ht0\box0%
		}%
		\raise0.25\ht2\box2%
	}%
}

\newcommand{\tw}{\textwidth}

\newcommand{\ct}{\Theta}

\newcommand{\inv}{^{-1}}

\newcommand{\ope}\odot

\usepackage{manfnt}

\newcommand{\bi}{\begin{itemize}}
	\newcommand{\ei}{\end{itemize}}

\newcommand{\igpfc}[1]{\vcenter{\hbox{\includegraphics[width=.5\textwidth]{#1}}}}

\usepackage{amsthm}

\newcommand\bd            {\begin{definition}}
	\newcommand\ed            {\end{definition}}
\newcommand\bt            {\begin{theorem}}
	\newcommand\et            {\end{theorem}}
\newcommand\be            {\begin{equation}}
	\newcommand\ee            {\end{equation}}
\newcommand\ba            {\begin{aligned}}
	\newcommand\ea            {\end{aligned}}
\newcommand\bea{\begin{equation}\begin{aligned}}
		\newcommand\eea{\end{aligned}\end{equation}}

\newcommand{\sss}{\subsubsection}
\renewcommand{\ss}{\subsection}

\renewcommand{\a}{\alpha}
\renewcommand{\b}{\beta}
\renewcommand{\d}{\delta}
\newcommand{\De}{\Delta}
\newcommand{\g}{\gamma}

\newcommand{\s}{\sigma}

\renewcommand{\l}{\lambda}
\renewcommand{\L}{\Lambda}

\renewcommand{\o}{\omega}
\renewcommand{\O}{\Omega}

\renewcommand{\r}{\rho}

\newcommand{\bfsig}{{\boldsymbol{\sigma}}}

\newcommand{\bfl}{{\boldsymbol{\lambda}}}

\newcommand{\bfm}{\mathbf{m}}

\newcommand{\bfr}{\mathbf{r}}
\newcommand{\bfs}{\mathbf{s}}

\newcommand{\bfx}{\mathbf{x}}

\newcommand{\bfy}{\mathbf{y}}

\newcommand{\zt}{\mathbb{Z}_2}

\newcommand{\EE}{\mathbb{E}}

\newcommand{\rr}{\mathbb{R}}

\newcommand{\zz}{\mathbb{Z}}

\newcommand{\mcc}{\mathcal{C}}
\newcommand{\mck}{\mathcal{K}}
\newcommand{\mcb}{\mathcal{B}}
\newcommand{\mcf}{\mathcal{F}}
\newcommand{\mco}{\mathcal{O}}

\newcommand{\mce}{\mathcal{E}}

\newcommand{\mcd}{\mathcal{D}}
\newcommand{\mcl}{\mathcal{L}}

\newcommand{\mcs}{\mathcal{S}}

\newcommand{\mch}{\mathcal{H}}
\newcommand{\mca}{\mathcal{A}}

\newcommand{\mcn}{\mathcal{N}}

\newcommand{\mcr}{\mathcal{R}}

\newcommand{\sfA}{\mathsf{A}}
\newcommand{\sfB}{\mathsf{B}}
\newcommand{\sfC}{\mathsf{C}}
\newcommand{\sfD}{\mathsf{D}}
\newcommand{\sfE}{\mathsf{E}}

\newcommand{\sfN}{\mathsf{N}}

\newcommand{\sfP}{\mathsf{P}}

\newcommand{\sfS}{\mathsf{S}} %

\newcommand{\sfW}{\mathsf{W}}

\newcommand{\sfa}{\mathsf{a}}
\newcommand{\sfb}{\mathsf{b}}

\newcommand{\sfd}{\mathsf{d}}

\newcommand{\sfn}{\mathsf{n}}

\newcommand{\sfp}{\mathsf{p}}

\newcommand{\sfs}{\mathsf{s}} 
\newcommand{\sft}{\mathsf{t}}

\usepackage[mathscr]{eucal} %
\newcommand{\sca}{\mathscr{A}}
\newcommand{\scb}{\mathscr{B}}

\newcommand{\scp}{\mathscr{P}}

\newcommand{\scs}{\mathscr{S}}

\renewcommand{\tt}[1]{\mathtt{#1}}

\renewcommand{\k}[1]{|#1\rangle}

\usepackage{dcolumn}

\usepackage{pifont}
\usepackage{ulem}

\usepackage{tableof}

\usepackage{CJKutf8}
\newcommand{\tdec}{T_{\sf dec}}
\newcommand{\emp}{\varnothing}
\newcommand{\plog}{p_{\sf log}}
\newcommand{\dist}{{\sf dist}}

\newcommand{\qq}{\qquad} 

\newtheorem{theorem}{Theorem}[section]
\newtheorem{proposition}[theorem]{Proposition}
\newtheorem{lemma}[theorem]{Lemma}
\newtheorem{remark}{Remark}[section]
\newtheorem{definition}{Definition}[section]

\newtheorem{conjecture}{Conjecture}
\newtheorem{corollary}[theorem]{Corollary}

\usepackage{tabularx}

\newcommand{\plc}{{\sf PLF}}
\newcommand{\flc}{{\sf FLF}} 
\newcommand{\plf}{{\sf PLF}}

\newcommand{\tsim}{\sft_{\sf sim}}
\newcommand{\tsimmin}{\sft_{\sf sim,min}}
\newcommand{\tsimmax}{\sft_{\sf sim,max}}
\newcommand{\lag}{{\sf lag}}
\newcommand{\asy}{{\sf desynch}}

\newcommand{\plogae}{{p_{\sf log}^{\sfa,(\mce)}}}
\newcommand{\tdece}{\tdec^{(\mce)}}
\newcommand{\ploge}{\plog^{(\mce)}}
\newcommand{\tdeca}{T_{\sf dec}^\sfa}
\newcommand{\tdecae}{T_{\sf dec}^{\sfa,(\mce)}}
\newcommand{\supp}{{\rm supp}}
\newcommand{\Tr}{{\sf Tr}}
\renewcommand{\bot}{\bigotimes} 
\newcommand{\new}{{\sf new}}
\newcommand{\old}{{\sf old}}
\newcommand{\cp}{\Phi}
\renewcommand{\bfl}{{\mathbf{l}}}
\newcommand{\fut}{{\sf future}}
\newcommand{\pres}{{\sf present}}
\newcommand{\dsy}{{\sf desync}}
\newcommand{\lind}{\dsy}
\newcommand{\rlog}{{\r_{\sf log}}}
\newcommand{\qu}{{\sf qu}}
\newcommand{\cl}{{\sf cl}}
\newcommand{\trec}{t_{\sf rec}}
\newcommand{\trev}{t_{\sf rev}}

\newcommand{\normalmath}[1]{\begingroup\normalfont\mathversion{normal}#1\endgroup}

\begin{document}
	\title{Fast offline decoding with local message-passing automata}
	
	\author[]{\fontsize{12pt}{16pt}\selectfont \fnm{Ethan} \sur{Lake}}
	
	\affil[]{\small Department of Physics, University of California Berkeley, Berkeley, CA 94720}

	\abstract{
		We present a local offline decoder for topological codes that operates according to a parallelized message-passing framework. The decoder works by passing messages between anyons, with the contents of received messages used to move nearby anyons towards one another. We prove the existence of a threshold, and show that in a system of linear size \normalmath{$L$}, decoding terminates with an \normalmath{$O((\log L)^\eta)$} average-case runtime, where \normalmath{$\eta$} is a small constant. For the toric code subject to i.i.d Pauli noise, our decoder has \normalmath{$\eta=1$} and a threshold at a noise strength of \normalmath{$p_c\approx 7.3\%$}. 
	}
	
	\maketitle
	
	\tableofcontents
	
	\section{Introduction} \label{sec:intro}

	The problem of performing error correction in quantum memories is one of growing theoretical and experimental importance.  
	From a practical standpoint, the ability to reliably and efficiently correct errors is an essential step towards developing useful quantum computers. 
	Quantum error correction is also playing an increasingly significant role in fundamental physics: error correction algorithms which operate locally in space are closely connected to non-equilibrium quantum phases of matter, the understanding of which is still in its infancy. 
	
	The last several years have seen significant experimental progress in quantum error correction, particularly in the case of the surface code \cite{campbell2024series}. Significant challenges nevertheless remain, including on the classical control side, where standard approaches can struggle to keep pace with demanding error correction cycle times and to integrate smoothly with cryogenic quantum hardware \cite{google_qec,reilly2019challenges,brennan2025classical}. Decoding algorithms which overcome these challenges will need to be fast, accurate, and simple.

	Many state-of-the-art approaches focus on parallelizing (to various degrees) global decoding algorithms like minimal weight perfect matching and union-find (see e.g. \cite{liyanage2024fpga,demarti2024decoding,google_qec,ziad2024local,wu2023fusion,higgott2025sparse,barber2023real,delfosse2021almost}). 
	To fully parallelize the decoding process at the hardware level, one may imagine laying out an array of small classical processors next to the qubits, with each processor using a small amount of working memory to locally measure syndromes, communicate measurement results to its neighbors, and locally apply feedback. Systems which perform error correction in this way are known as {\it cellular automaton} (CA) decoders \cite{Harrington2004,herold2015cellular,herold2017cellular,balasubramanian2024local,kubica2019cellular,breuckmann2016local}, which in addition to their manifest parallelization, have the benefit of operating according to (usually) very simple rules; as such, they are particularly well suited to implementation with FPGAs or similar parallelized architectures. 
	
	CA decoders are also of interest for the role they play in the study of non-equilibrium and mixed-state phases of matter, the classification of which is an ongoing program in condensed matter physics (see e.g. \cite{coser2019classification,rakovszky2024defining,sang2024mixed,Dennis_2002,fan2024diagnostics,cubitt2015stability}). In one classification scheme, two states are said to be in the same mixed-state phase if there exist local Lindbladians that rapidly turn one state into the other \cite{coser2019classification,sang2024mixed}. Decoders which may be formulated as (quasi)local Lindbladians---which all CA decoders can be, as we will see---thus provide a direct way of investigating the landscape of mixed-state quantum phases of matter.

	This work develops local CA decoders for topological codes (Abelian or non-Abelian)
	defined by a simple message-passing framework. Our decoders operate by finding a matching between anyons in a parallel, decentralized fashion: under the decoding dynamics, anyons send out signals broadcasting their locations, and use local feedback based on incoming signals to move nearby anyons towards one another. This scheme is naturally translation invariant, can be implemented in continuous time with a strictly local Lindbladian (with each term having support on $O(1)$ sites), and on a system of linear size $L$, achieving a suppression of the sub-threshold logical error rate which is superpolynomial in $L$ requires access to only $c\log \log L$ classical control bits per site, where $c$ is a small $O(1)$ constant.

	The decoders we develop in this work are designed for {\it offline} decoding: given a logical state $\rlog$ and a noise channel $\mce$, they recover $\rlog$ when input $\mce(\rlog)$, under the assumptions that topological charge measurements can be made perfectly, and that the noise is not too strong (sometimes known as the ``code capacity scenario'').\footnote{``Offline decoding'' could also be used in a situation in which a decoder is run on an accumulated spacetime trajectory of (unreliable) syndrome measurements. Our decoders can also be used in this context.} This is enough for establishing mixed-state phase equivalence in the sense discussed above, and our work thus provides an explicit construction of a strictly local Lindbladian that quickly connects noisy topologically ordered states to their clean counterparts (our decoders can be regarded as an alternative to the explicit RG-based decoders developed for this purpose in Ref.~\cite{sang2024mixed}, which lack translation invariance and require nonlocal operations over $\O({\rm polylog}(L))$ length scales). 
	
	Offline decoders do not by themselves stabilize fault-tolerant memories. This must be done with an {\it online} (or ``real-time'') decoder, which operates in a setting where noise occurs continuously and measurement outcomes are unreliable. With the notable exception of codes allowing for single-shot error correction \cite{stahl2024single,kubica2022single,bombin2015single}---where the assumption of perfect measurements is not needed---the speed of an offline decoder is therefore generally not of direct practical significance. More broadly though, understanding how to quickly perform offline decoding is a necessary stepping stone on the path towards designing efficient fault-tolerant quantum memories, and extensions of the present decoders to this setting are presented in Ref.~\cite{online}.

	\begin{figure}
		\centering \includegraphics[width=.8\tw]{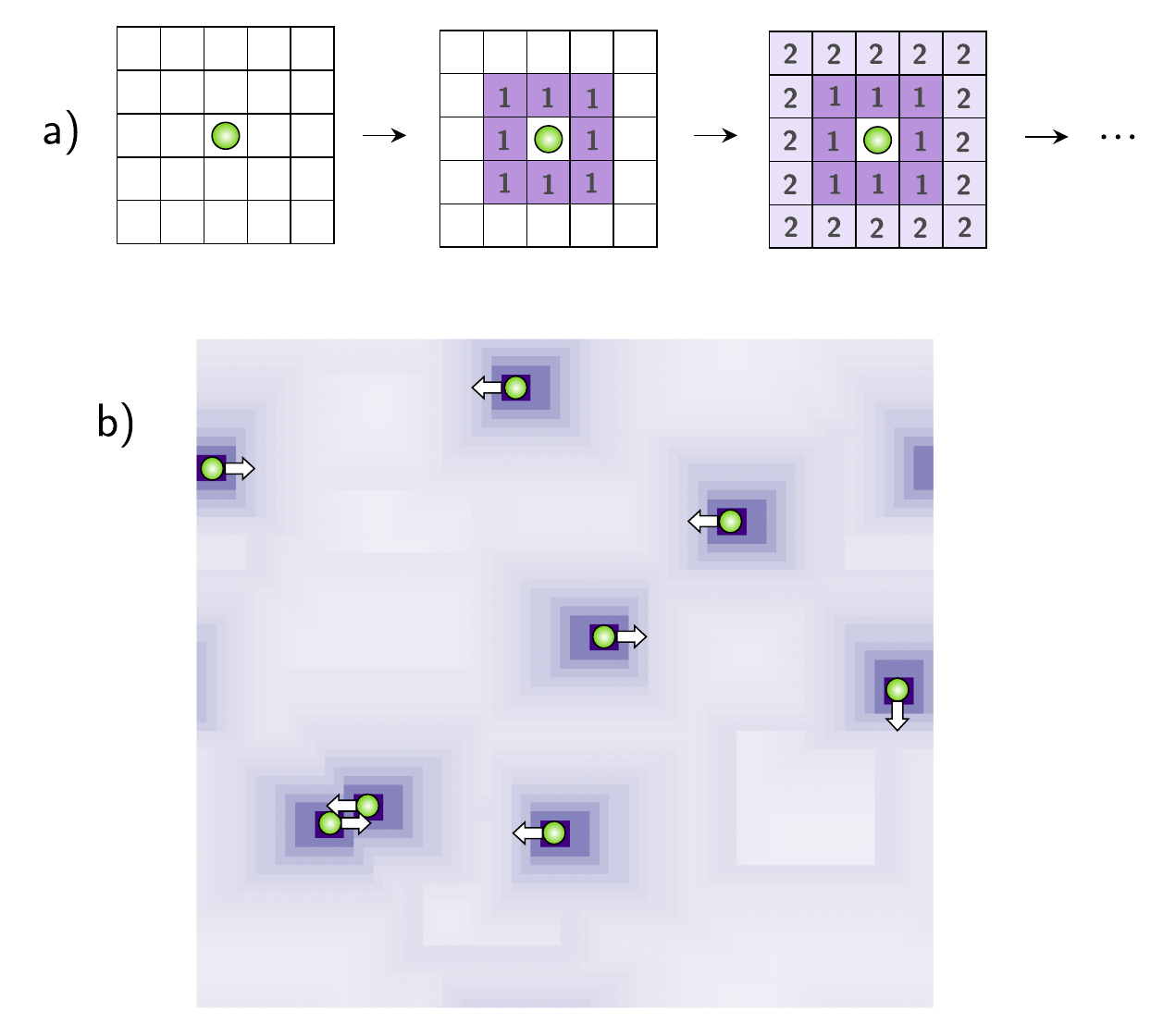}
		\caption{\label{fig:inspir_schem} Illustration of the message passing architecture in two dimensions. {\bf {\sf a)}} A schematic of how messages are produced when the decoder operates synchronously. A site hosting an anyon (green circle) sends out messages which propagate outward from the anyon's location at a constant speed (with respect to the $\infty$-norm), and increase their value by 1 at each time step.
			The decoder makes use of this information by moving each anyon in the direction $(\pm\uvx$ or $\pm\uvy$) of the received message with the smallest value (to convey directional information, the decoder actually employs four different types of messages; this distinction is not shown in the figure).  
			{\bf {\sf b)}} A snapshot of the decoding dynamics in a system of size $L=75$. Darker colors indicate messages with smaller values, and the directions along which each anyon will move at the next time step are indicated by the white arrows. If the message transmission speed was infinite, each anyon would move in the direction of its nearest neighbor. For finite message speeds, it is possible for an anyon to move towards the location where an anyon was annihilated in the recent past (the rightmost anyon in the figure being an example). }
	\end{figure}
	
	\ss{Summary and overview}

	In this subsection we provide an extended summary of the results of this paper. 
	
	
	This work focuses on the problem of designing parallelized offline CA decoders for topological  codes with point-like anyons (in any dimension). 
	Our interest is in decoders which perform computations and apply feedback according to translationally invariant spatially local CA update rules.\footnote{\label{local_footnote} Our emphasis on spatial locality perhaps deserves some elaboration. In a strict sense, any offline decoder that operates non-locally on a system-size-dependent scale $R_{\sf nl}(L)$ can be rendered local by adding additional degrees of freedom that communicate instructions to different parts of the lattice before nonlocal operations are applied. Doing so slows down decoding by an amount at least of order $R_{\sf nl}(L)$, and compromises the simplicity of the decoding scheme. Moreover, this construction leads to an extensively long increase in waiting times between the application of nontrivial feedback operations, which causes problems when trying to extend to the setting of real-time decoding. Our interest in this work is therefore in decoders which are ``naturally'' local,  with the ratio between local classical data-processing steps and nontrivial feedback operations being fixed to an $O(L^0)$ constant. } These decoders are built from an array of classical processors, with one processor located at the location of each minimal stabilizer generator.  
	A processor at site $\bfr$ is restricted to performing the following actions at each time step (temporarily working in a discrete-time setting for simplicity): 
	\begin{enumerate} 
		\item Measuring the topological charge  
		 at site $\bfr$, obtaining outcome $\s_\bfr$.
		\item Performing a CA update $\mca$ by modifying its internal state in a way that depends on $\s_\bfr$ and the internal states of its nearest neighbors, but not on $\bfr$ or the elapsed time.
		\item Performing a local feedback operation $\mcf$\footnote{We phrase this as performing feedback on the system, but since we are only concerned with offline decoding in this work, all feedback can of course be absorbed into classical post-processing. } conditioned on the value of its internal state to the qubits near site $\bfr$.
	\end{enumerate}
	
	We will be concerned only with decoders which operate {\it locally}, meaning that the ratio of CA updates to nontrivial feedback operations is a constant number, independent of the system size $L$. Decoders which require communication over a length scale that diverges as $L\ra\infty$ in between each feedback operation---even if the divergence is only as $\log L$---will not be regarded as local in this work (c.f. the discussion in footnote \ref{local_footnote}). 

	\sss{Construction}
	
	We begin by briefly describing how our decoders are constructed. For concreteness, we will specialize in most of the paper to the case of the 2D surface code, so that $\s_\bfr \in \{ \pm1\}$. The generalization to arbitrary (potentially non-Abelian) topological codes is essentially immediate, and will be explained in sec.~\ref{ss:nonab}. For simplicity, we will first explain their operation in discrete time, where a global synchronizing clock ensures that all sites of the system update simultaneously at each time step. 
	
	Consider dynamics whereby at each time step, each anyon moves by one lattice site in the direction of the anyon closest to it (without attempting to predetermine if the nearest neighbor is the one that it pairs with in a correct matching). While this decoding rule is extremely simple---and can in fact be shown to possess a threshold, using the techniques of this work---locating the closest anyon to a given site explicitly requires non-local communication. The decoding algorithm developed in this work can be viewed as a simple way of making this decoding rule completely local. 
	The automaton rule $\mca$ and feedback operations $\mcf$ are defined by the following rules, which are illustrated schematically in fig.~\ref{fig:inspir_schem} and spelled out in detail in sec.~\ref{sec:msg_passing}. At each time step, the following occur in sequence: 
	\begin{enumerate}
		\item Each anyon emits messages in all directions, which propagate outward from its location at a constant speed $v$. Messages come in four different types, with each type propagating along one of the directions $\pm\uvx,\pm\uvy$. Message propagation occurs in such a way that by time $r/v$, an anyon broadcasts messages to all sites in an $\infty$-ball of radius $r$ centered on its location (c.f. fig.~\ref{fig:inspir_schem}).
		\item The value of a message is updated to equal the number of time steps it has been propagating. 
		\item If two messages of the same type overlap, the message with smaller value overwrites the message with larger value. 
		\item An operator $\mcf$ is applied which moves each anyon by one lattice site in the direction ($\pm\uvx$ or $\pm\uvy$) of the smallest-valued message it receives.
	\end{enumerate}
	The first three steps constitute the automaton rule $\mca$, and the pair $(\mca,\mcf)$ together define a time-delayed version of the ``move towards nearest anyon'' rule discussed above.

	We will see in sec.~\ref{sec:offline} that in a system of linear size $L$, we can take messages to disappear after propagating for time $O(\log L)$ without compromising the existence of a threshold (at the expense of getting a slower, albeit still superpolynomia, suppression of the sub-threshold logical error rate). With this rule, the largest value a message can take is accordingly $O(\log L)$, and the number of control bits our decoder uses per site will then be shown to scale as $c\log_2 \log_2 L$, where the constant prefactor $c$ will be shown to be no larger than $10$. The divergence as $L\ra\infty$ is so weak that we will not be concerned with addressing it.\footnote{If $L$ is upper bounded by the number of atoms in the observable universe, $\log_2 \log_2 L < 10$.} 
	
	In order to fully parallelize error correction, decoding needs to be performed asynchronously, with each site executing operations independently of a global control clock. We will address this issue by extending our decoder to operate in continuous time, where $\mca,\mcf$ are performed by an appropriate Lindbladian.  The most naive extension---where each site performs the same updates as in the discrete-time setting---is not a priori guaranteed to work, as the behavior of synchronous cellular automata are generically completely altered in this setting. Nevertheless, we develop a desynchronization scheme which provably allows any discrete-time decoder to be extended to a continuous time Lindbladian without changing the value of its threshold or its expected decoding time. This construction is based on the  ``marching soldiers'' rule of Berman and Simon \cite{berman1988investigations,gacs_synch_slides} (closely related to the restricted solid-on-solid model of surface growth \cite{kim1989growth}), and operates by using a continuous time system to simulate the computation performed by a discrete time one (assuming the absence of errors in the classical control bits). The details are provided in sec.~\ref{sec:lind}.

	\sss{Decoding performance}  \label{ss:offline_intro}
	
	We now summarize the performance of our decoders. Given a noise channel $\mce$, the two diagnostics we are most interested in are:
	\begin{enumerate}
		\item The {\it logical failure rate} $\ploge$, defined as the probability for the decoder to output a logical state other than $\rlog$ when input a noisy state $\mce(\rlog)$, and 
		\item The {\it expected decoding time} $\tdece$, which is defined to capture the average-case time complexity of the decoding problem, and equals the expected amount of time the decoder must run on input $\mce(\rlog)$ before returning a logical state. 
	\end{enumerate}
	
	For simplicity, we restrict our attention throughout to incoherent stochastic noise channels of the form
	\be \mce(\r) = \sum_\bfs p^{(\mce)}(\bfs) P_\bfs \r P_\bfs^\da,\ee 
	where $P_\bfs$ is a tensor product of Pauli operators labeled by the string $\bfs$, and $p^{(\mce)}(\bfs)$ is a probability distribution. Our only requirement on $p^{(\mce)}(\bfs)$ is that the marginal probability for $\bfs$ to be nonzero on all sites in a subset $A$ be upper bounded by $ p^{|A|}$ for some constant $p$ (regardless of the geometric locality of $A$); such noise channels will be referred to as being ``$p$-bounded''. We say that the decoder has a threshold at $p_c$ under $\mce$ if $\ploge$ vanishes in the thermodynamic limit for all $p<p_c$.  
	
	The restriction to this class of stochastic noise models is common in fault tolerance proofs (c.f. Refs.~\cite{balasubramanian2024local,gottesman2024surviving,gottesman2009introductionquantumerrorcorrection,knill1998resilient}). However, this type of noise does not encompass (rather benign) types of coherent errors, e.g. a small unitary rotation performed on all qubits. The techniques used to treat non-Markovian noise in \cite{aliferis2005quantum,gottesman2024surviving} can be used to extend our analysis to general types of coherent noise, but the details are somewhat involved and will be omitted. 

	\begin{figure}
		\centering 
		\includegraphics[width=.8\tw]{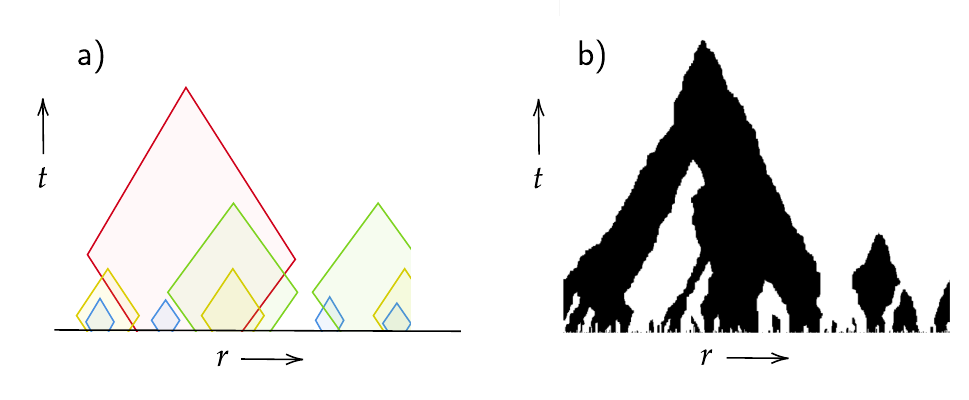}
		\caption{\label{fig:erosion_demo} Clustering and error correction under the message-passing decoder, illustrated for the case of the 1D repetition code. {\bf {\sf a)}} A spacetime schematic illustrating how the decoder corrects errors. The noise in the input state can be recursively organized into clusters of different sizes. The smallest clusters are indicated by the locations where the blue regions intersect the $t=0$ axis. The decoder is guaranteed to eliminate these clusters before they have a chance to merge with one another, and the spacetime support of anyons contained in these small clusters is restricted to the interiors of the blue diamonds. Disregarding these small clusters, the remaining part of the noise is broken up into larger clusters (yellow regions). The decoder similarly eliminates anyons in these regions before they have a chance to merge. This clustering process is continued up to the scale of the largest clusters in the system (here, the red region), with $\ploge,\tdece$ determined by how these large clusters are eliminated. {\bf {\sf b)}} elimination of clusters in action. A visualization of the decoding dynamics for a system of size $L=512$, initialized on a random bit string. Black and white areas denoting regions where the bits take values of $0$ and $1$, respectively. }
	\end{figure}
	
	A heuristic understanding for how our decoder performs error correction can be obtained as follows. Consider for simplicity the case where $\mce$ implements i.i.d bit-flip noise of strength $p$ on each link,	
	and examine a pair of anyons $(a_1,a_2)$ created by $\mce$ which are separated by a distance $\De$. If $\De$ is smaller than the average distance $d_{\sf avg}$ between anyons, then $a_1$, $a_2$ are likely to be closer to one another than to any other anyon; if this is the case, they will exchange messages and annihilate each other in time $O(\De)$. 
	
	Consider on the other hand a case where $\De > d_{\sf avg}$, so that some smaller anyon pairs are initially present between $a_1$ and $a_2$. The messages sent by $a_1$ will then not be received by $a_2$ (and vice versa) until the intervening smaller pairs have been eliminated. If these small pairs are eliminated in a time linear in their separation, then the messages from $a_1$ will reach $a_2$ in a time linear in $\De$, and $a_1,a_2$ will subsequently annihilate one another by a time which is again $O(\De)$.\footnote{While $a_1,a_2$ may fuse with anyons in smaller pairs before communication between them is established, these fusion processes essentially amount to moving the initial positions of $a_1,a_2$ by amounts of order $d_{\sf avg}$. This does not affect the scaling of the annihilation time with $\De$.}
	Very roughly, we can imagine self-consistently extending this analysis to all scales: if pairs of size $\ell<\De$ annihilate in a time linear in $\ell$, then so too do pairs of size $\De$, implying that in fact all pairs are annihilated in a time linear in their size. This logic requires pairs to be relatively well-separated from one another, with the value of $p$ below which this holds determining the threshold error strength. 
	
	We now estimate the scaling of $\ploge, \tdece$ with $p$ and $L$ using this logic (not putting any cutoff on the largest message size for simplicity). At small enough $p$, the probability for a pair $(a_1,a_2)$ to be separated by $\De$ is $\sim p^\De$.
	Assuming the error correcting dynamics proceeds as above, a logical error will only occur if $\De \sim L$, giving a logical error rate exponentially small in $L$. To estimate $\tdece$, note that for a typical noise realization, at least one pair of size $\De$ will be present in the system provided $p^\De L^D \gtrsim 1$. The largest pair in the system will thus with high probability have size $\De_{\sf max} \sim \log (L) / \log(1/p)$ and will be well-isolated from pairs of comparable size. The arguments above then suggest $\tdec \sim \De_{\sf max} \sim \log(L)$. Summarizing, this intuition suggests the scalings 
	\be \label{expected_scaling} \plog = (p/p_c)^{\ct(L)},\qq  \tdec = \ct(\log L).\ee 
	
	Sections~\ref{sec:erosion} and \ref{sec:offline} make this intuition mathematically rigorous by proving the following result: 
	
	\ms\begin{theorem}[Threshold for offline decoding, informal]\label{thm:1d_offline}
		Consider the message-passing decoder operating on any topological code defined on a $D$-dimensional square lattice of linear size $L$, with either periodic or open boundary conditions, and in either discrete or continuous time. If we allow the decoder to store at least $2D\log_2(L)$ control bits at each site, then there exist positive $O(1)$ constants $p_c,\b,\eta$ such that as long as $p<p_c$, then under any $p$-bounded stochastic noise model $\mce$ we have 
		\be \label{intro_scaling} \ploge = (p/p_c)^{\O(L^\b)},\qq \tdece = O((\log L)^{\eta}). \ee 
		If we require that the decoder store only $O(\log \log L)$ control bits per site, the bound on the logical error rate is loosened to $\ploge = (p/p_c)^{\O ((\log L)^\g)}$ for a positive $O(1)$ constant $\g>1$, while the bound on $\tdece $ remains the same. 
	\end{theorem}
	We do not attempt to give tight bounds on $p_c$ or the allowable values of $\b$ and $\eta$, and content ourselves with the (likely extremely conservative) bounds $p_c > 1/29^D,\b > \log_{13}(2) \approx 0.27, \, \eta < \log_2(13) \approx 3.7$. Without additional assumptions on $\mce$ it is in fact impossible to improve the bounds on $\b,\eta$ to the values of $\b=\eta=1$ suggested by  \eqref{expected_scaling}, as in sec.~\ref{ss:general_comments} we will construct explicit examples of $p$-bounded noise models with $\b = \log_6(5) <1,\, \eta = \log_5(6) >1$. 
	These ``Gerrymandered'' noise models create fractal-like patterns of anyons and involve highly non-local correlations in space, and we conjecture that $\eta=1$ for noise with short-range spatial correlations. 
	
	We prove Theorem~\ref{thm:1d_offline} by formalizing the multi-scale approach schematically described above, showing that sufficiently nearby anyons rapidly annihilate, that anyons at larger scales annihilate after the smaller pairs have disappeared, and so on. The strategy of the proof is standard and consists of two parts, the logic of which are illustrated in fig.~\ref{fig:erosion_demo}. The first part proves a ``linear erosion'' property \cite{gacs2024probabilistic,rakovszky2024defining,gacs_eroder_slides}: all anyons in a sufficiently well-isolated cluster of linear size $\De$ annihilate against one another in a time proportional to $\De$, even when they (inevitably) receive signals from anyons outside the cluster. The second part uses a variant of a general renormalization-type argument originally due to Gacs \cite{gacs2001reliable,ccapuni2021reliable}---iterations of which have proven useful in many threshold proofs, see e.g. Refs.~\cite{bravyi2011analytic,Harrington2004,duclos2013fault,kubica2019cellular,paletta2025high}---which shows that the anyons created by the noise can be grouped into well-isolated clusters on a hierarchy of increasing length scales, with large clusters rapidly becoming rare provided $p$ is sufficiently small. The theorem then follows from a standard percolation-type argument, with the erosion property showing that different clusters never merge with one another under the decoding dynamics, despite surviving for long enough to exchange messages with one another. $\ploge$ is then fixed by the probability of the noise creating a system-spanning cluster, and $\tdece$ by the typical size of the largest cluster.

	In sec.~\ref{sec:numerics} we use Monte Carlo simulations to quantitatively address the performance of our decoder for surface codes under i.i.d Pauli noise.  
	Our simplest decoding scheme yields a threshold consistent with $p_c = 1/2$ in 1D (the repetition code) and $p_c \approx 7.3\%$ in 2D, and produces scalings consistent with $\ploge  = \exp(\ct(L))$, $\tdece = \ct(\log L)$ for $p$ not too much smaller than $p_c$, consistent with \eqref{expected_scaling} (for very small $p$, the aforementioned fractal-like noise patterns may dominate the scaling of $\ploge$). We believe that any local decoder must have $\tdece = \O(\log L)$ (see sec.~\ref{ss:general_comments}); if this is true, the scaling of $\tdece$ is optimal. Beyond increasing the numerical value of $p_c$ by a few (no more than 3 \cite{Dennis_2002}) percent, the only improvement would be to reduce the number of control bits per site to $O(1)$ instead of $O(\log \log L)$. While this is an academic point, in app.~\ref{app:constant_bits} we nevertheless give a construction using only $O(1)$ bits per site that we believe (but do not prove) retains a threshold.

	\ss{Comparison with prior work} 
	
	We now briefly compare our approach with a few other decoders studied in the literature. 
	
	The decoders most similar to ours---which served as direct inspiration for the present work---are the ``field-based'' CA surface code decoders of Refs.~\cite{herold2015cellular,herold2017cellular}. These decoders transmit information about syndrome locations in a physically natural fashion, using a scalar field $\phi$ updated according to a lattice discretization of the Poisson equation 
	\be \p_t \phi_\bfr = D\D^2 \phi_\bfr + \frac{1-\s_{\bfr}}2,\ee 
	where $\s_\bfr \in \{\pm1\}$ is the syndrome at site $\bfr$ and $D$ is a diffusion constant. In the limit where $\phi$ equilibrates instantaneously, it mediates an attractive interaction between anyons that decays with distance as $1/r^{D_\phi-2}$, where $D_\phi$ is the dimension of the space that $\phi$ propagates in. Ref.~\cite{herold2015cellular} gave numerical evidence suggesting that a decoder with $D_\phi=3$ has a threshold under offline decoding when used to correct a 2D toric code on the boundary of a 3D bulk.\footnote{The techniques developed in this work can be used to rigorously prove the existence of a threshold in this case.}
	
	For the purposes of this work, the most significant drawback of these decoders is that they do not meet our definition of locality, with the $D_\phi=3$ model requiring a $(\log L)^2$-divergent classical communication speed (in app.~\ref{app:no_pdes}, we argue that all PDE-based models suffer from this problem). Our decoder can thus be viewed as an improvement of this construction which brings down the communication speed to an $O(1)$ constant and reduces the number of control bits used per site from $\o((\log L)^3)$ \cite{herold2015cellular} to $O(\log \log L)$.

	Our decoders also bear some resemblance to 1D classical automata designed to perform majority voting in a translation-invariant, parallelized way. The most salient decoders in this class are the GKL automaton
	\cite{gacs1978one} and Toom's closely related two-line voting rule \cite{toom1995cellular}, variants and extensions of which have also appeared in the physics literature, c.f. Refs.~\cite{guedes2024quantum,lang2018strictly,paletta2025high}.\footnote{With the possible exception of \cite{paletta2025high}, these automata only have thresholds as offline (not online) decoders.} These automata operate according to a similar type of message-passing architecture as employed by the 1D versions of our decoders, with the added feature of having ``anti-messages'' which ensure that {\it all} messages emitted by a domain wall pair are erased in a time linear in the pair's initial  separation (the GKL and two-line voting automata also have the advantage of only requiring $O(1)$ control bits per site). While these rules perform well for the 1D repetition code, adapting them directly to the setting of topological codes in $D>1$ is not particularly straightforward.

	Aside from field-based decoders, the two existing types of cellular automaton decoders for the 2D toric code
	are Harrington's decoder \cite{Harrington2004} and its extensions \cite{duclos2013fault,schotte2022fault} (see also \cite{breuckmann2016local,sang2024mixed}), and the decoder of Ref.~\cite{balasubramanian2024local}, which can be regarded as a quantum extension of Tsirelson's automaton \cite{cirel2006reliable} (these decoders are designed for online decoding, but they can of course be used for offline decoding as well). Both of these decoders can loosely be thought of as eliminating anyons according to a locally-implemented  RG scheme, with the different length scales involved in the RG process hard-coded into the structure of the dynamics. 
	For the purposes of this work, drawbacks of these approaches are that they are only naturally defined on systems of size $n^l$ for integers $n,l$, and are numerically observed to have thresholds at rather low error rates, with the simplest formulation of  Ref.~\cite{balasubramanian2024local} achieving $p_c \lesssim 2\%$ for offline decoding. These features make them---at least in their current form---less practical for experimental implementation and numerical study.\footnote{In 2D, these decoders are also not naturally translation invariant in space or time, although this problem can be fixed formally by adding additional clock registers at each site and initializing the dynamics with an appropriate (fine-tuned) choice of initial conditions for the classical variables.} 

	
	\ss{Outlook}

	While our explicit construction is specialized to the setting of topological codes defined on $\zz^D$, it should be straightforward to generalize the message passing architecture to codes defined on more general types of graphs. 	
	From an implementation perspective, it would also be valuable to quantitatively investigate different types of message passing schemes, and to systematically study how the numerical value of the threshold depends on the choice of architecture and parameters involved (such as the message transmission speed). Relatedly, it could be fruitful to study how decoding can be optimized to take advantage of particular features in experimentally-relevant noise models. One notable example of this is provided by Ref.~\cite{chirame2024stabilizing}, which constructs a very simple local Lindbladian decoder that corrects heralded erasure errors. While these will never be the sole source of errors, there are certainly scenarios in which they dominate, and it is natural to wonder to what extent our decoder can be optimized for such cases.


	From both a practical and theoretical perspective, the most immediate question raised by this work is whether this and related message-passing schemes can operate as real-time decoders.
	As of writing, all known parallelized real-time decoding strategies (with provable thresholds) perform error correction according to a hierarchical RG-inspired approach. One is then naturally led to wonder: is hierarchical organization---either hard-coded into the circuit dynamics, as in Refs.~\cite{Harrington2004,balasubramanian2024local}, or emergent, as in Gacs' automaton \cite{gacs2001reliable}---necessary for performing local quantum error correction? 
	
	This question is addressed in Ref.~\cite{online}, where we rigorously show that the offline decoders presented in this work, as well as the field-based decoders proposed in Ref.~\cite{herold2017cellular}, do not have thresholds for real-time decoding. Fortunately, that work also demonstrates how to modify the message-passing architecture in a way which does guarantee the existence of a threshold, at least given the assumption that the noise acts only on the quantum system, and not on the data manipulated by the classical control infrastructure.

	\section{Preliminaries}\label{sec:generalities}

	In this preparatory section, we fix notation, introduce the general framework within which our decoders are defined, and discusses the types of noise models considered in this work. To streamline the presentation, the discussion in this section will make two simplifications. First, we will focus exclusively on synchronous decoders that operate in discrete time. The extension to asynchronous Lindbladian decoders which operate in continuous time case will be given later in sec.~\ref{sec:lind}. Second, we will explicitly adopt notation appropriate for the $D$-dimensional surface code. The generalization to arbitrary (potentially non-Abelian) topological codes is immediate, and will be deffered to sec.~\ref{ss:nonab}. 
	
		As a final preliminary comment, we note that we will focus solely on correcting logical qubits whose logical operators are string-like, since it is errors associated with these logical sectors whose correction is hardest to accomplish in a local / parallelized way. If $D$ is large enough so that some of the logical operators are supported on submanifolds of dimension $d>1$, these logical sectors can easily be corrected locally using a variant of Toom's rule \cite{kubica2019cellular}, and we will consequently ignore such logical sectors in what follows.

	\ss{Automaton decoders}\label{ss:aut_decoders}

	We begin by fixing notation. Our decoders are defined for codes built from qubits living on the links of a square lattice $\L$, which we will usually take to be $\zz^D_L$ with periodic boundary conditions (all of our results also hold for open boundary conditions with minor modifications; the details are provided in app.~\ref{app:open}). $\L_l$ will denote the set of links in $\L$, and $B_\bfr^{(\sfp)}(\De)$ will denote the ball of radius $\De$ in $\L$ centered at a lattice point $\bfr$ defined with respect to the $\sfp$ norm (when the choice of $\sfp$ is not important, the superscript will be omitted). The stabilizers which define the code space will be denoted by $g_\bfr$, with $\L$ defined such that each lattice point hosts a single stabilizer. A logical state will refer to a density matrix $\r_{\sf log}$ such that $\r_{\sf log}  g_\bfr = g_\bfr \r_{\sf log}  = \r_{\sf log} $ for all $\bfr$. 

	Since our decoders operate by using  autonomous processing and feedback rules, we will find it helpful to formally split the Hilbert space the decoder acts on into quantum and classical parts as $\mch = \mch_{\sf qu} \tp \mch_{\sf cl}$, where $\mch_{\sf qu} = \bot_{\bfl \in \L_l} \mch_{{\sf qu},\bfl}$ and $\mch_{\sf cl} = \bot_{\bfr\in\L} \mch_{{\sf cl},\bfr}$. We will find it helpful to introduce notation where 
	\be \mch_{{\sf cl},\bfr} = {\sf span}\{\k{m_\bfr,\s_\bfr}\},\ee 
	with $\s_\bfr \in \{\pm1\}$ a variable which will store the eigenvalue of $g_\bfr$. We will use boldface font to denote the collection of classical variables at all sites, so that $\bfm = (m_{\bfr_1},m_{\bfr_2},\dots)$ and likewise for $\bfsig$. If an operator $\mco$ either commutes or anti-commutes with every $g_\bfr$, so that $g_\bfr \mco = (-1)^{\s_{\mco,\bfr}} \mco g_{\bfr}$, its syndrome $\bfsig_\mco$ will be defined as the string $\bfsig_\mco = (\s_{\mco,\bfr_1},\s_{\mco,\bfr_2},\dots)$. Finally, we will use the notation $P_\bfs = \bot_{\bfl\in\L_l} X^{s_{1,\bfl}} Z^{s_{2,\bfl}}$ to denote a Pauli operator labeled by the string $\bfs \in \zt\times\zt^{|\L_l|}$.
	
	The discrete-time decoders considered in this work all fall under the framework of the dynamics introduced in the following definition: 
	
	\ms \begin{definition}[synchronous automaton decoder]\label{def:synch_ccqs}
		A {\it synchronous automaton decoder} $\mcd = (\mca,\mcf)$ is a discrete-time dynamics determined by an automaton rule $\mca$ and a set of feedback operators $\mcf$. For an input state $\r$ on $\mch_{\sf qu}$, the decoder first measures the topological charges at each lattice site, producing the initial state $\r_0 = \sum_\bfsig \Pi_\bfsig \r \Pi_\bfsig \tp |\bfm_0,\bfsig\ran\lan\bfm_0,\bfsig|_{\sf cl}$, where $\bfm_0$ is a particular initialization of the control variables and $\Pi_\bfsig = \bot_\bfr (1+\s_\bfr g_\bfr)/2$. At each time step $t$ after this initialization, the following operations in sequence: 
		\begin{enumerate}
			\item The $\bfm$ variables in $\mch_{\sf cl}$ are simultaneously updated according to the automaton rule $\mca$, which uses the values of $\bfm,\bfsig$ to update the values of $\bfm$; we will use notation where $\mca(\bfm,\bfsig)[m_\bfr]$ denotes the value that $m_\bfr$ takes upon applying $\mca$.\footnote{We can in general allow $\mca$ to involve randomness, but for simplicity will stick with notation where $\mca$ is deterministic.} 
			Formally, this is done using a channel $\Phi_\mca$ with Kraus operators
			\be \mck_{\bfm,\bfsig} = \unit_{\sf qu} \tp |\mca(\bfm,\bfsig),\bfsig\ran\lan \bfm,\bfsig|.\ee 
			\item Feedback is applied to $\mch_{\sf qu}$ controlled by $\mch_{\sf cl}$ using a unitary operator $\mcf(\bfm,\bfsig)$ with definite syndrome $\bfsig_\mcf$, and the syndromes are updated accordingly. This is done using a channel $\Phi_\mcf$ with Kraus operators 
			\be \mck_{\bfm,\bfsig} = \mcf(\bfm,\bfsig) \tp |\bfm,\bfsig \bfsig_\mcf\ran\lan \bfm,\bfsig|,\ee 
					where $\bfsig\bfsig_\mcf$ denotes element-wise multiplication. 
		\end{enumerate}
		The state of the system after $t$ steps of this process is $\r_t = (\cp_\mcf \circ \cp_\mca)^t[\r_0]$. 
		
		A {\it local} synchronous automaton decoder is one in which both the automaton rule $\mca$ and the feedback operator $\mcf$ are computed locally. This means that for some fixed $\De$ independent of the system size, 
		\be \label{localcond} \mca(\bfm,\bfsig)[m_\bfr] =  \mca(\bfm|_{B_\bfr(\De)},\bfsig|_{B_\bfr(\De)})[m_\bfr],\qq \mcf(\bfm,\bfsig) = \bigotimes_{\bfr\in\L} \mcf_\bfr(\bfm|_{B_\bfr(\De)},\bfsig|_{B_\bfr(\De)}),\ee 
		where $ \supp(\mcf_\bfr)\subset B_\bfr(\De)$ for each $\bfr$, and the different $\mcf_\bfr$ can be taken to commute without loss of generality.\footnote{If we have different $\mcf_\bfr$ that do not commute, we can decompose $\mcf(\bfm,\bfsig)$ into a brickwork circuit of finite depth $d$, parametrize the depth into this circuit as an additional classical state variable taking on $d$ possible values, and then accordingly update $\mca$ to perform cyclic shifts on this variable. }
		
		Finally, a {\it quantum local} synchronous automaton decoder is one in which the feedback operators $\mcf$ are spatially local, but where the automaton rule $\mca$ needn't be. 
	\end{definition}\ms 
	
		The word ``synchronous'' is used here to refer to the fact that the automaton updates and feedback are applied synchronously, with all sites in the system updating in union at each time step. 
		
	We will be almost exclusively interested in local automaton decoders in this work,
	but at several points will briefly comment on a family of quantum local decoders inspired by \cite{herold2015cellular}.

	This definition shares some similarities with the ``measurement and feedback model'' of Ref.~\cite{balasubramanian2024local}, from which it differs by the addition of the explicit internal state variables $\{m_\bfr\}$ and the automaton rule $\mca$, the absence of an explicit dependence of the feedback operators on $t$, and by not mandating that the syndrome measurement outcomes be erased at each step.

	\ss{Noise models}\label{ss:noise_models}
	
	The input states our decoders will be tasked with decoding are obtained by subjecting logical states to certain types of noise channels. As mentioned in the introduction, we will restrict our attention to $p$-bounded stochastic noise channels, viz. channels which apply an incoherent mixture of Pauli operators to an input state, with the probability to apply any particular Pauli operator $P$ being exponentially suppressed in the size of $P$'s support. This is essentially the most general incoherent noise model one could consider, as no assumption on the spatial correlations between noise events are made. Importantly, this definition does not allow for any coherence between errors, although this assumption can be removed following the techniques of \cite{aliferis2005quantum,gottesman2024surviving}.

	In more detail, the types of stochastic noise channels we consider act as 
	\be\label{incoherent} \mce(\r) = \sum_{\bfs} p^{(\mce)}(\bfs) P_\bfs \r P_\bfs^\da,\ee 
	where 
	$p^{(\mce)}(\bfs)$ is a probability distribution satisfying the following $p$-bounded property: 
	
	\ms \begin{definition}[$p$-bounded distributions]\label{def:pbounded_def}
		A probability distribution $p(R)$ over subsets $R$ of a set $\mcs$ is said to be {\it $p$-bounded} if the marginal probability for a fixed subset $A$ to be contained in $R$ is exponentially small in the size of $A$: 
		\be \sfP_{R \sim p}(A \subset R) \leq p^{|A|}.\ee 
		Similarly, a probability distribution $p(\bfs)$ over Pauli strings will be said to be $p$-bounded if for all $A \subset \L_l$ we have 
		\be \sfP_{\bfs\sim p}(A \subset \supp(P_\bfs)) \leq p^{|A|}.\ee 
		Finally, a {\it $p$-bounded noise channel} is an incoherent channel $\mce$ of the form \eqref{incoherent} such that $p^{(\mce)}(\bfs)$ is  $p$-bounded. 
	\end{definition}\ms

	For CSS codes like the surface code---which is the main example studied in this work---we will always consider decoders $\mcd$ where the classical variables and syndromes can be separated as $\bfm = (\bfm^X,\bfm^Z)$, $\bfsig = (\bfsig^X,\bfsig^Z)$, the automaton rule factored as $\mca = \mca^X \circ \mca^Z$ with $\mca^{X/Z}$ only depending on $\bfm^{X/Z},\bfsig^{Z/X}$, and the feedback operators split as $\mcf(\bfm,\bfsig) = \mcf^X(\bfm^X,\bfsig^Z) \mcf^Z(\bfm^Z,\bfsig^X)$, where $\mcf^{X/Z}$ are built entirely as tensor products of $X$ and $Z$ operators, respectively. We can accordingly split the decoder into $X$ and $Z$ parts as $\mcd = \mcd^X \circ \mcd^Z$, and as discussed in detail in Ref.~\cite{balasubramanian2024local}, it suffices to separately consider the performance of $\mcd^{X/Z}$. While the $X$ and $Z$ errors will in general be correlated with one another, the generality of our $p$-bounded definition means that if we marginalize over e.g. the $Z$ errors, the resulting pattern of $X$ errors remains $p$-bounded. This means that if we construct $\mcd^{X/Z}$ to have thresholds under all $p$-bounded noise channels that apply only $X/Z$-type noise, $\mcd$ will have a threshold under general $p$-bounded noise channels. 
	Using this fact, until we discuss general topological codes in sec.~\ref{ss:nonab} we will focus only on noise channels $\mce$ which apply $p$-bounded $X$ noise, taking 
	\be \label{final_channel} \mce(\r) =  \sum_{\sfN \subset \L_l} p^{(\mce)}(\sfN) X^\sfN \r X^\sfN, \ee 
	where $\sfN$ runs over all subsets of $\L_l$ , $X^\sfN = \bot_{\bfl\in \sfN} X_\bfl$, and $p^{(\mce)}(\sfN)$ is a $p$-bounded distribution over subsets of $\L_l$. 
	To cut down on notation, in what follows we will write $\bfm,\bfsig$ for $\bfm^X$ and $\bfsig^Z$, respectively. 
	
	Since the states considered in \eqref{final_channel} are incoherent mixtures over different noise patterns, we may treat the decoding problem for an input state $\rlog$ as the task of decoding individual states $X^\sfN \r_{\sf log} X^{\sfN}$ sampled from the distribution $p^{(\mce)}(\sfN)$. In what follows we will use the notation $\r \sim \mce,\r_{\sf log}$ to denote a random state that equals $X^\sfN \r_{\sf log} X^{\sfN}$ with probability $p^{(\mce)}(\sfN)$. 
	
	As a final comment, we note that while the $p$-bounded property is often referred to as implying that the noise is ``local'' and has ``decaying spatial correlations'', this is not true in any geometric sense: $\sfP(\bfl,\bfl'\in \sfN)$ needn't factor as $ \sfP(\bfl\in \sfN) \sfP(\bfl'\in \sfN)$, even for arbitrarily large $||\bfl-\bfl'||$. Rather, $p$-boundedness simply implies that the marginal probability of observing a particular pattern $A$ in the noise $\sfN$ is exponentially suppressed in the size of $A$. In the case where connected correlations between points in $\sfN$ decay exponentially in distance, we will say $\mce$ has {\it (spatially) local correlations}.
	
	\begin{figure}
		\centering 
		\includegraphics[width=.9\tw]{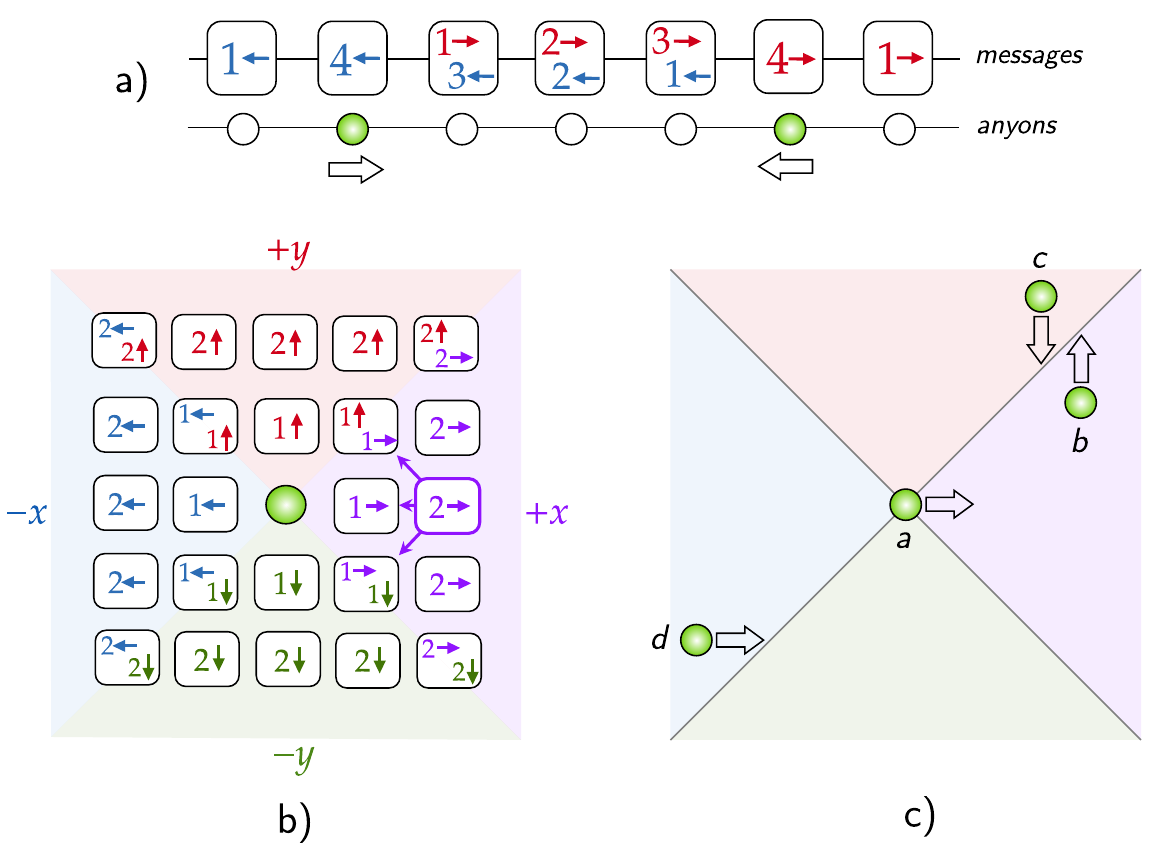}
		\caption{\label{fig:construction} A schematic of the message passing architecture. {\bf {\sf a)}} The message-passing system in 1D. Anyons (green circles) live on a 1D chain, each site of which hosts two message fields $m^\pm$. Once the messages have equilibrated, the value of $m^+_x$ at site $x$ is equal to the distance between site $x$ and the nearest anyon to the left of site $x$; in the figure, a red number $n\ra$ indicates that $m^+=n$ at this site. Similarly, the value of $m^-_x$ is determined by the distance between $x$ and the nearest anyon to the right of $x$; in the figure these are indicated by the blue numbers with $\leftarrow$ arrows. With the message fields shown in the figure, the left anyon will move right, and the right anyon will move left. {\bf {\sf b)}} The message-passing system in 2D. Each site now hosts four message fields $m^{\pm a}$, $a = x,y$, with the figure indicating the fields created by an anyon in the indicated location. $m^{\pm a}$ fields indicate the distance (in the $\infty$ norm) from an anyon in the $\mp \uva$ direction, and propagate along the $\pm a$ section of the anyon's lightfront; in the figure, these lightfront segments are indicated by the shaded regions. The $m^{\pm a}$ field at a site $\bfr$ is updated by consulting the three sites at a infinity-norm distance of 1 in the $\mp\uva$ direction from $\bfr$. As an example, the site outlined in purple is updated according to the three sites indicated by the arrows. 			
			{\bf {\sf c)}} Movement of anyons in 2D. In the limit where messages are transmitted instantly, each anyon moves in the direction of its nearest neighbor (with distance calculated using the $\infty$ norm). If an anyon's nearest neighbor is located in the $\pm a$ section of that anyon's instantaneous lightfront, the anyon moves along $\pm \uva$. Thus the anyon $a$ in the figure moves along $+\uvx$ towards anyon $b$, anyon $b$ moves along $+\uvy$ towards anyon $c$, and so on. When the message transmission speed is finite, each anyon moves instead towards the nearest anyon in its past lightfront.   }
	\end{figure}
	
	\section{Message passing decoders} \label{sec:msg_passing}
	
	The automaton decoders studied in this work will be based on a particular type of message-passing scheme, where information about the spatial locations of anyons is transmitted using the classical variables $\bfm$. In the present subsection we describe the message-passing scheme for synchronous surface code decoders operating on systems with periodic boundary conditions, which will be the setting considered in the majority of the main text. Extensions to general topological codes are given in sec.~\ref{ss:nonab}, the Lindbladian extensions of these decoders are presented in sec.~\ref{sec:lind}, and open boundary conditions are treated in app.~\ref{app:open}.

	We begin in sec.~\ref{ss:1d_setup} by explaining the construction in one dimension with periodic boundary conditions, where it performs error correction on the 1D repetition code. The generalization to higher dimensional systems with periodic boundary conditions is given subsequently in sec.~\ref{ss:general_d_setup}.
	Finally, sec.~\ref{ss:power_law_construction} shows how our construction can be used to produce improved versions of the decoders in Ref.~\cite{herold2015cellular}. 

	\ss{Message passing in 1D} \label{ss:1d_setup} 
	
	We will first explain the construction in one dimension, where the notation is a bit simpler. 
	
	We will use notation where at each syndrome location $r$ (getting rid of the boldface notation in 1D), the classical variables $\bfm$ accessed by the decoder are a collection of two variables $m_r = \{m^+_r,m^-_r\}$, each taking values in the set 
	\be \label{mvals} m^\pm_r \in \{0,1,\dots,m_{\sf max}\}.\ee 
	The value of $m^+_r$ is designed to convey the distance between site $r$ and the nearest anyon to the left of $r$, while that of $m^-_r$ is designed to convey the distance of $r$ from the nearest anyon to the right of $r$; a value of $m^\pm_r = 0$ will be taken to mean that there is {\it no} anyon to the left (right) of $r$. We will let $m_{\sf max}$ remain unspecified for now, although we will see later in sec.~\ref{sec:offline} that having a threshold with super-polynomial suppression of $\ploge$ below threshold only requires $m_{\sf max} = \O(\log L)$. Since  $\lceil\log_2(m_{\sf max})\rceil+1$ bits suffice to store the values of the $m^\pm_r$ fields, the number of classical bits the decoder needs to store at each site in order to produce a threshold scales only as 
	\be {\rm dim}(\mch_{{\sf cl},\bfr}) = \ct(\log \log L),\ee which increases extremely slowly with $L$. 
	
	We first give a schematic overview of how the decoding rules work. The $m^\pm_r$ messages are emitted by anyons, and propagate ballistically along the $\pm$ direction. The value of a message increases by 1 at each time step, so that the value of $m^\pm_r$ is equal to the amount of time the message at this site has been propagating; this allows $m^\pm_r$ to serve as a (time delayed) proxy for the distance between $r$ and the nearest anyon to the left (right) of $r$. The propagation speed of the messages will be denoted by $v \in \zz^{>0}$, with $v$ message updating steps happening in between each application of the feedback operators. 
	
	The feedback rules of the 1D decoder use the $m^\pm_r$ fields to attempt to move each anyon by one lattice site in the direction of the anyon nearest to it. Consider an anyon at site $r$. If $m^+_r = m^-_r = 0$, the anyon has not received  any messages, and no feedback is performed. If only $m^\pm_r$ is nonzero---so that a signal is received  only from an anyon to the left (right) of the one at $r$---feedback is performed to move the anyon along the $\mp $ direction. Finally, if both $m^\pm_r \neq 0$, the anyon is moved along in the $\pm$ direction if $m^\mp_r < m^\pm_r$, since then the message coming from the $\pm$ direction was emitted earlier (hopefully, from a closer anyon) than the message coming from the $\mp$ direction. No anyon movement occurs if $m^+_r = m^-_r$. 
	
	The above automaton and feedback rules are described formally in the following definition: 
	
	\ms \begin{definition}[message passing automaton, 1D]\label{def:1d_msg_passing}
		Let $v$ be a positive integer. The {\it 1D message passing automaton decoder with speed $v$}  is the synchronous automaton decoder with automaton rule $\mca^{(v,1)}$ and feedback operators $\mcf^{(1)}_r$ defined as follows. 
		
		One step of the automaton update $\mca$ sources $m^\pm_r$ messages to the right (left) of anyons, propagates them along the $\pm$ direction, and increases their values so that $m^\pm_r$ is equal to the time at which the message at site $r$ has been propagating, with newer messages rewriting older ones when appropriate. This is described by the rule
		\be \mca^{(1)}(\bfm,\bfsig)[m^\pm_r] = \frac{1-\s_{r\mp1}}2 + \frac{1+\s_{r\mp1}}2
		(1-\d_{m^\pm_{r\mp 1},0}) (m^\pm_{r\mp1}+1),\ee 
		and $\mca^{(v,1)}$ is defined to execute $v$ repetitions of this rule: 
		\be \mca^{(v,1)} =( \mca^{(1)})^v.\ee 
		
		Let $f_r = {\sf sgn}(-m^+_r + m^-_r) \in \{\pm1,0\}$ be the {\it force} at site $r$ (with ${\sf sgn}(0) = 0$). 
		If $f_r = \pm1$, the feedback operators move an anyon on site $r$ the right / left, while they act trivially if $f_r = 0$: 
		\be \mcf^{(1)}_r(\bfm,\bfsig) =X^{\s_rf_r \vee \s_{r+1}(- f_{r+1})}_{\lan r,r+1\ran} \ee 
		where for two variables $x,y \in \{\pm1,0\}$ we define the OR function as 
		\be x\vee y =  \frac{x(x+1)}2 + \frac{y(y+1)}2 - \frac{xy(x+1)(y+1)}4 = \begin{dcases} 1 & x = 1 \,\, {\rm or} \,\, y = 1 \\ 
			0 & {\rm else} \end{dcases}\ee 
	\end{definition}\ms 
	We will let the message passing speed $v$ be arbitrary for now, although in sec.~\ref{sec:offline} we will see that our threshold proof necessitates $v > 2$. 
	
	\begin{figure} 
		\centering
		\includegraphics[width=.8\tw]{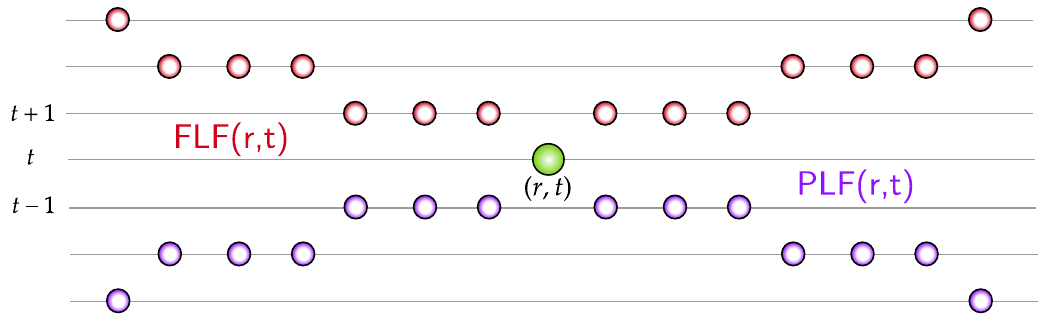}
		\caption{\label{fig:lightfronts1d} lightfronts in one dimension for a message propagation speed of $v=3$; the vertical direction is time (flowing up) and the horizontal direction is space. The red points indicate (part of) $\flc(r,t)$ with $(r,t)$ the location marked with the green point; the purple points indicate $\plc(r,t)$.
		} 
	\end{figure}

	The feedback rules are designed to move an anyon $a$ in the direction of the anyon closest to $a$, but this interaction is necessarily delayed by the finite speed $v$ of the message transmission. The spatial positions of neighboring anyons that are encoded in $m_r^\pm$ at time $t$ are determined not by the locations of anyons at time $t$, but by the locations of anyons at spacetime points that can send messages to the spacetime point $(r,t)$. To make this precise, we define the notion of a lightfront:\footnote{If we were being more faithful to relativistic terminology, we would call this a ``null cone''.} 
	
	\ms	\begin{definition}[lightfronts, 1D]\label{def:lightfronts1d}
		Consider a spacetime coordinate $(r,t)$. The {\it future lightfront} $\flc(r,t)$ is the set of spacetime points reached by messages emitted from a source at spacetime location $r,t$: 
		\be \flc(r,t) = \{ (r',t') \, : \, t'>t, \, v(t'-t-1) < |r-r'| \leq v(t'-t) \}.\ee 
		The {\it past lightfront} $\plc(r,t)$ is the set of points whose future lightfronts contain $(r,t)$: 
		\be \plc(r,t) = \{ (r',t') \, : \, (r,t) \in \flc(r',t')\}.\ee 
		See fig.~\ref{fig:lightfronts1d} for an illustration. 
	\end{definition}\ms 
	Note that due to the discrete time nature of the decoding dynamics, our lightfronts have a staircase-like structure; again see fig.~\ref{fig:lightfronts1d}. 
	
	To describe the anyon dynamics, the following definitions will be useful: 
	\ms \begin{definition}[anyon indicator function; worldlines]\label{def:anyon_indicator}
		The {\it anyon indicator function} $\sfA$ is defined to take on the value 1 on the spacetime points where anyons are present, and the value 0 on other points: 
		\be \label{adef} \sfA(r,t) = \begin{dcases} 
			1 & \s_r = -1 \, \text{at time $t$} \\ 
			0 & \text{else} \end{dcases}.\ee 
		If $\mcs$ is a set of spacetime points, we define $\sfA(\mcs) = \{ (r,t) \in \mcs \, : \, \sfA(r,t) = 1\}$. 
		
		Suppose an anyon $a$ is present at a spacetime point $(r,t)$. The {\it worldline} $\sfW(a)$ of $a$, also written $\sfW(r,t)$, is defined to be the set of spacetime points occupied by this anyon. Formally, we have $(r,t)\in \sfW(a)$, and if $(r',t') \in \sfW(a)$, then either $(r'+1,t'-1)\in \sfW(a)$ or $(r'-1,t'-1)\in \sfW(a)$.\footnote{Ambiguities occurring when two anyons come within a single lattice site of one another will not concern us and can be resolved arbitrarily.} 
	\end{definition}\ms 
	
	The decoding rules of Def.~\ref{def:1d_msg_passing} ensure that an anyon $a$ at spacetime coordinate $(r,t)$ moves in the spatial direction of the point spatially nearest to $r$ in $\sfA(\plc(r,t))$. 	
	A very important property of these rules is that long as $v>1$, the worldline of an anyon never intersects its past lightfront, $\sfW(r,t) \cap \plf(r,t) = \emp$, and thus an anyon can never send signals to itself (anyon motion is slower than message transmission when $v>1$, and hence an anyon can never catch up to the signals it emits). This property is desirable because anyon self-interaction generically prevents anyons from being paired correctly. Our choice of $\mca$ thus eliminates the self-interaction problem suffered by the field-based decoders of Refs.~\cite{herold2015cellular,herold2017cellular}, which were shown to require (in the present language) a value of $v$ which diverged with $L$. app.~\ref{app:no_pdes} shows that this self-interaction problem is in fact generic to models where the classical control fields are updated according to a (discretization of a) local PDE; the present message-passing architecture is therefore crucial for producing a decoder where the speeds of anyon motion and classical communication do not differ by a thermodynamically large amount.

	\ss{Message passing in general dimensions}\label{ss:general_d_setup}
	
	We now generalize the construction of the previous subsection to $D$ dimensions. The basic idea is the same as in one dimension, except now we introduce one message field for each of the $2D$ directions $\pm \uva$, $a = 1,\dots,D$. We write $m_\bfr = \{ m^{\pm a}_r\}$, where as before, each message field is valued in $\{0,1,\dots,m_{\sf max}\}$.
	
	The automaton dynamics is defined so that $m^{\pm a}_\bfr$ propagates along the $\pm \uva$ direction, with the value of $m^{\pm a}_\bfr$ designed to convey the distance between $\bfr$ and the nearest anyon in the direction $-\uva$ from $\bfr$. In a bit more detail, the $m^{\pm a}$ messages emitted from an anyon at $\bfr$ are designed to travel outward and convey information about the anyon's position to lattice sites in the cone-like region 
	\be \sfC^{\pm a}_\bfr= \{ \bfr' \, : \, |(r')^a - r^a| \geq |(r')^b - r^b| \, \,\forall \, \, b, \, (r')^a > \pm r^a\},\ee 
	with the value of a message equal as before to the time over which that message has been propagating (see panel $\sfb$ of fig.~\ref{fig:construction}). $\mca$ is designed such that the messages propagate over a lightfront dictated by the $\infty$ norm, with $m^{\pm a}_\bfr$ indicating the value of $||\bfr - \bfr'||_\infty$, where $\bfr'$ is the closest anyon to $\bfr$  in $\sfC^{\pm a}_\bfr$ (as in 1D, a value of $m^{\pm a}_\bfr = 0$ indicates that no such anyon is present). These rules for $D=2$ are illustrated in panels b and c of fig.~\ref{fig:construction}, and are made precise in the following definition: 
	\ms \begin{definition}[message passing automaton, general dimensions]\label{def:Dd_msg_passing}
		The {\it $D$-dimensional message passing automaton decoder with speed $v$} is the synchronous automaton decoder with automaton rule $\mca^{(v,D)}$ and feedback operators $\mcf^{(D)}_\bfr$ defined as follows. 
		
		Let 	
		\be C^{(a)}_\bfr = \{ \bfr' \in B^{(\infty)}_\bfr(1) \, : \, (r')^a = 0\}\ee 
		denote the unit $\infty$-norm ball in the $D-1$ dimensions normal to $\uva$ centered at $\bfr$, and define $\sfA_\ct(S)$ for a set $S \subset \L$ as 
		\be \sfA_\ct(S) = \begin{dcases} 1 & \exists \,\bfr \in S \, : \, \s_\bfr = -1 \\ 
			0 & \text{else}\end{dcases}.\ee  
		The automaton rule sources $m^{\pm a}_\bfr$ messages on the sites contained in $C^{(a)}_{\bfr \pm \uva}$, propagates them along the $\pm \uva$ direction, and increases their values so that $m^{\pm a}_\bfr$ is equal to the time at which the message at site $\bfr$ has been propagating, overwriting older messages with newer ones when appropriate. This is done using the rule 
		\be \mca^{(D)}(\bfm,\bfsig)[m^{\pm a}_\bfr] = \begin{dcases} 
			1 & \sfA_\ct(C^{(a)}_{\bfr\mp\uva}) = 1
			\\ 0 & \(m^{\pm a}_{\bfr'} = 0 \, \, \forall \, \, \bfr' \in C^{(a)}_{\bfr\mp\uva}\) \wedge \(\sfA_\ct(C^{(a)}_{\bfr\mp\uva}) = 0\) \\ 
			\min_{\bfr' \in C^{(a)}_{\bfr\mp\uva}}\{m^{\pm a}_{\bfr'}\} + 1 & \text{else} 
		\end{dcases}
		,\ee 
		with $\mca^{(v,D)} = (\mca^{(D)})^v$. 
		
		To write down the feedback operators, define the quantities
		\bea \label{blahdef}M_\bfr &= \{ (s,a) \in \{\pm1\}\times \{1,\dots,D\} \, : \, m^{sa}_\bfr > 0\}, \\ (s_\bfr,a_\bfr) & =\begin{dcases} {\sf argmin}'_{(s,a) \in M_\bfr} \{m^{sa}_\bfr\} & M_\bfr \neq \emp \\ 
			(0,0) & \text{else} \end{dcases}, \\ 
		f_\bfr^a & =\d_{a,a_\bfr} {\sf sgn}(-m^{+a}_\bfr + m^{-a}_\bfr),\eea 
		where the function ${\sf argmin}'$ is defined in a way which removes degeneracies between elements of $M_\bfr$ as 
		\be \label{degbreak} {\sf argmin}'_{(s,a)}\{m^{sa}_\bfr\} = {\sf argmin}_{(s,a)}\{ m^{sa}_\bfr + sa/(D+1)\} . \ee 	
		The feedback operators are then 
		\be \label{ddfeedback} \mcf_\bfr^{(D)} = \bot_{a = 1}^D X_{\lan \bfr,\bfr+\uva\ran}^{\s_\bfr f_\bfr^a \vee \s_{\bfr+\uva} (-f_{\bfr+\uva}^a)}.\ee 	
	\end{definition}\ms 
	The necessity of using a method of degeneracy breaking like that of  \eqref{degbreak} can be seen e.g. by considering a state with a pair of anyons at positions $\bfr,\bfr + s\uva +s' \uvb$ with $a \neq b$ and $s,s'\in\{\pm1\}$. 
	
	Animations illustrating anyon movement and field propagation are available at \cite{code}.

	The indicator function $\sfA$, lightfronts, and and worldlines in $D$ dimensions are defined in complete analogy with Def.~\ref{def:lightfronts1d}, with the distances measured using the $\infty$ norm: 
	\ms \begin{definition}[lightfronts, indicator functions, and worldlines: general dimensions]\label{def:lightfronts2d}
		Consider a spacetime coordinate $(\bfr,t)$. The {\it future lightfront} $\flc(\bfr,t)$ is the set of spacetime points reached by messages emitted from a source at $(\bfr,t)$: 
		\be \flc(\bfr,t) = \{(\bfr',t) \, : \, t'>t, \, v(t'-t-1) < ||\bfr-\bfr'||_\infty \leq v(t'-t) \}.\ee 
		The {\it past lightfront} is defined similarly as $\plc(\bfr,t) = \{ (\bfr',t') \, : \, (\bfr,t) \in \flc(\bfr',t') \}.$ We will define the segment of the past lightfront along the $\pm\uva$ direction as 
		\be \plf^{\pm a}(\bfr,t) = \{ (\bfr',t') \in \plf(\bfr,t) \, : \, \pm(r')^a > \pm r^a\}.\ee 
		
		The {\it anyon indicator function} $\sfA$ takes on the value 1 (0) on spacetime points where anyons are (not) present: 
		\be \sfA(\bfr,t) = \begin{dcases} 1 & \s_\bfr = -1 \,\, \text{at time $t$} \\ 
			0 & \text{else} \end{dcases}.\ee 
		If $\mcs$ is a set of spacetime points, $\sfA(\mcs) = \{ (\bfr,t) \in \mcs \, : \, \sfA(\bfr,t) = 1\}$.
		
		Let an anyon $a$ be present at a spacetime point $(\bfr,t)$. The {\it worldline} $\sfW(a)$ of $a$, also written $\sfW(\bfr,t)$, is defined to be the set of spacetime points occupied by this anyon; formally, $(\bfr,t) \in \sfW(a)$ and  $(\bfr',t')\in \sfW(a)$ only if $(\bfr'',t'-1) \in \sfW(a)$ for some $\bfr'' \in B_{\bfr'}^{(\infty)}(1)$.\footnote{As in 1D, ways of resolving this ambiguity when two anyons occupy the same $B_\bfr^{(\infty)}(1)$ at some point in time will not be important to keep track of.}		
	\end{definition}\ms 
	
	As in 1D, no anyon self-interaction occurs provided $v>1$.

	\ss{Other types of interactions}  \label{ss:power_law_construction}
	
	The message passing rules defined above are designed to move each anyon in the direction of its closest neighbor (albeit in a time-delayed way, due to the finiteness of $v$). 
	It is however straightforward to modify these rules to simulate many different types of interactions between anyons.
	
	Suppose we want to consider a feedback rule that moves each anyon $a$ in a way that depends on the positions of the $N$ anyons closest to $a$. This can be done by using $N$ ``flavors'' of messages, with the control variables at each site being packaged as $m_\bfr = \{ m^{\pm a,f}_\bfr\, : \, a \in \{1,\dots,D\},\, f \in \{1,\dots,N\}\}$. The update rules of $\mca$ are then modified so that instead of newer messages overwriting older ones, a newer message of flavor $f$ turns an older message of flavor $f$ into a message of flavor $f +1$ (with a message of flavor $f = N$ simply being deleted). This allows the signals from up to $N$ anyons to simultaneously be stored at each site, at the cost of increasing the number of control variables at each site to ${\rm dim}(\mch_{{\sf cl},\bfr}) =\ct(N \log \log(L))$. Since anyon pairs below threshold are very unlikely to be separated by a distance greater than $O(\log L)$, it suffices to set $N = \ct((\log L)^D)$ for any reasonable interaction scheme (taking $D=2$ then gives a slightly better scaling compared to the $\ct((\log L)^3 \log \log(L))$ required by the ``constant-$c$'' decoder of Ref.~\cite{herold2015cellular}). 
	
	There is however not obviously any benefit to this added flexibility. For example, consider using this scheme to engineer a $1/r^\a$ attractive interaction between anyons. If $\a$ is too small, one can show that the resulting decoder lacks a threshold. If $\a > D$ is large enough, the existence of a threshold can be proven using the techniques of sec.~\ref{sec:offline}, but it is unclear that these decoders offer any benefit over the simpler nearest-anyon rule defined above: numerically, all schemes with power-law interactions were found to produce smaller thresholds than the nearest-anyon rule decoder. 
	
	\section{Clustering and erosion}\label{sec:erosion} 
	
	In this section we develop the two main concepts needed to show that the message passing decoders possess a threshold, as well as to characterize the behavior of $\ploge,\tdece$ below it. In sec.~\ref{ss:sparsity}, we give a procedure for grouping errors in a given noise realization into isolated clusters of different sizes, and use a variant on a result of  Gacs~\cite{gacs2001reliable,ccapuni2021reliable} to show that for $p$-bounded error models, large-sized clusters are extremely rare. In sec.~\ref{ss:erosion}, we show that anyons created by the noise in sufficiently isolated clusters are guaranteed to be annihilated in a way which does not produce a logical error. These two results will be combined together in sec.~\ref{sec:offline} to prove Theorem~\ref{thm:1d_offline}. 
	
	To simplify the discussion, we will assume throughout that $m_{\sf max} = \lceil\log_2(L)\rceil$ (c.f.~\eqref{mvals}), so that messages may be transmitted and received between any two anyons in the system, no matter how distant. The effects of taking $m_{\sf max}$ to be smaller than this value will be discussed in the following section. 
	
	\ss{Clustering and sparsity} \label{ss:sparsity} 
	
	Our proof strategy will be to employ a recursive hierarchical organization of noise realizations into well-separated clusters of different sizes. This general strategy was first developed by Gacs \cite{gacs2001reliable,ccapuni2021reliable}, 
	although the procedure we use for clustering noise realizations will be a bit different from his original approach. 
	A more detailed analysis and proofs of the results to follow are provided in Appendix~\ref{app:sparsity}; in what follows we will simply state the most important definitions and results. 
	
	We first need to define the algorithm by which we perform clustering on different noise realizations $\sfN \subset \L_l$. For this, we need the following definition:\footnote{This approach to organizing points in $\sfN$ differs from that of Gacs, who uses {\it isolation} rather than clustering; see app.~\ref{app:sparsity} for details.} 
	\ms\begin{definition}[clustering]\label{def:clustering}
		Let $\sfN\subset\L_l$ be a noise realization. A $(W,B)${\it  -cluster} $C_{(W,B)}$ of width $W$ and buffer thickness $B$ is a subset of $\sfN$ such that 
		\begin{enumerate}
			\item all points in $C_{(W,B)}$ are contained in a ball $B_\bfr(W/2)$ for some $\bfr$, and 
			\item these points are separated from other points in $\sfN$ by a buffer region of thickness at least $B$:
			\be \sfN \cap (B_\bfr(W/2+B) \setminus  B_\bfr(W/2)) = \emp.\ee 			
		\end{enumerate}
		A point $\bfl \in \sfN$ is called {\it $(W,B)$-clustered} if it is a member of a $(W,B)$ cluster, and a set is said to be $(W,B)$ clustered if all points it contains are $(W,B)$-clustered. 
	\end{definition}\ms
	The following proposition is proved in app.~\ref{app:sparsity}: 
	\ms\begin{proposition}\label{prop:gooddecomp}
		If $B\geq W$, there exists a set of $(W,B)$-clusters $C^{(i)}_{(W,B)}$ such that every $(W,B)$-clustered point belongs to one and only one $C^{(i)}_{(W,B)}$.  
	\end{proposition}
	The clusters defined in all of the applications to follow will have $B>W$. 
	
	Note that we have left unspecified the choice of norm used to define the balls in Def.~\ref{def:clustering}. Clusters defined with both $||\cdot||_1$ and $||\cdot ||_\infty$ will find use in what follows.

	For a given noise realization $\sfN$, consider the pattern of anyons and noise in the state $\r_\sfN = X^\sfN \rlog X^\sfN$. 
	We will say that an anyon at site $\bfr$ with $\s_\bfr = -1$ is in a cluster $C_{(W,B)}\subset\sfN$ if an odd number of links containing $\bfr$ are contained in $C_{(W,B)}$. We then note the following:  \ms\begin{proposition}\label{prop:cluster_anyons}
		For $B>1$, any $(W,B)$-cluster $C_{(W,B)}$ contains an even number of anyons, and if $a$ is an anyon in $C_{(W,B)}$, the anyon strings created by $X^\sfN$ connect $a$ to another anyon in $C_{(W,B)}$. 
	\end{proposition}
	\begin{proof} 
		This follows directly from Def.~\ref{def:clustering}.
	\end{proof} 
	
	This means that if a feedback operator $\mcf_{C_{(W,B)}}$ is applied to $\r_\sfN$ which pairs each anyon in $C_{(W,B)}$ with another anyon in $C_{(W,B)}$, $\mcf_{C_{(W,B)}}$ will not cause a logical error. Intra-cluster pairing of anyons is thus always ``safe'' as far as decoding is concerned, and our goal in the following will be to argue that our message-passing decoder performs only intra-cluster pairing with high probability (provided $p$ is small enough). 
	
	We now define a scheme by which we hierarchically refine $\sfN$ into clusters of exponentially increasing sizes. 
	\ms\begin{definition}[noise hierarchies, error rates]\label{def:noise_hier}
		Define $\sfN_0 = \sfN$, and fix a positive integer $n$ and two numbers $w_0,b_0$ such that $0<w_0<b_0$. Define the level-$k$ width $w_k$ and buffer thickness $b_k$ as 
		\be w_k = w_0 n^k, \qq b_k = b_0 n^k.\ee 
		A $(w_k,b_k)_\infty$-cluster $C_{(w_k,b_k)}$ will be referred to as a {\it $k$-cluster} in what follows, and a particular such cluster will often be written more succinctly as $\mcc_k$. The largest level of cluster in a system of linear size $L$ will be denoted by $k_L$:\footnote{Strictly speaking $k_L$ is the floor function of this quantity, but we will omit these and related floor functions for notational clarity. } 
		\be \label{kldef} k_L = \log_n \frac L{w_0+2b_0}\ee 
		
		The {\it $(k+1)$th level clustered noise set $\sfN_{k+1}$} is defined as the set obtained by deleting from $\sfN_k$ all $(wn^k, bn^k)_\infty$-clustered points.\footnote{This removal process is well-defined on account of \ref{prop:gooddecomp}.}
		
		For a $p$-bounded error model $\mce$, the {\it level-$k$ error rate} $p_k$ is the largest probability for the location of a given qubit to belong to $\sfN_k$:
		\be p_k = \max_{\bfl \in \L_l}\sfP_{\sfN \sim \mce}\(\bfl\in \sfN_k\),\ee 
		with $p_0 = p$.
	\end{definition}\ms 
	Note that these clusters are specifically defined with respect to the $\infty$ norm. 
	
	The key utility of this definition is that, provided $p$ is small enough, the level-$k$ error rate decreases doubly-exponentially in $k$. This result is formulated in the following theorem, which closely parallels the analogous result in Ref.~\cite{ccapuni2021reliable}: 
	
	\ms\begin{theorem}[sparsity bound]\label{thm:sparsity} 
		Suppose the parameters $w_0,b_0,n$ satisfy 
		\be \label{bwineqs} 0 < w_0 < b_0 < \frac{n-3}4 w_0.\ee 
		Then in $D$ dimensions, there exists a constant $p_* \geq 1/(2b_0 + w_0)^D$ such that 
		\be \label{rarenoise} p_k \leq \frac{p_*}{n^{kD}} \(\frac p{p_*}\)^{2^k}.\ee 
	\end{theorem}
	The full proof is provided in app.~\ref{app:sparsity}, and is conceptually rather similar to proofs of fault tolerance in concatenated quantum codes (see e.g. \cite{aharonov1996faulttolerantquantumcomputation,gottesman2009introductionquantumerrorcorrection}).

	\ss{Linear erosion} \label{ss:erosion}
	
	We now turn to showing that the synchronous message passing decoder satisfies what we will refer to as a {\it linear erosion property}: there exists an $O(1)$ constant $c$ such that anyons in any given $(W,cW)$ cluster are guaranteed to all annihilate against one another in a time linear in $W$. A refinement of  property which takes into account effects of inevitable inter-cluster interaction at large scales will be combined with the sparsity results of the previous subsection to prove the threshold theorem in sec.~\ref{sec:offline}. 
	
	We should note that this property is not as strong as the linear erosion property discussed in the cellular automaton literature (see e.g. \cite{gacs2024probabilistic,rakovszky2024defining,gacs_eroder_slides}), which would require erosion in both $\mch_\qu$ and $\mch_\cl$. In our decoders by contrast, erosion only occurs in $\mch_{\sf qu}$. Indeed, the messages produced by anyons in a $(W,cW)$ cluster which propagate away from the cluster are {\it not} guaranteed to be erased by a time proportional to $W$, and thus have the potential of influencing anyons in an unboundedly large spacetime region. Nevertheless, the fact that the erosion time is linear in $\mch_\qu$ will be shown to imply that these errant messages only slow down error correction by a constant factor (if erosion in $\mch_{\sf qu}$ took longer than linear time, this effect would not be so innocuous, and would preclude the existence of a threshold). In passing, we note that the lack of erosion in $\mch_{\sf cl}$ is ultimately the reason why these (and related) decoders fail to be robust online decoders; details are provided in Ref.~\cite{online}.
	
	\sss{Modified message passing} 
	
	To keep the proof of linear erosion simple, we will find it helpful to slightly modify the feedback rules applied by the decoder (proving linear erosion is the only place in this paper where these modified rules will be utilized). Let $\uvx$ be the unit vector pointing along the first spatial coordinate $r^1$.  Our modified feedback rules apply operators $\wt \mcf_\bfr^{(D)}$ which are obtained from the feedback operators $\mcf_\bfr^{(D)}$ of Def.~\ref{def:Dd_msg_passing} by adding motion of anyons along $-\uvx$ to all feedback operators except those which move anyons along the $\pm\uvx$ direction. Formally, this property is captured by the following definition:	
	
	\ms\begin{definition}[modified message passing automaton]\label{def:modified_feedback}
		The modified message passing automaton decoder is defined as the decoder in Def.~\ref{def:Dd_msg_passing}, except with feedback operators $\wt \mcf_\bfr^{(D)}$ defined as 
		\bea \wt\mcf^{(D)}_\bfr & = X_{\lan \bfr-\uvx,\bfr\ran}^{G_\bfr(\bfm,\bfsig)} \mcf^{(D)}_\bfr, \\ 
		G_\bfr(\bfm,\bfsig) & = \frac12\(1-\s_\bfr \s_{\mcf^{(D)}(\bfm,\bfsig),\bfr}\)  \cdot  \bigvee_{a \neq 1}  \(\s_\bfr \wt f_\bfr^a \vee \s_{\bfr+\uva} (-\wt f_{\bfr+\uva}^a) \), \eea
		where the $\wt f_\bfr^a$ are defined in the same way as in \eqref{blahdef}, except with a different way of resolving degeneracies: degeneracies between messages from the $s_x\uvx$ and $s_a\uva$ directions with $a > 1$ and $s_{x,a} \in \{\pm 1\}$ are split to favor motion along $s_a\uva$ if $a \neq 1$, while no anyon motion occurs if $s_x=+1$ (see panel a of fig.~\ref{fig:erosion_setup} for an example in $D=2$).  The function $G_\bfr(\bfm,\bfsig)$ is defined such that it equals 1 when both a) $\mcf^{(D)}$ moves an anyon along a link containing $\bfr$ and normal to $\uvx$, and b) when an anyon remains at site $\bfr$ after the action of $\mcf^{(D)}$. Finally, note that this rule remains local since $\s_{\mcf^{(D)}(\bfm,\bfsig),\bfr}$ only depends on $\{\bfm,\bfs\}|_{B_\bfr^{(\infty)}(1)}$.
	\end{definition}\ms
	
	Note that the $\infty$-norm distance between an anyon and any fixed point still changes by at most $1$ under the modified feedback rules. The modified feedback rules are illustrated schematically in the left panel of fig.~\ref{fig:erosion_setup}.
	
	\begin{figure}
		\centering
		\includegraphics[width=.7\tw]{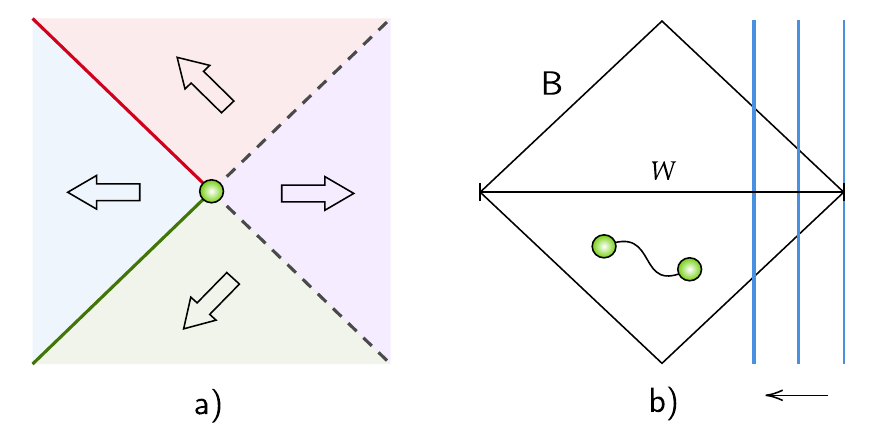}\caption{\label{fig:erosion_setup} {\bf{\sf a)}}  modified feedback rules in 2D. The arrows in each of the four quadrants indicate the direction that an anyon moves upon receiving a signal from the associated direction. When the signals are received  from the $\pm \uvy$ components of the lightfront, the anyon moves along $\pm \uvy$ before moving along $-\uvx$. Degeneracies between different movement directions are resolved using the colors of the shaded lines on the boundaries of the lightfront segments; dashed black lines mean that anyons do not move. {\bf {\sf b)}} strategy of the erosion proof for clusters defined with the 1 norm. Anyons initially in a ball $\sfB$ of radius $W/2$ remain there if they are sufficiently isolated. Fixing the center of the ball at $\bfr=\bfzero$, after time $t_{\sf move} = W/v$, all anyons in $\sfB$ are confined to be to the left of the line $r^1 = W/2-(t-t_{\sf move})$ (blue vertical lines); this leads to all anyons being annihilated after time $t_{\sf move} + W$.  }
	\end{figure}

	We now prove the aforementioned linear erosion property in various different settings. In the entirety of this subsection, we will only consider the synchronous message passing decoder with the modified feedback rules defined above. 
	
	\sss{Erosion of a single cluster: clean messages}
	
	To prove a threshold we will need to use linear erosion for clusters defined with the $\infty$ norm. However, the proof is slightly more natural for clusters defined within the 1-norm. We therefore give the proof for the 1 norm first; the extension to the $\infty$ norm is straightforward and will be given later.

	\ms\begin{lemma}[linear erosion: 1-norm]\label{lemma:linear_erosion}
		Consider a $(W,DW)_1$ cluster of errors in a state initialized with trivial messages $\bfm= \bfzero$, and with all anyons in the cluster contained within a 1-ball $\sfB = B^{(1)}(W/2)$. The anyons initially in $\sfB$ never leave $\sfB$, and annihilate with other anyons initially in $\sfB$ by a time $W(1+1/v)$.
	\end{lemma}
	In the limit $v = \infty$, a straightforward adaptation of the argument below shows that $(W,W)_1$ clusters are eroded in time $W$ with the original (unmodified) feedback rules. The modified rules are introduced to make this form of argument go through in a simple fashion even when anyon motion is retarded by virtue of $v$ being finite. Our proof will strictly overestimate the maximum time needed for the anyons in a $(W,DW)_1$ cluster to annihilate, but only by a constant factor, which we will not invest effort in identifying. 
	\begin{proof}
		For simplicity, we will first prove linear erosion for $(W,\infty)_1$ clusters, so that we work on $\zz^D$ and are given a 1-ball $\sfB$ of radius $W/2$ which contains every anyon in the system. We will use spacetime coordinates where $\sfB = B^{(1)}_\bfzero(W/2)$ is centered at the origin.
		
		We begin with two observations. First, all anyons in the system are confined to remain in $\sfB$ for all $t$. Indeed, if this wasn't true, there would have to be an event at a spacetime point $(\bfr,t)$ where an anyon $a$ first moves from $\p \sfB$ to the exterior $\sfB^c$ of $\sfB$. Let $\bfr'$ be the coordinate of the nearest anyon to $\bfr$ in $\sfA(\plc(\bfr,t))$. If $\bfr \in \sfB_{\sf int} =  \sfB\setminus \p\sfB$, one readily checks that $a$ moves away from $\p \sfB$ into $\sfB_{\sf int}$ at time $t+1$. If $\bfr' \in \p\sfB$, $a$ can either move into $\sfB_{\sf int}$ or stay along $\p \sfB$, but the degeneracy breaking rules in Def.~\ref{def:modified_feedback} can be exhaustively checked to guarantee that $a$ does not move outside $\sfB$. 
		
		Our second observation is that as long as
		\be t\geq t_{\sf move} = W/v,\ee 
		all anyons must their spatial positions (or annihilate with other anyons) at each time step.\footnote{As before, we are omitting floor functions here for notational convenience.} This follows directly from the previous observation and the fact that $\sfB\subset \plc(\bfr,t\geq t_{\sf move})$. 
		
		We now use these observations to prove linear erosion. Let $\mcr(t)$ denote the smallest convex subset of $\zz^D$ which contains every anyon at time $t$ (with $\mcr = \emp$ if no anyons remain). We will show (see the right panel of fig.~\ref{fig:erosion_setup})
		\be \label{shrinking} \mcr(t_{\sf move} + \d ) \subset \sfB \cap\{\bfr \, : r^1 \leq W/2 - \d \}.\ee 
		Since the RHS is empty when $\d  > W/v$, this will show the claim for $(W,\infty)_1$ clusters. 
		
		To show \eqref{shrinking}, we argue by induction. To streamline notation slightly, define $\sfS_k = \sfB\cap \{\bfr \, : \, r^1 = W/2-k\}$. Consider the system at $t=t_{\sf move}$, which will be our base case. We claim that $\sfS_0$---which is just the single point $\bfr = (W/2,0,\dots)$---must not be occupied by an anyon at time $t_{\sf move}+1$. 
		First, if there is an anyon $a$ at this point, then since $t \geq t_{\sf move}$, $a$ must move to a different position at time $t+1$. Since $\sfA_\ct(\plc^{+x}(\bfr,t)) = 0$, the movement rules of the modified decoder mandate that $a$ move to a smaller $r^1$ coordinate, viz. move within $\sfB\cap \{ \bfr \, : \, r^1 \leq W/2 - 1\}$. Further, if there is no anyon already at this point, one cannot move there, as in this case $\sfA(\plc^{+x}(\bfr,t)) = \emp$ for all $\bfr \in \sfS_1$. 
		
		Now consider time $t = t_{\sf move} + \d$, and assume that $\sfS_k$ with $k<\d$ has been anyon-free for $\d-k$ time steps. We claim that $\sfS_\d$ is  necessarily anyon-free at time $t_{\sf move} + \d +1$, which will then show the induction step. Indeed, suppose $\sfS_\d$ is not already anyon-free at time $t_{\sf move} + \d$. We claim that 
		\be \sfA_\ct(\plc^{+x}(\bfr,t) ) = 0,\quad \, \, \forall \, \, (\bfr,t) \in \sfA(\sfS_\d,t).\ee 
		This is true since we have assumed $v > 1$: $\plc^{+x}(\bfr,t)\cap \sfS_k\neq \emp$ at time $t-(\d-k)/v$ (again omitting floor functions), but our inductive assumption presumes that $\sfS_k$ was anyon-free starting at the earlier time $t-(\d-k) < t-(\d-k)/v$. Since $t > t_{\sf move}$, the anyon at $(\bfr,t)$ must move. In particular, since $\sfA_\ct(\plc^{+x}(\bfr,t) ) = 0$ and since our modified feedback rules 
		always move anyons along $-\uvx$ if the signal they receive does not come from $\plc^{+x}$, this anyon must move from $\sfS_\d$ to $\sfS_{\d+1}$. Thus $\sfS_{\d}$ must be anyon-free at time $t_{\sf move}+\d+1$, which is what we wanted to show. 
		
		The above argument proves the claim for $(W,\infty)_1$ clusters. That the argument holds for $(W,DW)_1$ clusters is however also straightforward. Indeed, any anyon $a$ in $\sfB^c$ can decrease its $\infty$ norm distance to $\sfB$ at a given time step by 1 only after time $t_{\sf move}$, but at this time each anyon $b$ in $\sfB$ already has at least one anyon $c$ in $\sfB$ in its past lightfront. Letting these anyons be at positions $\bfr_{a,b,c}$, we have 
		\be ||\bfr_b - \bfr_c||_\infty \leq W < ||\bfr_a - \bfr_b||_\infty,\ee 
		with the factor of $D$ in $(W,DW)_1$ guaranteeing that the second inequality holds with the $\infty$ norm (used since message decay is determined by the $\infty$ norm); one may further verify that the above inequalities hold for all $t \geq t_{\sf move}$. Thus even though the anyons in $\sfB$ have time to receive signals from $\sfB^c$ before they are annihilated, these signals are guaranteed to be overridden by signals from other anyons in $\sfB$, allowing the above erosion argument to proceed unchanged. 		
	\end{proof} 
	
	Extending this statement to clusters defined with the $\infty$ norm is straightforward, but the details are unfortunately a bit ugly: 
	
	\ms\begin{corollary}[linear erosion: $\infty$-norm]\label{lemma:linear_erosion_infty}
		Consider a $(W,3W/2)_\infty$ cluster of errors in a state initialized with trivial messages $\bfm= \bfzero$, with all anyons in the cluster contained with an $\infty$-ball $\sfB = B_\bfr^{(\infty)}(W/2)$. Define the pyramid-shaped regions
		\be \scp^{a}_\bfr(l) = \{ \bfr' \, : \,   r^a \leq (r')^a \leq  r^a + l \, \, \wedge \, \, (r')^a-r^a \geq (r')^b  - r^b\, \, \forall \, \, b \in \{1,\dots,D\}\}.\ee  
		The anyons initially in $\sfB$ never leave the region 
		\be \sfB \cup \scp^{x}_{\bfr-W\uvx}(W/2) \subset B_\bfr^{(\infty)}(W),\ee 
		and annihilate with another anyons initially in $\sfB$ by a time $d_vW$, where $d_v =\frac{3}2(1+1/v)$. 
	\end{corollary}
	Note that in 1D the addition of $\scp^x_{\bfr-W\uvx}$ and the increased erosion time are unnecessary, since the 1-norm and $\infty$-norm are identical. Similarly to the previous Lemma, this result in fact strictly overestimates the time needed for erosion to be completed by an amount proportional to $W$. 
	\begin{proof}
		Consider first the case of a $(W,\infty)_\infty$-cluster.  Let us first show that the anyons initially in $\sfB$ never leave $\sfB\cup\scp^{x}_{\bfr-W\uvx}(W/2)$. By a similar argument as in the proof of Lemma~\ref{lemma:linear_erosion}, one easily sees that due to the modified feedback rules, an anyon at $\bfr\in \p \sfB$ can only move to $\L\setminus\sfB$ if $\bfr$ is on the component of $\p \sfB$ with unit normal along $-\uvx$; the same argument shows that anyons at positions $\bfr \in \p \scp^{x}_{\bfr-W\uvx}(W/2) \setminus \p \sfB$ cannot move to $ \L \setminus \scp^{x}_{\bfr-W\uvx}(W/2)$. The time at which all anyons have received  messages is upper bounded by $t_{\sf move} = \frac{3W}{2v}$ (the largest $\infty$-norm distance between two points in $ \sfB \cup \scp^{x}_{\bfr-W\uvx}(W/2)$), which is a loose upper bound since the initial positions of the anyons are contained within $\sfB$. At $t \geq t_{\sf move}$, the same inductive argument as in the proof of Lemma~\ref{lemma:linear_erosion} shows a direct analog of \eqref{shrinking}, viz. 
		\be \label{shrinking2} \mcr(t_{\sf move} + \d ) \subset \(\sfB\cup \scp^{x}_{\bfr-W\uvx}(W/2)\)  \cap\{\bfr \, : r^1 \leq W/2 - \d \}\ee 
		where $\mcr(t)$ is defined in the same way as previously. The RHS is empty when $\d > 3W/2$, which shows the claim for $(W,\infty)_\infty$ clusters. That this also holds for $(W,3W/2)_\infty$ clusters follows from the fact that the most recent message received  by an anyon in $\mcr(t)$ is then guaranteed to be from another anyon in $\mcr(t)$. 			
	\end{proof} 
	
	\sss{Erosion of a single cluster: arbitrary messages} 
	
	The above statements about erosion were made for initial states with trivial messages. For arbitrary initial message patterns, we have the following statement: 
	\ms\begin{corollary}[linear erosion with arbitrary messages: $\infty$-norm]\label{cor:infarb}
		Define $f_v = \frac{v+1}{v-1}$. Consider a $(W,f_vW)_\infty$-cluster of errors, with all anyons contained in an $\infty$-ball $\sfB = \sfB^{(\infty)}_\bfr(W/2)$ for some site $\bfr$, and let the initial pattern of messages in the system be arbitrary. Then if $v>2$, the anyons initially in $\sfB$ never leave the region
		\be \sfB^{(\infty)}_\bfr(f_vW/2) \cup \scp^x_{\bfr - f_vW\uvx}(f_vW/2),\ee 
		and annihilate with other anyons initially in $\sfB$ by a  time $c_vW$, where $c_v = (3v+5)/(2(v-1))$.
	\end{corollary}
	\begin{figure}
		\centering
		\includegraphics[width=.6\tw]{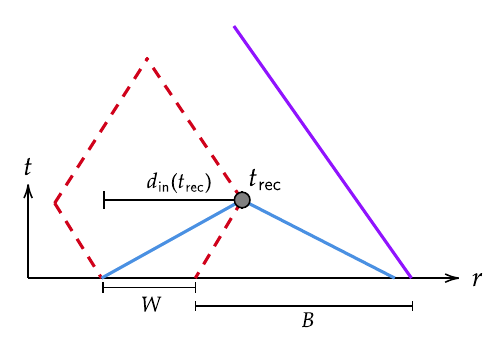} \caption{\label{fig:arb_msg_proof} Geometric considerations relevant to the proof of Corollary~\ref{cor:infarb}, illustrated in the case of a $(W,B)$ cluster in 1D. Blue lines indicate spacetime trajectories of messages, and the anyons in the cluster at $t=0$ are contained in the indicated interval of width $W$. For a fixed time $t$, the distance between the two red dashed lines equals $2R(t)$. While $R(t)$ can increase at small $t$ due to the arbitrary pattern of messages in the initial state, $R(t)$ must decrease after the time $\trec$ when the rightmost anyon in the cluster begins to receive signals from the leftmost anyon, provided that $v>2$ and that $B$ is large enough. The minimal value of $B$ is determined by requiring that signals emitted by anyons not in the cluster do not reach the rightmost anyon by time $\trec$. In the figure, the purple line denotes the worldline of the leftmost anyon to the right of the cluster, and the above requirement is satisfied by virtue of the purple line intersecting the line $t=0$ to the right of the place where the right blue line intersects. In higher dimensions, the linear dimensions of the spacetime rectangle in which worldlines of anyons in the cluster are confined is increased due to the $\scp^x$ regions appearing in Corollary~\ref{lemma:linear_erosion_infty}.}
	\end{figure}
	To ensure that this result always holds, we will restrict to $v\geq 3$ in what follows. The basic idea of the proof is that while the arbitrary initial message pattern can cause the support of the cluster to expand at small times, the fast ($v\geq3$) speed of the messages, and the fact that newer messages overwrite older ones, means that this initial expansion is guaranteed to be inexorably reversed at late enough times. The geometric considerations relevant to the proof are illustrated in fig.~\ref{fig:arb_msg_proof} for the case of $D=1$ (where the union with $\scp^x_{\bfr-f_v W\uvx}(f_vW/2)$ is unnecessary and erosion consequently faster). As usual, the result proved here will overestimate the time needed for erosion by an amount proportional to $W$. 
	\begin{proof}
		Consider a $(W,B)_\infty$ cluster of errors  in a state with arbitrary initial message distribution, and let $R(t)$ be the radius of the $\infty$-ball in which all anyons initially in $\sfB$ are contained in at time $t$ (so that $R(0) = W/2$). The longest erosion time will be realized when the initial message pattern causes $R(t)$ to linearly increase at small $t$. Let $\trec$ be the latest time at which an anyon initially in $\sfB$ receives a signal from another anyon initially in $\sfB$. Assuming the worst-case scenario where $dR/dt=2$ for $t \leq \trec$, a simple geometric argument shows $\frac{W + \trec}v \geq \trec$ (see fig.~\ref{fig:arb_msg_proof}), giving 
		\be \trec \leq \frac W{v-1}.\ee 
		
		At time $t$, the maximum possible $\infty$-norm distance $d_{\sf in}(t)$ of signals received  by anyons in $B^{(\infty)}_\bfr(R(t))$ from other anyons in $B^{(\infty)}_\bfr(R(t))$ satisfies (again, see fig.~\ref{fig:arb_msg_proof})
		\be d_{\sf in }(t) \leq W + \trec + 2(t-\trec) .\ee 
		If this distance is smaller than the minimal distance $d_{\sf out}(t)$ of messages received  by anyons in $B_\bfr^{(\infty)}(R(t))$ from those in $\L \setminus B_\bfr^{(\infty)}(R(t))$ as well as from messages in the initial state, the latter types of messages will not influence the erosion of the anyons in $B_\bfr^{(\infty)}(R(t))$. Note that we are guaranteed to have $d_{\sf out}(t) = tv$ if the size of the buffer region satisfies $B - t \geq tv$. Assuming this is true, we are guaranteed to have $d_{\sf in}(t) < d_{\sf out}(t)$ at $t= \trev + 1$ provided $v > 2$ (since $d_{\sf in}(t)$ increases by at most $2$ each time step, while $d_{\sf out}(t)$ increases by $v$, and $d_{\sf out}(\trec) \geq d_{\sf in}(\trec)$). 
		Self-consistency then requires $B - \trec \geq W + \trec$, which is satisfied if $B \geq \frac{v+1}{v-1}W$. 
		
		At times $t > \trec$, the anyons now contained in $\sfB^{(\infty)}(W/2 + \trec)$ all have worldlines of anyons initially in $\sfB$ as the closest worldlines in their past lightfronts, and they thus annihilate with each other for the same reason as in the proof of Corollary~\ref{lemma:linear_erosion_infty}. From this starting point, annihilation takes a time no longer than $3(W/2 + \trec)$, and so all anyons initially in $\sfB$ are eliminated by a time no longer than 
		\be \frac{3W}2 + 4\trec \leq \frac{3v+5}{2(v-1)},\ee 
		showing the claim. 
	\end{proof} 
	
	\sss{Erosion of multiple clusters}\label{sss:const_slowdown}
	
	We now generalize Corollary~\ref{cor:infarb} by considering the simultaneous erosion of all clusters in an input state $\r \sim\mce$, and show that message-mediated interactions between clusters only slow down erosion by a constant factor. 
	
	In what follows, a {\it spacetime rectangle of width $W$ and height $T$ based at $(\bfr,t)$} will refer to the set of spacetime points $\{(\bfr',t')  \, : \, \bfr' \in B_\bfr^{(\infty)}(W/2),\,\, t' \in [t,t+T]\}$.

	\ms\begin{lemma}[message exchange slows down cluster erosion by a constant factor]\label{lem:const_slowdown}
		Consider a state initialized with trivial messages $\bfm = \bfzero$. For $v>2$, there exists choices of parameters $w_0,b_0,n$ such that for all $k\leq k_L$, worldlines of anyons that originate in a $k$-cluster $\mcc_k$ centered at position $\bfr_{\mcc_k}$ are contained in a spacetime rectangle $\wt\mcr_{\mcc_k}$ based at $(\bfr_{\mcc_k},0)$ whose height $T_k$ and width $W_k$ are upper bounded as 
		\be \label{twcond} T_k \leq Aw_k,\qq W_k \leq B w_k \ee 
		for $O(1)$ constants $A,B$ which may chosen such that if $\mcc_k,\mcc'_k$ are two distinct $k$-clusters, then 
		\be \label{disjoint} \wt\mcr_{\mcc_k} \cap \wt\mcr_{\mcc'_k} = \emp.\ee   
	\end{lemma}
	
	Since this holds for all $k$, this implies that anyons in a $k$-cluster $\mcc_k$ are guaranteed to annihilate only against other anyons in $\mcc_k$. In the following subsection, this will enable the use of a percolation argument to bound $\plog, \tdec$. The proof below will at places give rather suboptimal bounds on $A,B$, especially in one dimension (where erosion is quicker); this is done with the aim of streamlining the presentation, as in the proofs of the previous results. The various spacetime regions used in the proof are illustrated in fig.~\ref{fig:1d_spacetime} in the case of $D=1$. 
	
	\begin{proof}
		For a $k$-cluster $\mcc_k$, let $\mcr_{\mcc_k}$ denote a spacetime rectangle of width $2w_k$ and height $d_vw_k$ based at $(\bfr_{\mcc_k},0)$, where $\bfr_{\mcc_k}$ is the spatial center of $\mcc_k$ and $d_v$ as in the statement of Lemma~\ref{lemma:linear_erosion_infty}. In the absence of message exchange between clusters, the worldlines of anyons in $\mcc_k$ would be confined to $\mcr_{\mcc_k}$, as we have $\mcr_{\mcc_k} \cap \mcr_{\mcc'_k} = \emp$ for distinct $k$-clusters $\mcc_k,\mcc'_k$ on account $w_0 < b_0$ (c.f. \eqref{bwineqs}), so that the spatial expansion of each cluster into the pyramids $\scp^{+x}$ (by a distance at most $w_0/2$) is not sufficient to connect the spacetime supports of distinct $\mcc_k,\mcc_{k'}$.\footnote{This is one of the places where our bounds on $A,B$ can be improved in 1D.} We will now inductively show that message exchange between clusters only increases the height and width of the $\mcr_{\mcc_k}$ by a factor proportional to $w_k$, and that the constants $w_0,b_0,n$ in the sparsity theorem can be chosen to guarantee that the resulting increased-size spacetime rectangles for different $k$ clusters remain disjoint.

		The base case for $k=0$ follows directly from the linear erosion property and our assumption that $w_0 < b_0$; this gives $T_0 \leq d_vw_0$ and $W_0 \leq  2w_0$, which is consistent with \eqref{twcond} if 
		\be \label{areq1} A \geq d_v, \qq B \geq 2\ee 
		where as before $d_v =\frac32(1+1/v)$. 
		The erosion analysis guarantees that $\wt\mcr_{\mcc_0} \cap\wt \mcr_{\mcc_0'} = \emp$, and that anyons originating from $\mcc_0$ never annihilate with anyons from any other 0-cluster.

		\begin{figure}
			\centering 
			\includegraphics[width=\tw]{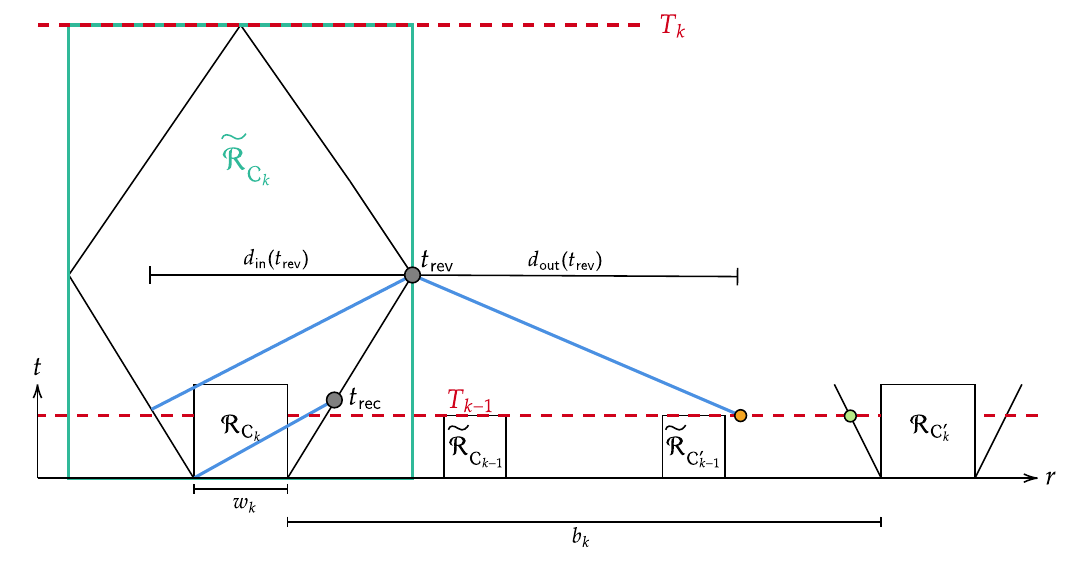}
			\caption{\label{fig:1d_spacetime} An illustration of the causal relationships between different regions of spacetime appearing during the proof of Lemma~\ref{lem:const_slowdown} in one dimension. Blue lines indicate signals. Anyons in $k$-clusters would be confined to the spacetime boxes marked $\mcr_{\mcc_k}$ in the absence of message-induced interactions between clusters; the Lemma shows that in the presence of interactions they are instead confined to larger spacetime rectangles $\wt \mcr_{\mcc_k}$ of height $T_k$, with $\wt \mcr_{\mcc_k} \cap \wt \mcr_{\mcc_k'} = \emp$ for distinct $k$-clusters $\mcc_k, \mcc_{k'}$.  $t_{\sf rec}$ is an upper bound on the latest time at which the rightmost anyon in a $k$-cluster $\mcc_k$ can receive right-moving signals from another anyon in $\mcc_k$, and $t_{\sf rev}$ is an upper bound on the latest time at which the rightmost anyon initially in $\mcc_k$ can move right. The parameters $w_0,b_0,n$ are chosen such that at time $t_{\sf rev}$, the rightmost anyon originally in $\mcc_k$ has to begin moving left before receiving signals from a distinct $k$ cluster $\mcc'_k$; this necessitates that the orange dot in the figure lie to the left of the green dot. In higher dimensions, the linear dimensions of the spacetime rectangles $\mcr_{\mcc_k},\wt\mcr_{\mcc_k}$  are larger due to the extra $\scp^x$ regions occupied by anyons during cluster erosion, c.f. Corollary~\ref{lemma:linear_erosion_infty}. }
		\end{figure}
		
		We now assume that for all $k'<k$, $T_{k'} \leq A w_{k'}, W_{k'} \leq B w_{k'}$, that $\wt\mcr_{\mcc_{k'}} \cap \wt \mcr_{\mcc_{k'}'} = \emp$ for distinct $k'$-clusters $\mcc_{k'},\mcc'_{k'}$, and that anyons in $\mcc_{k'}$ only annihilate with other anyons in $\mcc_{k'}$. Our aim is to show that there exist $O(1)$ choices of $A,B$ such that the induction step holds, provided $w_0,b_0,n$ are chosen appropriately.

		To upper bound $T_k,W_k$, we consider the longest possible trajectory by which anyons in a given $k$-cluster $\mcc_k$ are annihilated, using a strategy similar to the proof in Corollary~\ref{cor:infarb}. Consider this worst-case trajectory, and let $R(t)$ be the radius of the smallest $\infty$ ball centered at $\bfr_{\mcc_k}$ which at time $t$ contains all anyons that originated in $\mcc_k$. The inductive hypothesis implies that anyons in $B_{\bfr_{\mcc_k}}^{(\infty)}(R(t))$ cannot annihilate against any anyons from $(k'<k)$-clusters, which means that such anyons move change their positions by distance at most 1 (w.r.t the $\infty$ norm) at each time step until they annihilate against an anyon from a cluster $\mcc_{k''}$ with $k''\geq k$ (which we will show must be $\mcc_k$ itself). 
		
		In the following we will fix coordinates where $\bfr_{\mcc_k} = \bfzero$. Define $\trec$ and $d_{\sf in}(t)$ as in the proof of Corollary~\ref{cor:infarb}, with $\trec < w_k/(v-1)$\footnote{It is $<$ here and not $\leq$ since anyons in $\mcc_k$ cannot actually begin to move into $\L\setminus B_{\bfzero}^{(\infty)}(W/2)$ until a time of at least $b_0/v$ (as we initialize the system with trivial messages).} and $d_{\sf in}(t) < w_k+ \trec + 2(t-\trec)$. Define the regions 
		\be \scb_{\mcc_k}(t) = B_\bfzero^{(\infty)}(R(t)) ,\qq \sca_{\mcc_k}(t) = \scb_{\mcc_k}(t) \cup \scp^{x}_{-2R(t)\uvx}(R(t)).\ee 
		An anyon in $\scb_{\mcc_k}(t)$ can move into $\L \setminus \sca_{\mcc_k}(t)$ only if it receives signals from anyons in $\L \setminus \scb_{\mcc_k}(t)$ at distances less than $d_{\sf in}(t)$, and our inductive hypothesis implies that no such signals are received  past a certain time. Indeed, once $t > T_{k-1}$, all anyons in $\mcc_{k'<k}$ clusters have been eliminated, and so letting $d_{\sf out}(t)$ denote the minimal distance of signals received  by anyons in $\scb_{\mcc_k}(t)$ from those outside, we have 
		\be d_{\sf out}(t) > v(t - T_{k-1}),\ee
		provided that the past lightfronts of anyons in $\p\scb_{\mcc_k}(t)$ intersect the line $t = T_{k-1}$ before they intersect worldlines of anyons in a nearby $k$-cluster $\mcc_{k}'$. This situation is guaranteed to be avoided if (see fig.~\ref{fig:1d_spacetime})
		\be \label{next_cluster_constraint} w_k/2 + t + v(t-T_{k-1}) < b_k + w_k/2 - T_{k-1}.\ee 
		Some algebra then shows that if the above inequality holds, $d_{\sf out}(t) > d_{\sf in}(t)$ provided $t> t_{\sf rev}$, where
		\be t_{\sf rev} < \frac{w_k}{v-1} + \frac v{v-2} T_{k-1}.\ee  
		Since the closest worldlines in the past lightfronts of anyons in $\scb_{\mcc_k}(t)$ belong to anyons which originated in $\mcc_k$ when $d_{\sf out}(t) > d_{\sf in}(t)$, the satisfaction of \eqref{next_cluster_constraint} at $t=t_{\sf rev}$ will guarantee that anyons from $\mcc_k$ will not annihilate with anyons from any distinct $k$-cluster $\mcc_k'$, showing that $\wt \mcr_{\mcc_k} \cap \wt \mcr_{\mcc'_k} = \emp$ as desired. Using the above bound for $t_{\sf rev}$, this happens if 
		\be \frac{w_k}{v-1} + \frac v{v-2}T_{k-1} < \frac{b_k+ (v-1)T_{k-1} }{v+1}. \ee
		Using $w_k = n^k w_0, b_k = n^kb_0$, some algebra yields the constraint 
		\be \label{areq2} A \leq  n\frac{v-2}{4v-2} \(\frac{b_0}{w_0} - \frac{v+1}{v-1} \).\ee 
		The same argument as in Corollary~\ref{cor:infarb} then shows that the anyons which originated in $\mcc_k$ do not move outside of \be B^{(\infty)}_\bfzero(w_k/2 + \trev) \cup \scp^x_{-(w_k+2\trev)\uvx}(w_k/2+\trev),\ee  
		and are eliminated by time $\trev + 3(w_k/2 + \trev)$. Therefore\footnote{This is another place where we use a rather large overestimate of $W_k$.}
		\bea T_k  & <  \frac32 w_k + 4\trev <\( \frac32 + \frac4{v-1}  + \frac{4v}{(v-2)n} A\)w_0n^k , \eea 
		and similarly for $W_k$. 
		The induction step then holds provided  
		\be \label{areq3} A \geq \frac{\frac32 + \frac4{v-1}}{1 - \frac{4v}{n(v-2)}}\ee 
		and 
		\be \label{breq3} B \geq 2 + \frac4{v-1} + \frac{4v}{(v-2)n} A.\ee 
		The constraints that must be satisfied by $A$ are \eqref{areq1}, \eqref{areq2}, and \eqref{areq3}, although the lower bound in \eqref{areq3} is stronger than that in \eqref{areq1}, so we may focus only on \eqref{areq2} and \eqref{areq3}. We are then tasked with identifying a set of $w_0,b_0,n$ which enable $A$ to satisfy these constraints, while at the same time respecting \eqref{bwineqs}. Once this is done, we may then choose $B$ according to the RHS of \eqref{breq3}. 
		
		By inspecting the equations, it is clear that  \eqref{areq2}, \eqref{areq3}, and \eqref{bwineqs} can all be satisfied provided that we choose an appropriately large $O(1)$ value of $n$. One particular choice which works is 
		\be\label{choices} v=10, \qq w_0 = 5, \qq b_0 = 12,\qq n = 13,\qq A = 3.2,\ee 
		with smaller values of $n$ at fixed $v$ being attainable in 1D, where our bounds on erosion times can easily be made more stringent. 
	\end{proof}

	\section{Threshold theorems and decoding times} \label{sec:offline}
	
	In this section we prove that the local message-passing decoder of Def.~\ref{def:Dd_msg_passing} possesses a threshold under any $p$-bounded noise model $\mce$, and provide general bounds on $\ploge,\tdece$.
	Similar proofs can be given for various types of modified interaction rules achieved through the construction of sec.~\ref{ss:power_law_construction}.\footnote{As an example, attractive power-law interactions of the type considered in Ref.~\cite{herold2015cellular} can be shown to have thresholds as long the interactions are weak enough, decaying with a power $\a>D$.} 
	
	This section is structured as follows. sec.~\ref{ss:ptdef} takes care of some preliminary remarks about the definitions of $\ploge$ and $\tdece$. 
	In sec.~\ref{ss:offline_threshold}, we use the linear erosion results of sec.~\ref{ss:erosion} and a type of percolation argument to prove the existence of a threshold, with the sparsity result of sec.~\ref{ss:sparsity} showing that $\ploge,\tdece$ scale as in Theorem~\ref{thm:1d_offline}.
	Finally, in sec.~\ref{ss:general_comments} we provide some additional results on decoding times. In particular, we give explicit constructions of $p$-bounded error models with $\ploge = p^{\O(L^{1/\b})}$ and $\tdece = \O((\log L)^{\eta})$ with $\b,\eta>1$. These error models however have non-local correlations in space, and we conjecture that if $\mce$ has local correlations (and if we are allowed to store at least $2D\log_2(L)$ control bits per site), the decoding time improves to $\tdece = \ct(\log L)$ (which we will confirm numerically in sec.~\ref{sec:numerics} for i.i.d noise). 
	
	We continue to specialize to synchronous decoding for the entirety of this section. Analogous results for Lindbladian decoders will be given in sec.~\ref{sec:lind}.

	\ss{Preliminaries} \label{ss:ptdef} 
	
	We begin by giving more rigorous definitions of $\ploge$, $\tdece$, and the notion of a threshold. We will find it convenient to use the shorthand 
	\be \mcd^t(\r)_{\sf qu} = \Tr_{\mch_{\sf cl}}[\mcd^t(\r)] \ee 
	to denote the state obtained after evolving under $\mcd$ for $t$ time steps and then throwing away the classical control registers. 
	
	\ms\begin{definition}[logical failure rate, expected decoding time, and thresholds]\label{def:plogandtdec}
		For an error model $\mce$, the {\it logical failure rate} $\plog^{(\mce)}$ is the worst-case trace distance between an input logical state and the long-time output of $\mcd$ running on $\mce(\r_{\sf log})$:  
		\be \plog^{(\mce)} = \frac12\max_{\r_{\sf log}} \lim_{t \ra \infty}||  \mcd^t(\mce(\r_{\sf log}))_{\sf qu} - \r_{\sf log}||_1.\ee 
		When $\mce$ is a stochastic noise channel, this may be rewritten
		more transparently as 
		\be\plog^{(\mce)} =  \max_{\r_{\sf log}}\sfP_{\r \sim \mce,\r_{\sf log}}\( \lim_{t \ra\infty} \mcd^t(\r)_{\sf qu} \neq \r_{\sf log} \)\ee 
		where the notation $\r\sim \mce,\r_{\sf log}$ was defined near \eqref{final_channel}. $\mcd$ will be said to have a {\it threshold at $p_c$} if
		\be \lim_{L \ra \infty} \ploge(p,L) = 0 \, \, \forall \, \, p < p_c.\ee 
		
		The {\it decoding time} $\tdec^{(\mce)}$ is the average-case time complexity of the decoding process, viz. the expected number of times $\mcd$ must act on an input state with definite syndrome sampled from $\mce(\r_{\sf log})$ in order to obtain an anyon-free output:\footnote{Since the feedback applied by $\mcd$ depends only on the syndrome locations in its input state, this quantity is independent of the choice of $\r_{\sf log}$.}
		\be \tdec^{(\mce)} = \sum_{\bfsig} \Tr[\Pi_\bfsig \mce(\r_{\sf log})] \min\{ t\, : \, \Tr_{\mch_{\sf qu}}[\mcd^t(\Pi_\bfsig \mce(\r_{\sf log})\Pi_{\bfsig})_{\sf qu}\Pi_\unit] = 1 \} .\ee 
		We may similarly rewrite this as 
		\be \tdec^{(\mce)} = \EE_{\r \sim \mce} \tdec(\r),\ee 
		where 
		\be \label{tdecrhodef} \tdec(\r) = \min\{ t\, : \, \Tr[\mcd^t(\r)_{\sf qu} \Pi_\unit] = 1 \}.\ee
	\end{definition}\ms

	The discussion in sec.~\ref{ss:noise_models} means that we may without loss of generality take $\rlog$ to be the $+1$ eigenstate of each logical $Z$ operator, and (as before) take $\mce$ to only apply Kraus operators built from tensor products of $X$ operators; these choices will be made in what follows, and we will consequently omit the $\rlog$ in $\r \sim \mce,\rlog$.

	Recall also that the noise channels under consideration can always be taken to have the form $\mce(\r) = \sum_{\sfN \subset \L_l} p^{(\mce)}(\sfN) X^\sfN\r X^\sfN$. In what follows, we will refer to a particular subset $\sfN \subset \L_l$ as a {\it noise realization}, and will use the notation $\sfN \sim \mce$ to denote a random subset of $\L_l$ chosen with probability $p^{(\mce)}(\sfN)$.

	\ss{Bounds on logical error rates and decoding times} \label{ss:offline_threshold}

	We now use a percolation-type argument to combine the linear cluster erosion property of Lemma~\ref{lem:const_slowdown} and the sparsity theorem to prove the existence of a threshold and provide bounds on $\ploge,\tdece$. The heuristic argument is as follows. 
	Linear cluster erosion means that under the error correcting dynamics, small clusters  fail to link up and form a percolating cluster of anyons which can cause a logical error. As a result, logical errors can arise only from large clusters of size $L$, which the sparsity theorem guarantees are extremely rare for sufficiently small $p$. This shows the existence of a threshold error strength $p_c$, below which $\plog^{(\mce)}$ is suppressed super-polynomially in $L$. It also guarantees that the decoding time is set simply by the size of the largest cluster, which the sparsity theorem shows is only polylogarithmic in $L$ with high probability, giving $\tdec^{(\mce)} = O({\rm polylog} (L))$. 
	
	To keep the discussion simple, we will first derive bounds on $\ploge,\tdece$ in the setting where the decoder has access to at least $2D\lceil\log_2(L)\rceil$ classical control bits at each site; this allows us to set $m_{\sf max} = L$ (c.f.~\eqref{mvals}) and directly use the erosion results derived in the previous section. Bounds relevant for the setting where less bits are stored at each site will be discussed later in this subsection. 
	
	The precise statement we prove is the following: 
	
	\ms\begin{theorem}[threshold scaling and expected decoding time]\label{thm:precise_offline}
		Consider the local message passing decoder on a $D$-dimensional torus of linear size $L$ subjected to any $p$-bounded noise model $\mce$. If the decoder has access to at least $\lceil2D \log_2L\rceil$ control bits on each site, then for any update velocity $v>2$, there exists $O(1)$ constants $p_c,\b,\eta$ such that when $p < p_c$, 
		\be \label{plogthm} \ploge = (p/p_c)^{\O(L^{\b})}\ee 
		and 
		\be \label{tdecthm} \tdece = O((\log L)^\eta).\ee 
	\end{theorem}

	We have not endeavored to provide sharp bounds on the allowed values of $p_c,\b,\eta$. Our crude lower bounds on $p_c,\b$ monotonically increase with $v$, and our upper bound on $\eta$ monotonically decreases with $v$. 
	At the very least, for $v \geq 10$ we have 
	\be p_c > \frac1{29^{D}},\qq \b > \log_{13}(2) \approx 0.27, \qq \eta < \log_2(13) \approx 3.7,\ee 
	which are almost certainly extremely conservative bounds. 
	In sec.~\ref{ss:general_comments} we will show that for any $D$, there exist $p$-bounded noise models with $\b =\log_6(5)\approx 0.9, \eta = 1/\log_6(5)\approx 1.1$, and will motivate a conjecture that $ \eta = 1$ for locally correlated noise. 
	
	\sss{Bounding $\ploge$} 
	We first show the result \eqref{plogthm} about the scaling of $\ploge$. 
	
	\begin{proof} 
		We will assume throughout that $p < p_*$, where as before $p_*$ is the critical error strength appearing in the sparsity theorem (which sets a lower bound on the threshold error strength $p_c$). 
		
		Since the decoder has access to at least $\lceil 2D \log_2L\rceil$ bits on each site, we may set $m_{\sf max} = L$. 
		Since Lemma~\ref{lem:const_slowdown} showed that 
		anyons in distinct $k$-clusters are cleaned up within disjoint spacetime boxes, and since the intra-cluster pairing of anyons cannot cause a logical error, the logical error rate is upper bounded by the likelihood of the noise possessing a nontrivial $k_L$-cluster (c.f. \eqref{kldef}). Thus by a union bound and Theorem~\ref{thm:sparsity},
		\be \label{nonemptybound} \ploge \leq  \sfP(\sfN_{k_L} \neq \emp) \leq  \frac{p_*L^D}{ n^{Dk_L}} \( \frac p{p_*} \)^{2^{k_L}},\ee 
		and so 
		\be \ploge = \(\frac p{p_*}\)^{\O(L^{\log_n(2)})}\ee 
		with \eqref{choices} then giving $\b \geq \log_{13}(2)$ for $v \geq 10$.  
	\end{proof}
	
	\sss{Bounding $\tdece$} 
	
	We now show the bound on $\tdece$ given in \eqref{tdecthm}. 
	
	\begin{proof}
		Again, assume $p < p_*$. 
		Recall the definition of $T_k$ from Lemma~\ref{lem:const_slowdown}. In order for $\tdec > T_{k-1}$ for a particular noise realization, it is necessary that $\sfN_{k} \neq \emp$, so 
		\be p(\tdec > T_{k-1}) \leq \sfP(\sfN_{k} \neq \emp).\ee 	
		We then calculate 
		\bea \EE[\tdec] & = \sum_{t=0}^\infty \sfP(\tdec > t) \\
		& = \sum_{k=0}^{k_L} \sum_{t=T_{k-1}+1}^{T_k} \sfP(\tdec > t) \\ 
		& \leq \sum_{k=0}^{k_L} \sum_{t=T_{k-1}+1}^{T_k} \sfP(\tdec > T_{k-1}) \\ 
		& \leq \sum_{k=0}^{k_L} T_k \sfP(\sfN_k \neq \emp).\eea 
		where in the first line we used the standard tail sum formula $\EE[x] = \sum_{y=0}^\infty \sfP(x > y)$ and tacitly defined $T_{-1} = 0$. 
		We then split the sum up into terms $k< \ell$ and $k\geq\ell$, where $\ell$ is an integer to be specified shortly. Using \eqref{rarenoise} for $k \geq \ell$, we have 
		\bea \label{secondtdec} \EE[\tdec] & \leq \sum_{k=0}^{\ell-1} T_k +  \sum_{\d=0}^{k_L-\ell} T_{\ell+\d} \sfP(\sfN_{\ell+\d}\neq \emp)
		\\
		& \leq Aw_0 \sum_{k=0}^{\ell-1} n^k + \frac{Aw_0p_* L^D}{n^{(D-1)\ell}} \(\frac p{p_*}\)^{2^\ell} \sum_{\d = 0}^{k_L-\ell} n^{\d(1-D)} \(\frac p{p_*}\)^{2^\d} \\ 
		&	\leq Aw_0 n^\ell \(1 + \frac{p_* L^D k_L}{n^{D\ell }}  \(\frac p{p_*}\)^{2^{\ell}}\) .
		\eea 
		We then choose $\ell$ so that the second term in the parenthesis is negligible compared to the first in the $L\ra \infty$ limit. Writing $\ell = \log_2(\a \log_n(L))$ for some $O(1)$ constant $\a$,\footnote{As before, we will avoid writing explicit floor functions in contexts like this.} the second term is
		\be \frac{p_* L^D k_L}{n^\ell}  \(\frac p{p_*}\)^{2^{D\ell }} = O \((\log_nL)^{1-D\log_2(n)} L^{D-\a \log_n(p_*/p)}\),\ee 
		which is $o(L^0)$ if $\a > D/\log_n(p_*/p)$. Choosing e.g. $\a = 2D/\log_n(p_*/p)$, the first term in \eqref{secondtdec} then gives 
		\be \label{tdecfinal} \EE[\tdec] = O( (\log L)^{\log_2(n)}).\ee 
		The particular choice of parameters \eqref{choices} then gives $\eta \leq \log_2(13)$ for $v \geq 10$. 
	\end{proof}
	
	We now discuss how further restricting the number of classical bits available at each site changes these bounds. If the number of bits $n_{\sf bits}$ satisfies $n_{\sf bits} < \lceil 2D  \log_2 L \rceil$---so that $m_{\sf max} < L$ (see \eqref{mvals})---a given anyon may be incapable of transmitting messages to every site in the system. In this case, it is thus possible to construct examples of errors which are never eliminated, even at infinite times (giving the possibility of having $\tdece = \infty$). To fix this problem, we may add a small amount of random motion to the feedback operators, by having each anyon move in a random direction once every $u$ time steps, where $u$ is an $O(1)$ constant. As long as $u$ is sufficiently large, the erosion results of sec.~\ref{sec:erosion} go through unchanged for clusters of size less than $m_{\sf max}$. For clusters larger than $m_{\sf max}$ linear erosion is not guaranteed, and their annihilation may be diffusion-limited (when anyon pairs in a cluster separate to scales larger than $m_{\sf max}$). The diffusive behavior nevertheless ensures the eventual elimination of all clusters, and this is sufficient for deriving the following result: 
	
	\ms 
	
	\ms\begin{corollary}[thresholds with limited classical memory]
		Let $p_c,\b,\eta$ be as in Theorem~\ref{thm:precise_offline}, and let $\mce$ be any $p$-bounded stochastic noise model. 
		Suppose the number of control bits available at each site is $ n_{\sf bits} = \lceil 2D \a \log_2(L)\rceil$, where $0 < \a \leq 1$ is an $O(1)$ constant. Then 
		\be \ploge = (p/p_c)^{\O(L^{\a \b})}.\ee 
		Additionally, there is an $O(1)$ constant $c>0$ such that as long as $n_{\sf bits} > \lceil c \log_2 \log_2 L \rceil$, 
		\be \ploge = (p/p_c)^{\O((\log L)^{\l})},\ee 
		where $\l> 1$. 
		Finally, there is an $O(1)$ constant $c'>0$ such that as long as $n_{\sf bits} \geq \lceil c'\log_2\log_2 L\rceil$, the scaling of the decoding time is unchanged from the case where $m_{\sf max} = L$, viz. $\tdece = O((\log L)^\eta)$. 
	\end{corollary}
	
	\begin{proof}
		We will assume throughout that $p<p_*$, and will first show the bounds on $\ploge$. Logical errors may occur as long as the noise contains at least one $k_*$-cluster (assuming without loss of generality that $k_* < k_L$), where
		\be k_* =\lceil \frac{n_{\sf bits}}{2D} \log_n 2 \rceil,\ee 
		so that the cluster's size is not smaller than $m_{\sf max}$. 
		Then 
		\be \ploge \leq \sfP(N_{k_*}\neq \emp) \leq \frac{p_* L^D}{n^{Dk_*}} \( \frac p{p_*}\)^{2^{k_*}} = \begin{dcases} (p/p_c)^{\O(L^{\a \log_n(2)})} & n_{\sf bits} = \lceil 2D\a\log_2L \rceil \\ 
			(p/p_c)^{\O((\log L)^\l)} & n_{\sf bits} = \lceil\frac{2D\l}{\log_n2}\log_2\log_2L \rceil \end{dcases}, \ee 
		where the second line on the RHS requires $\l >1$ (with $2D\l/\log_n2$ at e.g. $n=13$ setting a bound on the constant $c$ above). 
		
		Now we show the statement about $\tdece$; to avoid getting bogged down on this point we will be slightly schematic. The basic fact that we will use is that for pair-annihilating random walkers on $\zz_L^D$, any initial configuration with an even number of particles is annihilated in a diffusion-limited time upper bounded by $CL^{2D}$ for constant $C$ w.h.p.\footnote{For $D>1$ this is a grossly conservative bound, which we use only to simplify notation. } We may then crudely bound $\tdece$ as 
		\bea \tdece & \leq T_{k_*} \sfP(\sfN_{k_*+1}=\emp) + CL^{2D} \sfP(\sfN_{k_*+1} \neq \emp).\eea 
		We will recover the same bound on $\tdece$ as in Theorem~\ref{thm:precise_offline} provided the second term is negligible compared to the first, and $k_*$ is large enough. The bound in \eqref{rarenoise} implies that the second term vanishes as $L\ra\infty$ if 
		\be 2^{k_* } > \frac{3D \ln L}{\ln(p_*/p)},\ee 
		which will be the case below threshold if e.g. $k_* \geq \log_2(4D \log_n(L)) = \ct(\ell)$, where $\ell$ is as in the proof of Theorem~\ref{thm:precise_offline}. Making this choice, one is then led to the same bound on $\tdece$ as in \eqref{tdecfinal}. 
	\end{proof}
	
	In the remainder of this work we will return to assuming $n_{\sf max} = \lceil 2D \log_2 L\rceil $, unless explicitly specified otherwise. 
	
	\ss{Gerrymandering} \label{ss:general_comments}

	Recall the general arguments we gave in sec.\ref{sec:intro}, which suggested the scalings 
	\be \label{naivescaling} \ploge = e^{-\ct(L)},\qq \tdece = \ct(\log L)\ee 
	below threshold for noise channels $\mce$ with local spatial correlations.
	Our threshold proof gave weaker bounds, but in a more general setting (for $\mce$ with non-local correlations). One may naturally wonder how tight the bounds of our threshold result are: could a refinement of our proof be used to improve the guaranteed scaling to \eqref{naivescaling} for {\it all} $p$-bounded $\mce$? 
	
	The following proposition shows that the answer is negative: there are indeed (highly spatially correlated) $p$-bounded error models for which $\ploge$ is only a stretched exponential in $L$, and for which $\tdece$ scales as $(\log L)^{\b}$ for some $\b > 1$. This is possible because the message-passing decoder produces logical errors on certain noise events whose weight is sublinear in $L$, which will be demonstrated shortly by showing the existence of certain sparse fractal-like patterns of errors that are guaranteed to cause logical failures. We will refer to this phenomenon as {\it Gerrymandering}, since the logical failures arise from a dilute collection of strategiclally-placed errors that sway the outcome of the global error correction operation. Gerrymandering is ubiquitous in error correction algorithms based (either explicitly or implicitly) on clustering; see also the discussions in Refs.~\cite{wootton2015simple,hutter2015improved,paletta2025high,balasubramanian2024local}. 
	
	 Gerrymandering implies that proving a scaling of $\tdece$ as in \eqref{naivescaling} for spatially local $\mce$ will necessitate techniques that specifically take advantage of the lack of $\mce$'s spatial correlations.\footnote{Gerrymandered configurations can dominate the scaling of $\ploge$ at very small $p$---even for local $\mce$---but they are usually too entropically disfavored to influence $\tdece$.}
	
	\ms\begin{proposition}[Gerrymandering]
		In any dimension, there exists a $p$-bounded error model $\mce_{\sf bad}$ such that for all communication velocities $v\geq3$, the logical error rate satisfies 
		\be \plog^{(\mce_{\sf bad})} = e^{-O(L^{\log_6(5)})},\ee 
		and a $p$-bounded error model $\mce'_{\sf bad}$ with 
		\be  \tdec^{(\mce'_{\sf bad})} = \O((\log L)^{\frac1{\log_6(5)}}).\ee 
	\end{proposition}
 
	\begin{proof}
		We first show the result about $\ploge$. The error model $\mce_{\sf bad}$ will be chosen to create a dilute chain of anyons along a 1D slice of the lattice which is guaranteed to be mapped to a logical error by the decoder. For this reason, we will focus on the 1D repetition code for simplicity; the generalization to higher dimensions is immediate. 
		
		For positive integers $k\leq n$, define the operator substitution rule 
		\be \scs_{n,k}(\mco) = \mco^{\tp k} \tp \unit^{n-k} ,\qq  \scs_{n,k}(A\tp B) = \scs_{n,k}(A)\tp \scs_{n,k}(B).\ee  
		If $k<n$, one readily verifies that iterating $\scs$ $m$ times on input $\mco=X$ produces an operator whose syndromes are arranged as a fractal patter, with a density of nontrivial Paulis that decreases exponentially with $m$ according to a fractal dimension set by $k/n$: 
		\be |\supp(\scs_{n,k}^m(X))| = n^m e^{-m \ln(k/n)}.\ee 
		For simplicity, consider a system where $L$ is a power of $n$, and let 
		\be \r_{n,k} = \scs_{n,k}^{\log_nL}(X) \r_0 \scs_{n,k}^{\log_nL}(X),\ee  
		with $\r_0 = (|0\ran\lan 0|)^{\tp L}$. We claim that for any update speed $v \geq 3$, there exist choices of $k<n$ such that the decoder outputs $\r_1 = (|1\ran\lan 1|)^{\tp L}$ when run for long enough on input state $\r$.\footnote{For simplicity, we use the  deterministic feedback rule, where an anyon does not move along $\uva$ if the forces it receives are balanced along the $\pm \uva$ directions. To ensure that the logical states are the only steady states of the decoder, this mandates that the system size (and hence $n$) be odd.} 
		
		For simplicity of presentation, we will specify to $k=n-1$ and momentarily take $v = \infty$. We then prove the that the decoder maps $\r_{n,k}$ to $(|1\ran\lan1|)^{\tp L}$ after a time linear in $L$ if $n\geq5$.
		To ease the notation, we will use $\scs$ as shorthand for $\scs_{n,n-1}$. 
		
		Consider the anyons created by applying $\scs^2(X)$ to a large logical state. Letting the beginning of this operator string be applied to the link with left vertex at site 0, the positions of the anyons are $\{ln,(l+1)n-1\, : \, l = 0,\dots,n-2\}$, and all but the first and last anyons are a distance of 1 away from another anyon. After one step of the decoder, only two anyons remain, at positions $\{1, (n-1)n-2\}$. The anyon configuration on the $n^2$ sites on which $\scs^2(X)$ was applied thus contains only a single string of anyons; we will use this to define a modified recursion rule. Define the operator 
		\be \mco_{a,b,c} = \unit^{\tp a} \tp X^{\tp b} \tp \unit^{\tp c}.\ee
		In this notation, $\scs(X) = \mco_{a_1,b_1,c_1}$ with $a_1 = 0, b_1 = n-1, c_1 = 1$. 
		The aforementioned anyon configuration created by applying one step of the decoder to $\scs^2(X)$ is created by $\mco_{a_2,b_2,c_2}$, where $a_2 = 1,b_2=  n(n-1)-2,c_2 = n+1$. To understand what happens to $\scs^3(X)$, we examine what happens to $\scs(\mco_{a_2,b_2,c_2})$, and so on. Continuing this recursion, we see under the action of the decoder, the anyons created by $\scs(\mco_{a_l,b_l,c_l})$ are equal to the anyons created by $\mco_{a_{l+1},b_{l+1},c_{l+1}}$, where\footnote{Omitting floor functions as usual.}
		\bea  \label{recurs} a_{l+1} = a_l + \frac{a_l + c_l}2,\qq b_{l+1} = \max(0,n^{l+1} - a_{l+1} - c_{l+1}),\qq c_{l+1} = c_l + \frac{a_l+c_l}2 + n^l,\eea 
		and that this mapping takes a time $ \frac{a_l+c_l}2$. In order for the decoder to map $\r_{n,n-1}$ to $\r_1$, we require that the anyon strings created by $\mco_{a_l,b_l,c_l}$ constantly grow in size; this mandates 
		\be \label{bl} b_l  > a_l + c_l.\ee 
		Solving the recursion relation with the given initial conditions, we have 
		\be a_l + c_l = \frac{n^l - 2^l}{n-2}.\ee 
		We may then use \eqref{recurs} to check that \eqref{bl} is satisfied if $n\geq 5$. Provided this is the case, the decoding time of this state is 
		\be \tdec(\r_{n,n-1}) = \frac1{2(n-2)}\sum_{l=1}^{\log_n(L)} (n^l - 2^l) = \ct(L)  .\ee 
		
		The extension of the above analysis to finite $v$ is tedious but straightforward, and will be omitted. For the minimum value of $n=5$, one finds that the decoder takes $\r_{n,n-1}$ to $(|1\ran\lan1|)^{\tp L}$ only if $v \geq 5$. For $n=6$ however, any $v\geq3$ is sufficient. 
		In these cases $\tdec(\r_{n,n-1})$ is still clearly $\ct(L)$ since taking $v < \infty$ only slows down the decoding; for $v=3,n=6$ it is numerically observed to approach $\tdec(\r_{n,n-1}) \approx 0.35L$ in the $L\ra\infty$ limit. 
		
		To show the desired result about the logical failure rate, it suffices to consider the channel $\mce_{\sf bad}$ constructed from two Kraus operators $\mck_\unit$ and $\mck_{n}$, defined as 
		\be \mck_\unit = \sqrt{1-p_{\sf eff}}\unit ,\qq \mck_{n} = \sqrt{p_{\sf eff}}\scs^{\log_n(L)}(X),\ee 
		where $p_{\sf eff} = p^{(n-1)^{\log_n L}}$. This is clearly a $p$-bounded noise channel, and the probability of a logical error is simply the relative weight of $\mck_{n}$, giving 
		\be \plog^{(\mce_{\sf bad})} = e^{-\ct(L^{\log_n(n-1)})}. \ee 
		Since we require $v \geq 3$ for a threshold, for which the argument goes through with $n \geq 6$, we may fix $\mce_{\sf bad}$ to have $n = 6$. Therefore 
		\be \plog^{(\mce_{\sf bad})} = e^{-O(L^{\log_6(5)})},\ee 
		which is what we wanted to show. 
		
		A very similar construction can be used to show the result on $\tdece$. For a fixed constant $\b$ to be determined later, let us divide the system into blocks of size $r = (\log L)^\b$; we will use notation in which we pretend $r$ is a power of $n$ (thus suppressing yet more floor functions). Let $\scs^{\log_n r}(X)_i$ denote the operator $\scs^{\log_n r}(X)$ acting on the $i$th block of the system. We take $\mce_{\sf bad}'$ to apply $\scs^{\log_nr}(X)_i$ independently on each block $i$ with probability 
		\be p^{|\supp(\scs^{\log_nr}(X))|} = p^{(n-1)^{\log_n r}}.\ee 
		From the above analysis of $\mce_{\sf bad}$, the time for the decoder to eliminate the anyons created by the noise on a single block is $\ct(r)$ if $n\geq 6$. Therefore for there is an $O(1)$ constant $c$ such that 
		\bea \sfP_{\r \sim \mce'_{\sf bad}}(\tdec(\r) < c r)  &= \( 1 - p^{(n-1)^{\log_n r}}\)^{L/r} \leq e^{- \frac Lr p^{(n-1)^{\log_n r}}}\\ 
		& = \exp\(- (\log L)^{\b \log_n(n-1)} \ln(1/p) + \ln L + \ct(\log \log L)\),\eea 
		since $\mce_{\sf bad}'$ independently creates errors in each block. This probability vanishes as $L \ra \infty$ if $\b > \frac1{\log_n(n-1)}>1$. Since we needed to take $n\geq 6$ to guarantee a long decoding time, setting $n=6$ gives the desired result. 
	\end{proof} 
	
	The channels $\mce_{\sf bad}, \mce'_{\sf bad}$ relied crucially on having long range correlations in space. What can be said about i.i.d noise? The easiest statement is the following trivial lower bound: 
	\ms\begin{proposition}
		For any quantum local automaton decoder in $D$ dimensions operating under $D$-dimensional i.i.d noise, 
		\be \tdec^{(\mce_{i.i.d})}  = \O( (\log L)^{1/D}).\ee 
	\end{proposition}
	\begin{proof} 
		For a noisy state $\r \sim \mce$, let the random variable $r_{\sf iso}(\r)$ denote the radius of the largest ball which contains only a single anyon.
		Then the event $r_{\sf iso}(\r) = r$ implies the existence of a ball of volume $\sim r^D$ in which only one anyon is located; for i.i.d noise the probability for this to happen for a ball centered at a particular site decays like $e^{-cr^D}$ for an $O(1)$ constant $c$. Summing over all sites, this gives constants $a,b$ for which $a (\log L)^{1/D} < r_{\sf iso}(\r) < b (\log L)^{1/D}$ occurs w.h.p. Since a given anyon can move by at most one site at each time step, at least $r_{\sf iso}(\r)/2 \sim (\log L)^{1/D}$ time steps are needed to decode $\r$ w.h.p.
	\end{proof}
	This shows that for $\r \sim \mce_{i.i.d}$, there is a matching of the anyons in $\r$ whose edge lengths are all less than $b (\log L)^{1/D}$ w.h.p. Whether or not there is a perfect matching with the same bound on the largest edge length is a different question---since the largest separation between two anyons paired by the operators created by the noise goes like $\log(L)$ w.h.p, independent of $D$---although the following footnote argues that there is.\footnote{As just mentioned, there always exists a perfect matching in $\r$ which pairs anyons up according to the way in which they were created by the noise, and the longest edge length in this perfect matching goes as $\log(L)$ w.h.p. Given an anyon pair at coordinates $\bfr_1,\bfr_2$ with $||\bfr_1-\bfr_2||\sim \log(L)$, we can however w.h.p find a ``bridge'' of smaller anyon pairs connecting the two anyons in the large pair, in the sense that there exist anyon pairs at coordinates $\bfr_{i,1},\bfr_{i,2}$, $i=1,\dots,n$ with $n \lesssim (\log L)^{1-1/D}$, $||\bfr_1-\bfr_{1,1}||,||\bfr_{i,\a} - \bfr_{i+1,\b}||,||\bfr_{n,2}-\bfr_2|| < b(\log L)^{1/D}$ w.h.p. Flipping the perfect matching along the anyons in this bridge then gives a perfect matching where the edge of length $\sim \log(L)$ has been substituted for one of length at most $\sim (\log L)^{1/D}$. Repeating this process for all $\sim \log(L)$-sized pairs then suggests the existence of a perfect matching with largest edge length $\sim (\log L)^{1/D}$. } 
	Nevertheless, constructing a perfect matching with this bound on the largest edge length appears to require nonlocal information, and we conjecture that $\log(L)$ is the best one can do for i.i.d noise with a local decoder. Since any noise model with quickly decaying spatial correlations should roughly speaking renormalize to i.i.d noise at large enough (but still $O(L^0)$) distances, we conjecture the following: 
	\ms\begin{conjecture}\label{conj:local_noise}
		If $\mce$ is a $p$-bounded error model with local correlations, then for $p<p_c$, 
		\be \tdec^{(\mce)} = O(\log L).\ee 
	\end{conjecture} \ms 
	In sec.~\ref{sec:numerics} we will numerically confirm this conjecture at near-threshold error rates for i.i.d noise in 1 and 2 dimensions. 	
	
	\ss{Extension to arbitrary topological codes} \label{ss:nonab} 
	
	In this short subsection, we show that the message-passing decoder has thresholds for all topological codes $\mcc$, not just the surface code. In particular, acheiving a threshold for non-Abelian anyon theories does not require modifying the message-passing achitecture in any way, regardless of whether or not they are acyclic (although making $\mcc$-dependent modifications may of course increase the numerical value of threshold). 
	
	We first define how our decoder operates for a generic anyon theory $\mcc$.	Following most treatments of non-Abelian decoders in the literature (e.g.~\cite{brell2014thermalization,wootton2014error,dauphinais2017fault,hutter2015improved}), we work in a phenomenological setting in which anyons live on the sites of the lattice, and charge-neutral pairs of anyons are created by the noise according to a $p$-bounded noise distribution on the links of the lattice (see e.g. \cite{schotte2022quantum} for a discussion of how a general qubit-based noise model may be converted to a phenomenological anyon-based one). 
	We assume the ability to make topological charge measurements at each site,\footnote{Performing measurements by braiding only unambiguously resolves the anyon type in modular anyon theories; for non-modular theories, the anyon type is resolved only up to elements of the Muger center of the associated tensor category. Since the anyons in the center cannot implement logical operations on a torus, we will not care about keeping track of them, and they can be ignored as far as decoding is concerned. } and also assume the ability to deterministically move anyons from one lattice site to another through the application of appropriate string operators (which in general require linear-depth circuits to be implimented, but this point is unimportant for us).  
	
	With the ability to measure topological charge and move anyons, we may impliment decoding in precisely the same way as was done for the surface code. In particular, nontrivial anyons always act as message sources, and the messages they send are independent of their topological charge.\footnote{We chose to use separate messages for charges and fluxes in our discussion of the surface code, but this is not actually necessary, and by the arguments in this paragraph, ignoring the distinction between them can be done without compromising the existence of a threshold (although doing so will certainly result in a smaller value of $p_c$).}  The decoder then moves anyons in the direction of the smallest-valued recieved message, and topological charges are re-measured at each time step after motion has occured. 
	The non-Abelian nature of the anyons involved does not compromise the existence of a threshold because, at small enough $p$, the decoder has a very high probability of fusing all anyons within a given noise cluster together. Since the anyons in a given noise cluster were created from the vacuum (by the noise), they are as a whole topologically neutral, and fusing them together (in any order) is guaranteed to produce the vacuum, independent of any mutual braiding events among them which may have occured along the way (for a similar discussion in the context of RG-based decoders, see \cite{hutter2015improved,wootton2016active}). 
	This argument thus implies 
	\ms \begin{corollary}[message-passing thresholds for arbitrary topological codes]\label{cor:arb_thresholds} 
		The statements of theorem~\ref{thm:precise_offline} hold independent of $\mcc$. 
	\end{corollary}

	\section{Desynchronization and Lindbladian dynamics } \label{sec:lind}
	
	The decoders discussed thus far have all operated {\it synchronously:} all sites in the system possess access to a reliable shared clock, which is used to ensure that automaton updates and feedback operations are applied in unison at every site. In this section we discuss a procedure for ``desynchronizing'' any such synchronous decoder, removing the need for a synchronizing global clock. Decoders which have been desynchronized in this way can naturally operate in continuous time, with each site performing local automaton updates and feedback operations at a series of independently chosen random time intervals (distributed according to e.g. independent Poisson processes). 
	For this reason, we will focus desynchronized decoders that perform error correction using local Lindbladians. 
	
	This section is structured as follows. In sec.~\ref{ss:marching soldiers} we formulate the class of Lindbladian decoders considered in this work and describe our desynchronization procedure, which is based on the ``marching soldiers'' scheme of  Ref.~\cite{berman1988investigations} and developed further in \cite{gacs2001reliable,gacs_synch_slides}. In sec.~\ref{ss:slowdowns} we prove that this scheme produces Lindbladian decoders that perform the same computations as their synchronous counterparts, and that the time it takes the two types of decoders to perform a computation is almost identical. Finally, in sec.~\ref{ss:asynch_offline} we use these results to prove that the bounds on $\plog,\tdec$ derived in sec.~\ref{sec:offline} for the message passing decoders continue to hold in the Lindbladian setting.

	\ss{Desynchronization with marching soldiers}\label{ss:marching soldiers}
	
	The continuous-time decoders we study below will all fall within the framework provided by the following definition: 
	\ms\begin{definition}[Lindbladian automaton decoder]\label{def:asynch_ccqs}
		Consider the same setup as in Def.~\ref{def:synch_ccqs}. A {\it  Lindbladian automaton decoder} $\mcd_\lind = (\mca_\lind,\mcf_\lind,\mcc)$ is a continuous-time dynamics determined by a topological stabilizer code $\mcc$ with stabilizers $g_\bfr$, an automaton rule $\mca_\lind$, and a set of unitary feedback operators $\mcf_\lind$ with well-defined syndrome. The dynamics evolves an input state $\r$ on $\mch_{\sf qu}$ by first forming the initial state $\r_0 = \r \tp |\bfm_0,\unit\ran\lan \bfm_0,\unit|_{\sf cl}$, where $\bfm_0$ is a particular initialization of the control variables, and then evolving $\r_0$ under a Lindbladian $\mcl_\asy$ defined by the jump operators 
		\be \label{jumpdef} L_{\bfr \bfm \bfm'\bfsig\bfsig'}= \frac1{\sqrt\mu} \bot_{\bfr'\in B_\bfr(\De)} \( \d_{m_{\bfr'},\mca_\lind(\bfm',\bfsig')[m_{\bfr'}]}  |m_{\bfr'},\s_{\bfr'}\ran\lan m'_{\bfr'},\s'_{\bfr'}| \) \tp \mcf_{\lind,\bfr} (\bfm,\bfsig) \Pi_{\bfsig|_{B_\bfr(\De)}},\ee 
		where $\mu>0$ is a real number controlling the speed of updates, $\De$ is chosen such that $\supp(\mcf_{\lind,\bfr}(\bfm,\bfsig)) \subset B_\bfr(\De)$ for all $\bfm,\bfsig$, and $\Pi_{\bfsig|_{A}} = \bot_{\bfr' \in A} \frac{1+\s_{\bfr'}g_{\bfr'}}2$. These jump operators perform automaton updates, measure stabilizers,\footnote{Here we have written the jump operators as always performing stabilizer measurements on $\mcr_\bfr$ before $\mcf_\bfr$ is applied. This is of course overkill in the context of offline decoding, where the measurements are reliable and only need to be performed once; we write the jump operators in this form only to emphasize that synchronous measurements are unnecessary.} and perform feedback operations in regions that the feedback operators draw syndromes information from, and update the values of the stored syndromes to reflect the implementation of feedback. We will write $\mcd^t(\r)$ for the state $e^{t \mcl_\asy}(\r \tp |\bfm_0 ,\unit\ran\lan \bfm_0,\unit|)$, and $\mcd^t(\r)_{\sf qu}$ for $\Tr_{\mch_{\sf cl}}[\mcd^t(\r)]$. 
		
		A {\it local} Lindbladian automaton decoder is one for which $\De$ is an $O(1)$ constant, and where the actions of $\mca_\lind$ and $\mcf_\lind$\footnote{The different $\mcf_\bfr$ are taken to mutually commute without loss of generality, as in Def.~\ref{def:synch_ccqs}.} at $\bfr$ only depend on $\bfm|_{B_\bfr(\De)}, \bfsig|_{B_\bfr(\De)}$, c.f. \eqref{localcond}. As above, we use the adjective {\it quantum local} if $\mca$ is non-local. 
	\end{definition}\ms 
	
	Since the decoding process is essentially classical after measurements have been performed, evolution under $\mcd_\lind$ can be described by a continuous time Markov process. Formally, this follows by unraveling $e^{t \mcl_\mcd}(\r_0)$ as a stochastic sum over pure-state trajectories $|\psi(t)\ran$, which undergo evolution under the non-Hermitian Hamiltonian $H = -\frac i2 \sum_i L_i^\da L_i$ and are acted on by $L_i$ at random times, where $i=\bfr\bfm\bfm'\bfsig\bfsig'$ is a composite index (see e.g. \cite{breuer2002theory} for background). The present case is especially simple since one may verify that $\sum_i L_i^\da L_i = \mu\inv |\L|\unit$, so that $H = -i\mu\inv|\L|/2$ and the total rate of transitions out of a particular pure state $|\psi(t)\ran$ is state-independent. The rate $\g_i$ at which the jump operator $L_i$ is applied to $|\psi(t)\ran$ is 
	\be \label{rates} \g_i = || L_i |\psi(t)\ran ||^2 = \frac1\mu \prod_{\bfr'\in B_\bfr(\De)} \d_{m_{\bfr'},\mca_\lind(\bfm',\bfsig')[m_{\bfr'}]} \lan \psi(t) | \(|\bfm'\bfsig'\ran\lan \bfm'\bfsig'| \tp \Pi|_{\bfsig|_{B_\bfr(\De)}} \) | \psi(t)\ran .\ee 
	When $|\psi(t)\ran$ is a stabilizer state---which it will become after a jump operator has acted at each $\bfr$ ---this rate is equal to $1/\mu$ for jump operators that update the control variables in the way dictated by $\mca_\lind$, and is zero otherwise. Therefore after unraveling, the dynamics indeed reduces to a classical continuous time Markov process, where
	each jump operator respecting the automaton transition rule is applied according to independent Poisson processes with rate $1/\mu$, with waiting times drawn from an exponential distribution of mean $\mu$ (in what follows we will mostly choose units of time where $\mu=1$).

	\begin{figure}
		\centering \includegraphics[width=.8\tw]{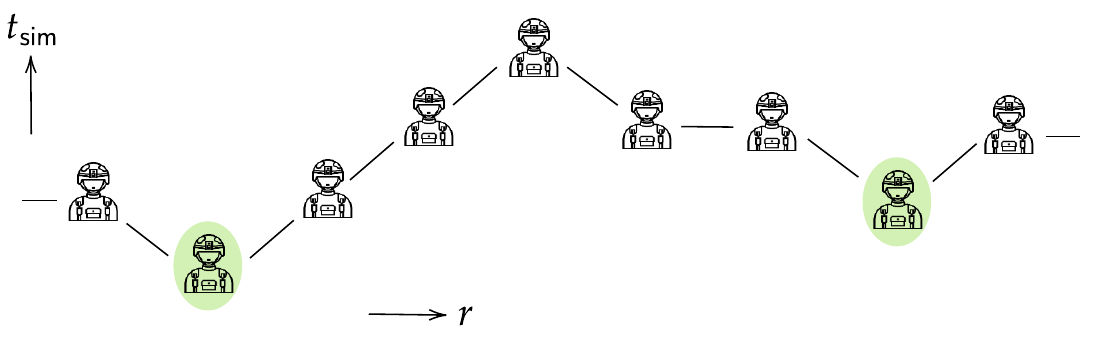}
		\caption{\label{fig:marching_soldiers} The marching soldiers scheme used to simulate a synchronous decoder with an asynchronous one. Each solider represents the classical variable stored by the decoder at a given spatial site, with the soldier's vertical position indicating the simulation time $\tsim$ of that site, viz. how many automaton updates have occurred at that site. Soldiers are prevented from leaving their neighbors behind; this mandates that only the soldiers marked in green can perform an automaton update (increase their value of $\tsim$).  } 
	\end{figure}
	
	Our construction for producing a Lindbladian decoder from a synchronous one is based on the ``marching soldiers scheme'' of  Refs.~\cite{berman1988investigations,gacs2001reliable,gacs_synch_slides}, which as we will see can be implemented at the cost of increasing the number of classical control variables stored at each site by a constant factor.
	The basic idea is to use $\mcd_\lind$ to simulate the computation performed by a given synchronous decoder $\mcd$. This is done by having each site keep track of how many automaton updates it has performed, and to then require that the number of automaton updates performed by neighboring sites cannot differ by more than 1: if a jump operator applied by $\mcd_\lind$ attempts to update a site which would violate this rule, it instead acts trivially. This update rule is directly analogous to the process by which surface growth occurs in the restricted solid-on-solid model \cite{kim1989growth}, with the constraint ensuring that no site is ``left behind'' during the evolution: nearby sites are guaranteed to have made progress on $\mcd$'s computation by an amount that differs by at most a single time step (see fig.~\ref{fig:marching_soldiers} for a schematic). We let each site store the values of its control variables for both the current and previous simulation time steps, and $\mcd_\lind$ may then draw on this information to update any updateable site in the same way that $\mcd$ would. Implementing this scheme is slightly complicated by the fact that the feedback operators of $\mcd$ will generically modify the values of syndromes on multiple sites, before these sites have had a chance to be updated. This necessitates also storing several variables related to the recent history of syndrome measurements, in a manner that will be explained below. 
	
	In what follows, we will add superscripts $\sfa,\sfs$ to explicitly distinguish the control variables for $\mcd_\lind$ and $\mcd$, respectively, and we will let $\De$ denote the radius of the balls in which the support of all feedback operators are contained. The control variables $m_\bfr^\sfa$ for $\mcd_\lind$ are organized as 
	\be \label{new_variables} m^\sfa_\bfr = \{m_\bfr^\old, m_\bfr^\new,\s_\bfr^\old,\s_\bfr^\new,\s_\bfr^{\sf future},\lag_{\bfr\ra\bfr'},u_\bfr\},\ee 
	where $m_\bfr^{\old,\new}$ are both valued in the alphabet in which $\mcd$'s control variables $\bfm^\sfs$ are encoded, $\s^{\old,\new,\fut}_\bfr \in \{\pm1\}$, $\lag_{\bfr\ra\bfr'} \in \{-1,0,1\}$ is defined for each $\bfr' \in B_\bfr(2\De) \setminus \{\bfr\}$ (the reason for the factor of 2, which means that the radius over which the automaton rule operates is twice as large as that of the feedback operations, will become clear later), and $u\in \{0,\dots,v-1\}$, where $v$ as before determines the ratio of automaton updates to feedback operations.\footnote{We will always initialize the decoding dynamics by setting 
		\be \label{trivial_init} m^{\old,\new}_\bfr = 0,\qq\s^{\old,\new,\fut}=1,\qq\lag_{\bfr\ra\bfr'}=0,\qq u_\bfr=0.\ee
	} We now describe the roles played by each of these variables.
	
	We start by defining notions that will help us keep track of the progress made by the Lindbladian system in simulating its synchronous counterpart: 
	\ms\begin{definition}[simulation and lag times]\label{def:sim_times}
		Fix a spacetime coordinate $(\bfr,t) \in \L \times \rr$. The {\it simulation time} $\tsim(\bfr,t) \in \zz^{\geq 0}$ is the number of times the decoder has successfully updated the control variables $m^{\new}_\bfr$. The {\it minimum simulation time} $\tsimmin(t)$ is the smallest simulation time in the system at time $t$:  
		\be \tsimmin(t) = \min_\bfr \tsim(\bfr,t),\ee 
		and the {\it maximum simulation time} is likewise $\tsimmax(t) = \max_\bfr \tsim(\bfr,t)$. 
		For a pair of neighboring sites $\bfr,\bfr'$, the {\it lag time} $\lag_{\bfr\ra\bfr'}(t)$ from $\bfr$ to $\bfr'$ at time $t$ is the difference in the simulation times at the two sites: 
		\be \lag_{\bfr\ra\bfr'}(t) = \tsim(\bfr,t) - \tsim(\bfr',t).\ee 
	\end{definition}\ms 
	
	Our architecture will store the values of $\lag_{\bfr\ra\bfr'}$ at each site using the appropriate field of $m^\sfa_\bfr$, and use these to control how automaton updates are performed.\footnote{Note that the automaton will not actually store the values of $\tsim(\bfr,t)$, which are introduced purely for convenience in analyzing the performance of the decoder. } The full control architecture is laid out in the following definition: 
	\ms \begin{definition}[marching soldiers desynchronization]\label{def:marching_soldiers}
		Given a local automaton rule $\mca$ and a feedback operation $\mcf$ on the control variables $\bfm^\sfs,\bfsig^\sfs$, its {\it marching soldiers desynchronization} $\mca_\asy$ is the automaton on the control variables $\bfm^\sfa,\bfsig^\sfa$ (c.f.  \eqref{new_variables}) defined as follows. 
		
		Define 
		\be \wt B_\bfr(x) = B_\bfr(x) \setminus \{ \bfr\}.\ee 
		For a site $\bfr$, define the set 
		\be \bfm^{\sf present}_\bfr =\bigcup_{\bfr' \in \wt B_\bfr(\De) \, : \, \lag_{\bfr\ra\bfr'} = 0} \{m_{\bfr'}^\new\}  \cup \bigcup_{\bfr' \in \wt B_\bfr(\De) \, : \, \lag_{\bfr\ra\bfr'} = -1} \{m_{\bfr'}^\old\},\ee 
		and likewise for $\bfsig^{\sf present}_\bfr$. 
		
		Regardless of the values of the lag variables $\lag_{\bfr\ra\bfr'}$, $\mca_\asy$ always updates 
		\be\label{sigmaupdates} \s^\new_{\bfr'} \mt \s^\sfa_{\bfr'} \s^{\sf future}_{\bfr'} \,\,\, \forall \, \, \bfr' \in B_\bfr(\De).\ee 
		Furthermore, if $\lag_{\bfr\ra \bfr'} \leq 0 \, \, \forall \, \, \bfr' \in \wt B_\bfr(2\De)$---so that a proposed automaton update at $\bfr$ will be  accepted---$\mca_\asy$ simultaneously performs the following updates:
		\bea\label{asynch_rules}
		m_\bfr^\old & \mt m_\bfr^\new  \\
		m_\bfr^\new & \mt  \mca(\bfm^{\sf present}_\bfr,\bfsig^{\sf present}_\bfr) \\ 
		\s^\old_\bfr & \mt \s^\new_\bfr \\ 
		\s^\fut_\bfr &\mt 1 \,\, \text{ if} \,\, u_\bfr = v-1 \, \, \text{else }\,  \s^\fut_\bfr \\  
		\lag_{\bfr\ra\bfr'} & \mt  \lag_{\bfr\ra\bfr'}+1 \, \, \forall \, \, \bfr'\in \wt B_\bfr(2\De) \\   
		\lag_{\bfr'\ra\bfr} & \mt  \lag_{\bfr'\ra\bfr}-1 \, \, \forall \, \, \bfr' \in \wt B_\bfr(2\De) \\ 
		u_\bfr & \mt (u_\bfr + 1) \mod v\\
		\s^\fut_{\bfr'} & \mt \s^\fut_{\bfr'} \(\s_{\mcf_{\bfr}(\bfm^\pres,\bfsig^\pres),\bfr'}\)^{1 + \lag_{\bfr \ra \bfr'}} \, \, \forall \, \, \bfr'\in \wt B_\bfr(\De)\eea 			
	\end{definition}\ms 
	
	To help unpack this definition a bit, we offer the following comments:
	\begin{enumerate}
		\item The $\lag$ variables are updated according to the marching soldiers scheme, where an update can occur at $\bfr$ at time $t$ only if $\tsim(\bfr',t) \geq \tsim(\bfr,t)$ for all $\bfr' \in \wt B_\bfr(2\De)$.
		\item Feedback operations at a site $\bfr$ only occur after $v$ automaton updates have been applied to $\bfr$. 
		\item $\s^\fut_\bfr$ represents the accumulated change in the syndrome at site $\bfr$ due to feedback operations applied by neighboring sites $\bfr'$ that have advanced to larger simulation times.
		\item $\s^\new_\bfr$ is updated to reflect the eigenvalue that $g_\bfr$ would have if one were to ignore the changes to this eigenvalue that occur when sites neighboring $\bfr$ progress to larger simulation times. 
		\item The accumulated syndrome change from these sites is reset after feedback occurs at $\bfr$.
		\item $m^\old_\bfr, \s^\old_\bfr$ store the values of $m^\new_\bfr,\s^\new_\bfr$ at the most recent simulation step, which are used to update neighboring sites $\bfr'$ with $\tsim(\bfr',t) = \tsim(\bfr,t) - 1$ in the same way in which they would be updated according to the synchronous automaton.
		\item The functions of $\bfm^{\sf present}_\bfr, \bfsig^{\sf present}_\bfr$ are to store the values of the control variables that the synchronous automaton would use to update $m^\sfs_\bfr$ (that they indeed do this will be proven below). 
		\item The distinction at various places between balls of radii $\De$ and $2\De$ is made so that the marching soldiers scheme (viz. the interaction range of $\mca_\asy$) extends to a distance twice the size of that over which $\mca,\mcf$ draw input from the control variables (viz. twice the size of the range over which each feedback operation is applied). This separation will be useful for proving a result about desynchronized simulation in the next section. 		
	\end{enumerate}
	
	Our desynchronization procedure is defined by combining this scheme with an appropriately modified feedback operation: 
	\ms\begin{definition}[desynchronized decoders]\label{def:desynch}
		Let $\mcd = (\mca,\mcf,\mcc)$ be a synchronous automaton decoder with control variables $\bfm^\sfs,\bfsig^\sfs$. 
		The {\it desynchronization} of $\mcd$, written $\mcd_{\sf asynch} = (\mca_\asy,\mcf_\asy,\mcc)$, is the Lindbladian automaton decoder with control variables $\bfm^\sfa,\bfsig^\sfa$ as in \eqref{new_variables}, automaton rule $\mca_\asy$ defined in Def.~\ref{def:marching_soldiers}, and feedback operators 
		\be\label{asynch_feedback} \mcf_{\asy,\bfr}(\bfm^\sfa,\bfsig^\sfa) = \begin{dcases} \mcf_\bfr(\bfm_\bfr^{\sf present}, \bfsig_\bfr^{\sf present}) & \(\lag_{\bfr\ra\bfr'}\leq 0 \, \, \forall \,\, \bfr'\in \wt B_\bfr(\De)\) \wedge (u_\bfr = v-1) \\ 
			\unit & {\rm else} \end{dcases}.\ee 
	\end{definition}\ms 
	This rule ensures that the ratio of automaton updates to feedback operations is precisely equal to $v$, and that nontrivial feedback is applied only when an automaton update is allowed to occur, as desired.

	\ss{Faithful Lindbladian simulation and slowdowns} \label{ss:slowdowns}
	
	Having defined desynchronization, we now show exactly how the computations performed by  $\mcd_\asy$ and $\mcd$ are related. To help distinguish times in the discrete and continuous settings, we will mostly use sans-serif font for discrete time variables, and serif font for continuous ones. 
	
	The first result shows that the classical variables of the asynchronous system at sites with simulation time $\tsim$ match the states of their synchronous counterparts at this time step: 	
	\ms\begin{proposition}[faithful Lindbladian simulation: control variables]\label{prop:faithful_asynch_sim}
		Consider a spacetime point $(\bfr,t)$. Let $m^\sfs_\bfr(\sft),\s^\sfs_\bfr(\sft)$ denote the values of $m^\sfs_\bfr,\s^\sfs_\bfr$ after $\sft$ steps of a synchronous decoder $\mcd$ on a particular input state $\r$, and let  $m^\new_\bfr(t),\s^\new_\bfr(t)$ be the values of $m^\new_\bfr,\s^\new_\bfr$ after evolving with $\mcd_\asy$ for time $t$ on $\r$. Then for all $(\bfr,t)$, 
		\be \label{faithful_sim_control} m^\new_\bfr(t) = m^\sfs_\bfr(\tsim(\bfr,t)),\ee 
		and if $\lag_{\bfr\ra\bfr'} \leq 0 \,\,  \forall \, \, \bfr'\in \wt B_\bfr(\De)$, then 
		\be \label{faithful_sig_control} \s^\new_\bfr(t) = \s^\sfs_\bfr(\tsim(\bfr,t)).\ee 
	\end{proposition}
	The need to wait until sites in $\wt B_\bfr(\De)$ have ``caught up'' to $\bfr$ in order for \eqref{faithful_sig_control} to hold is due to the fact that feedback operators applied during updates of sites in $\wt B_\bfr(\De)$ will generically change the value of $\s^\new_\bfr$ (while $m^\new_\bfr$ is by contrast only modified by updates to site $\bfr$ itself). 
	\begin{proof}
		We argue by induction on $\tsim(\bfr,t)$. The base case of $\tsim(\bfr,t) = 1$ is trivial. For the induction step, fix a time $t$ and consider a point $\bfr$ where $\tsim(\bfr,t) = \sft$, and assume the result to be true for all $(\bfr',t')$ such that $\tsim(\bfr',t') \leq \sft$. Consider what happens when an update at $\bfr$ is accepted. If this happens we must have $\lag_{\bfr\ra\bfr'} \leq 0 \, \,\forall\,\,\bfr'\in \wt B_\bfr(\De) \subset \wt B_\bfr(2\De)$, and we first claim that the induction hypothesis implies 
		\be \label{presentcond} \bfm^{\sf present}_\bfr(t) = \bfm^\sfs(\sft)|_{B_\bfr(\De)}.\ee 
		First, we clearly have $(\bfm^\pres_\bfr(t))|_\bfr = m^\sfs_\bfr(\sft)$ by assumption. Consider then a point $\bfr'\in \wt B_\bfr(\De)$. If $\tsim(\bfr',t) = \sft$, then $\lag_{\bfr\ra\bfr'} = 0$ and 
		\be (\bfm_\bfr^{\sf present}(t))|_{\bfr'} = m^\new_{\bfr'}(t) = m^\sfs_{\bfr'}(\sft)\ee 
		by the induction hypothesis. Suppose on the other hand that $\tsim(\bfr',t) = \sft+1$. Then $\lag_{\bfr\ra\bfr'} = -1$ and $(\bfm^{\sf present}_\bfr(t))|_{\bfr'}=m^\old_{\bfr'}$, which equals the value of $m^\new_{\bfr'}$ at the time $t'$ right before the most recent update at $\bfr'$ occurred. Since $\tsim(\bfr',t')=\tsim(\bfr',t)-1=\sft$, the induction hypothesis then implies $m^\old_{\bfr'} = m^\sfs_{\bfr'}(\sft)$, showing the claim. 
		
		To complete the proof, we need to show the analogous result for $\s^\pres$, viz.
		\be \label{sigcond} \bfsig^\pres_\bfr(t) = \bfsig^\sfs(\sft)|_{B_\bfr(\De)}.\ee 
		Indeed, if this is true, then at the next time $t'$ when $\bfr$ successfully updates, $\mca_\asy$ runs to produce a value of $m^\new_\bfr$ which is computed 
		using the same function running on the same input as $\mca$, meaning that when $\tsim(\bfr,t') = \sft+1$, we continue to have $m^\new_\bfr(t') = m^\sfs_\bfr(\tsim(\bfr,t'))$; this then will show \eqref{faithful_sim_control}. 
		Furthermore, the definition in \eqref{asynch_feedback} ensures that the feedback operator $\mcf_{\asy,\bfr}$ applied at $\bfr$ also exactly matches the operator applied in the synchronous system. The next time $t'$ a measurement at $\bfr$ occurs, the resulting syndrome $\s^\sfa_\bfr(t')$ will then match the 
		value of $\s^\sfs_\bfr(\sft+1)$, up to the following corrections: 
		\begin{enumerate}
			\item Changes to $\s^\sfa_\bfr$ applied by feedback on sites $\bfr'\in \wt B_\bfr(\De)$ with $\tsim(\bfr',t') = \sft$ (as these operations have not yet been implemented by time $t'$). These changes will be correctly accounted for provided we wait until $\lag_{\bfr\ra\bfr'} \leq 0 \, \, \forall \, \, \bfr' \in \wt B_\bfr(\De)$, as in the statement of the proposition.
			\item Changes to $\s^\sfa_\bfr$ that occur due to the application of feedback at sites $\bfr'\in\wt B_\bfr(\De)$ with $\lag_{\bfr\ra\bfr'} = -1$. These changes are accounted for using $\s^\fut_\bfr$, with \eqref{sigmaupdates} and \eqref{asynch_rules} ensuring that these effects are appropriately subtracted in the computation of $\s^\new_\bfr$. 
		\end{enumerate}
		Proving \eqref{sigcond} will therefore also show \eqref{faithful_sig_control}.
		
		The proof of \eqref{sigcond} is very similar to that of \eqref{presentcond}.  
		Consider again what happens when at update at $\bfr$ is accepted, which happens only if $\lag_{\bfr\ra\bfr'} \leq 0 \, \, \forall \,\, \bfr' \in \wt B_\bfr(2\De)$. This implies 
		\be \label{goodcond} \lag_{\bfr_1 \ra \bfr''} \leq 0 \, \, \forall \, \, \bfr''\in \wt B_{\bfr_1}(\De),\ee 
		since $\wt B_{\bfr_1}(\De) \subset \wt B_\bfr(2\De)$, which allows us to apply the induction hypothesis at site $\bfr_1$, using the same argument as before.\footnote{If we instead had chosen the $\lag_{\bfr\ra\bfr'}$ variables to extend only to the sites in $\wt B_\bfr(\De)$, we would not necessarily have \eqref{goodcond}, and thus not be able to use the induction hypothesis at $\bfr_1$.} Explicitly, if $\lag_{\bfr\ra\bfr_1}=-1$, then $(\bfsig^\pres_\bfr)|_{\bfr_1} = \s^\old_{\bfr_1}$, which matches $\s^\new_{\bfr_1}$ at the time $t'$ just before the last update at $\bfr_1$ (where $\tsim(\bfr_1,t')=\sft$), and hence matches $\s^\sfs_{\bfr_1}(\sft)$ by the induction hypothesis.
		If on the other hand $\lag_{\bfr\ra\bfr_1} = 0$, then $(\bfsig^\pres_\bfr(t))|_{\bfr_1} = \s^\new_{\bfr_1}(t) = \s^\sfs_{\bfr_1}(\sft)$. This completes the proof of \eqref{sigcond}. 
	\end{proof}
	
	It is clear that $\mcd_{\sf asynch}$ runs more slowly than $\mcd$, since a nontrivial amount of time will inevitably be spent attempting to update sites $\bfr$ that have larger simulation times than at least one other site in $\wt B_\bfr(2\De)$, and hence reject the proposed update. How severe is the slowdown? A result from Berman and Simon \cite{berman1988investigations} shows that with high probability, the slowdown is actually an innocuous constant factor, provided $\mcd_{\sf asynch}$ has simulated at least $\O(\log L)$ steps of $\mcd$'s computation. We will in fact prove a bit more than this: 
	\ms\begin{lemma}[constant slowdown for desynchronization]\label{lem:slowdown_lemma}
		Consider an asynchronous automaton updating according to the marching soldiers scheme. There exist positive $O(1)$ constants $\g<1, \l>1,c>0$  such that the following are true as long as $t \geq \frac{2D}{c}\ln L$:\footnote{Here we are continuing to set $\mu=1$; if we did not do this we would need to multiply $t_{{\sf sim},{\sf max/min}}$ by $\mu$.}
		\be \label{tminbound_lemma}  \sfP( \tsimmin(t) < \g t) < e^{-c t} \ee
		and 
		\be  \label{tmaxbound_lemma} \sfP( \tsimmax(t) > \l t) < e^{-c t}. \ee
	\end{lemma} 	
	Since Berman and Simon did not provide the full details of their proof, formulated the constraints differently, and did not give the bound on $\tsimmax(t)$, we provide the full proof in Appendix~\ref{app:slowdown}.
	
	Suppose we are interested in simulating $\sft$ steps of the computation done by a synchronous decoder $\mcd$ using its desynchronization $\mcd_\asy$. We will only be interested in computations which end after a finite time. To this end, we will say that $\mcd$ {\it terminates in time $\sft_{\sf end}$ on input $\r$} if $\mcd^{\sft'}$ applies no nontrivial feedback operators for all $\sft' > \sft_{\sf end}$ (in the context of our decoders, $\sft_{\sf end} = \tdec(\r)$).\footnote{Note that the action of $\mcd^{\sft'}$ on $\mch_{\sf cl}$ needn't be trivial when $\sft'>\sft_{\sf end}$---we require only that $\mcd^{\sft'}$ act trivially on $\mch_{\sf qu}$. } In this case, the previous Lemma ensures that at long enough times, desynchronized decoders output the same quantum state as their synchronous counterparts w.h.p:  
	\ms\begin{corollary}[faithful Lindbladian simulation: quantum states]\label{prop:faithful_asynch_sim_quantum}
		Suppose that $\mcd$ terminates in time $\sft_{\sf end}$ on input $\r$. Then there are $O(1)$ constants $\g,f$ such that as long as $t > \sft_{\sf end}/\g + f \log(L)$,\footnote{Depending on the automaton rules, $\mcd^t_\asy(\r)$ may have a large probability to differ from $\mcd^{\sft_{\sf end}}(\r)$ on $\mch_{\sf cl}$. For the message passing rules of sec.~\ref{sec:msg_passing} however, where messages decay in the absence of anyon sources, the trace over $\mch_{\sf cl}$ is unnecessary provided one waits the additional time ($\O(\log L)$ w.h.p if $p<p_c$) needed for all messages to decay.}
		\be \label{goal} \frac12 ||\mcd_\asy^t(\r)_{\sf qu}- \mcd^{t_{\sf end}}(\r)_{\sf qu}||_1 < e^{-ct}.\ee 
	\end{corollary} 
	\begin{proof}
		We claim that
		\be \label{equaleventually} \tsimmin(t) \geq \sft_{\sf end} \implies \mcd_\asy^t(\r)_{\sf qu} =\mcd^{\sft_{\sf end}}(\r)_{\sf qu}.\ee 
		This is a direct consequence of Prop.~\ref{prop:faithful_asynch_sim}, which, together with the fact that the feedback operators are trivial at all spacetime points such that $\tsim(\bfr,t) > \sft_{\sf end}$, guarantees that the same feedback operators have been applied by the two decoders when $\tsimmin(t) \geq \sft_{\sf end}$. 
		
		Lemma~\ref{lem:slowdown_lemma} shows that $\sfP(\tsimmin(t) < \sft_{\sf end}) < e^{-ct}$ provided $t > \sft_{\sf end}/\g + f \log(L)$ for appropriate $O(1)$ $\g,f$; if $t$ satisfies this inequality, then \eqref{equaleventually} implies the existence of a density matrix $\wt \r$ such that 
		\be \label{dasydecomp} \mcd^t_\asy(\r)_\qu = \mcd^{\sft_{\sf end}}(\r)_\qu (1-p_{\sf slow}) + p_{\sf slow} \wt \r \ee 
		where $p_{\sf slow} < e^{-ct}$. Inserting \eqref{dasydecomp} into the trace distance in \eqref{goal} gives the desired result. 
	\end{proof}
	%
	%

	\ss{Performance guarantees for Lindbladian decoding}\label{ss:asynch_offline}
	
	In this subsection, we show that the results of sec.~\ref{sec:offline} continue to hold under desynchronization, which is a rather direct consequence of the constant slowdown guarantee provided by Lemma~\ref{lem:slowdown_lemma}. In particular, we show: 
	
	\ms\begin{theorem}[equivalent performance of synchronous and asynchronous decoders]\label{thm:equiv_asynch}
		Let $\mcd$ be a synchronous automaton decoder with logical error rate $\plog^{(\mce)}$ and decoding time $\tdec^{(\mce)}$ for offline decoding under error model $\mce$. Let $\plogae$, $\tdecae$ be the logical failure rates and expected decoding times for its desynchronization $\mcd_{\sf asynch}$. Then 
		\be \plogae = \plog^{(\mce)}, \qq \tdecae = O(\tdec^{(\mce)} + \log(L)) . \ee 
		Furthermore, if $\tdec^{(\mce)} = \O(\log L)$,\footnote{Technically, we actually require that for $\r\sim\mce$, $\tdec(\r)\ran > f\ln L$ for some $O(1)$ constant $f$ w.h.p. } then $\tdecae = \ct(\tdec^{(\mce)})$. 
	\end{theorem}
	As we have seen, there exist $p$-bounded error models with $\tdec^{(\mce)} = \O(\log L)$. Therefore the worst-case (with respect to the choice of $p$-bounded $\mce$) performance of $\mcd$ and $\mcd_\asy$ is asymptotically identical: 
	\be \max_{\mce \, \text{$p$-bounded}} \plogae = \max_{\mce \, \text{$p$-bounded}} \plog^{(\mce)}, \qq \max_{\mce \, \text{$p$-bounded}}  \tdecae  = \ct\(\max_{\mce \, \text{$p$-bounded}}\tdec^{(\mce)}\).\ee 
	\begin{proof}
		That $\plogae = \ploge$ is easy to see: by Prop.~\ref{prop:faithful_asynch_sim}, for any initial state $\r\sim \mce$  we have $\mcd^{\tsimmin(t)}(\r)_\qu = \mcd_{\sf asynch}^t(\r)_\qu$ if there are no anyons in the system by time $t$; the logical error rates of $\mcd$ and $\mcd_\asy$ are thus identical.

		The stated result about $\tdecae$ follows as a consequence of Lemma~\ref{lem:slowdown_lemma}. 
		We will first show $\tdecae = O(\tdec^{(\mce)} + \log(L))$. Given a noisy input state $\r\sim \mce$, let the random variable $\tdeca(\r)$ denote the decoding time of $\r$ under $\mcd_\asy$ for a particular (random) choice of proposed update times (defined in the context of the unraveling into stochastic pure-state trajectories, c.f. the discussion near \eqref{rates}). 
		We have 		 
		\be \EE [\tdeca(\r)]  =  \int_0^\infty dt\, \sfP(\tdeca(\r) > t) \leq \int_0^\infty dt \,\sfP(\tsimmin(t) < \tdec(\r)) ,\ee 
		where the expectation value is taken over the random proposed update times and $\tdec(\r)$ is as in \eqref{tdecrhodef}. Here the equality is the standard tail-sum formula, and the inequality holds because the event $ \tdeca(\r) > t$ implies $\tsimmin(t) < \tdec(\r)$ (again with $\tsimmin(t)$ treated as a random variable depending on the sequence of proposed update times). Breaking the integral up at a time $t_*$, we have 
		\bea  \EE [\tdeca(\r)]  & 
		\leq t_* + \int_{t^*}^\infty dt\, \sfP(\tsimmin(t) < \g_t t),\eea 
		where we defined $\g_t = \tdec (\r)/ t$. Letting $c,\g$ be the constants in the statement of Lemma~\ref{lem:slowdown_lemma}, we may apply the bound in \eqref{tminbound_lemma} to the integral if $t_* \geq \frac{2D}c \ln L$ and $\g_t \leq \g$. We then choose 
		\be t_* = \max\(\frac{2D}c \ln L, \frac{\tdec(\r)}\g\),\ee 
		which gives
		\be \EE [\tdeca(\r)]  
		\leq t_* + \int_{t_*}^\infty dt \, e^{-ct} = O(\tdec(\r) + \log(L)).\ee 
		Averaging over $\r\sim\mce$ then gives the desired result.
		
		We now show $ \tdecae = \O(\tdec^{(\mce)})$ if $\tdece = \O(\log L)$. 
		Again using the tail-sum formula and fixing two constants $a,b$ with $a<b$, we have 
		\bea  \EE [\tdeca(\r)] &  =   \int_0^\infty dt\, \sfP(\tdeca(\r) > t) \geq \int_0^\infty dt\, \sfP(\tsimmax(t) < \tdec(\r)) \\ & \geq \int_a^b dt\, (1 - \sfP(\tsimmax(t) \geq \l_t t )),\eea
		where we used that the event $\tsimmax(t) < \tdec(\r)$ implies the event $\tdeca(\r) > t$, and defined $\l_t = \tdec(\r)/t$. 
		
		We now apply Lemma~\ref{lem:slowdown_lemma} to bound the remaining integral. We in fact need to make use of the stronger result derived in \eqref{maxbound}, which gives us an exponentially small upper bound on $\sfP(\tsimmax(t) \geq \l_t t)$ provided 
		\be \frac{2D}{c_\l} \leq t \leq \frac{\tdec(\r)}{2e},\ee 
		where $c_\l = (1+\l \ln2)/2$. If we assume $\tdec(\r) >  f \log(L)$ for some $O(1)$ constant $f$, this can be accomplished by choosing e.g. $\l = \l_f = 2(4ez/f - 1)/\ln2$. Making this choice, and taking $a=\tdec(\r)/4e$, $b= \tdec(\r)/2e$, we have 
		\be \EE[\tdeca(\r)] \geq \int_{\tdec(\r)/4e}^{\tdec(\r)/2e} dt\, (1-e^{-c_{\l_f} t}).\ee 
		After averaging over $\r\sim\mce$, this then gives $\tdecae = \ct(\tdec^{(\mce)})$ provided $\tdec(\r) > f\log(L)$ w.h.p. 
	\end{proof}

	\section{Numerics } \label{sec:numerics}
	
	In this section, we use Monte Carlo simulations to quantitatively address the performance of the message-passing decoder under i.i.d Pauli noise. The simulations are performed by subjecting a clean initial state in the logical-$Z$ eigenbasis to Pauli-$X$ noise sampled from an i.i.d distribution of strength $p$, and then running the decoder on this noisy state until no anyons remain. $\ploge$ is estimated as the fraction of states $\r \sim \mce_{i.i.d}$ for which the decoder returns an output state different from the input, and $\tdece$ is estimated as the average number of steps the decoder performs before no anyons remain. For simplicity we will not set a truncation on the maximal allowed message size, and thus will utilize  $\ct(\log(L))$ control bits on each site. All code used to create the plots to follow is available at  \cite{code}. 
	
	We find threshold error rates of $p_c = 1/2$ in 1D and $p_c \approx 7.3\%$ in 2D (again, no attempt has been made to choose an automaton rule which maximizes this value). For both the 1D and 2D decoders, $\tdec$ and $\plog$ are observed to scale with $L$ below threshold in the manner predicted by Conjecture~\ref{conj:local_noise} and the general intuition in sec.~\ref{sec:intro}. 
	
	All of our simulations are run with periodic boundary conditions and a field update speed of $v=3$. We specify to synchronous decoders until sec.~\ref{ss:uncoord}; by virtue of Theorem.~\ref{thm:equiv_asynch} the same results hold upon desynchronizing.

	\ss{One dimension} 	
	\begin{figure}
		\hfill \includegraphics[width=.45\tw]{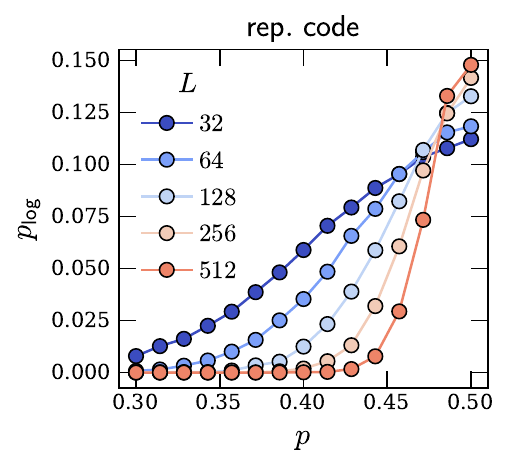}  \hfill 
		\includegraphics[width=.45\tw]{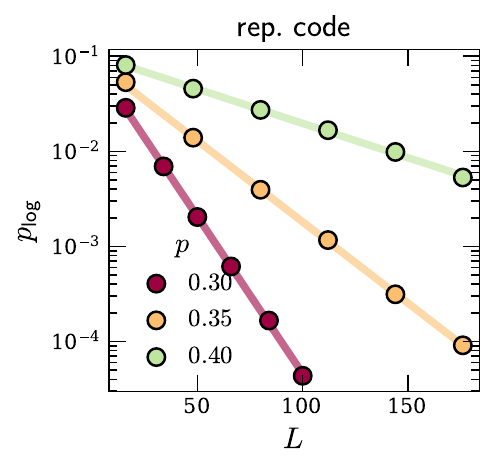} \hfill 
		
		\caption{\label{fig:1d_offline} Performance of the offline decoder for the 1D repetition code. {\it Left:} logical failure rate $\plog$ as a function of $p$. The data are consistent with a threshold of $p_c \ra 1/2$. $\plog(p=1/2) < 1/2$ for finite system sizes because we take decoding to succeed as long as the majority vote of the inputs matches the logical state of the output.  {\it Right:} Scaling of the logical failure rate with $L$ below threshold. The solid lines are fits of the data to the form $\plog = C(p/p_c)^{\g L}$ with $\g \sim 1/10$. }
		
	\end{figure}

	\begin{figure}
		\centering 
		\includegraphics[width=.45\tw]{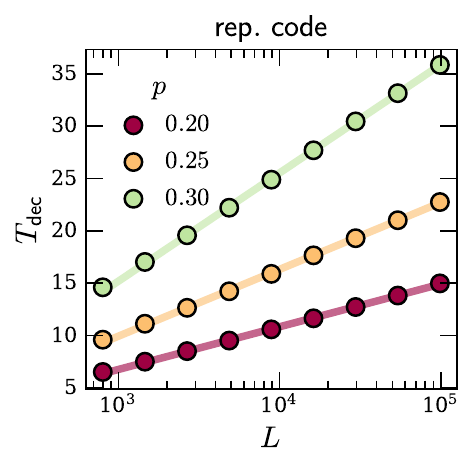}
		\caption{\label{fig:1d_tdec} Decoding time $\tdec$ as a function of $L$ for the 1D repetition code. The solid lines are fits of the data to the form $\tdec = A\log(L)+B$, with $p$-dependent constants $A,B$.}
	\end{figure}
	
	We start with the 1D repetition code. Here we find decoding times to be decreased when a small amount of randomness is added to the update rules; the simulations presented here move each anyon in a random direction $10\%$ of the time, and move it according to the correct update rule the remaining $90\%$. 
	The left panel of fig.~\ref{fig:1d_offline} plots $\plog$ as a function of $p$ at different system sizes. The crossing points for different values of $L$ drift towards {\it larger} values of $p$ as $L$ is increased, and at face value are consistent with a threshold at $p_c =1/2$. The right panel of fig.~\ref{fig:1d_offline} shows $\plog$ as a function of $L$ for different values of $p$, showing an excellent fit to an exponential scaling $\plog = e^{-\ct(L)}$ in this regime. 
	
	It is not a priori obvious that our decoder should have $p_c = 1/2$: this is the threshold one gets upon using a global majority vote decoder, and it is perhaps surprising that a completely parallelized local scheme can reproduce the same computation. A hint of why $p_c = 1/2$ is nevertheless reasonable comes from comparing the average distance between anyons in a pair to the inverse of the average anyon density: when the former is larger, we expect decoding to fail. To this end, define $\lan r_{\sf pair}\ran$ as the average distance between anyons in a pair (with anyons being paired by the strings created by the noise, as usual), and define $\lan \r\ran$ as the average anyon density. For i.i.d noise, one finds 
	\be\lan  \r\ran  = 2 p (1-p)\ee
	which comes from adding the (identical) probabilities of the two ways that an anyon can be present at a given site. The mean pair length is 
	\bea\lan r_{\sf pair}\ran & = (1-p)  \sum_{k=0}^\infty k p^{k-1}  = \frac1{1-p},\eea 
	and combining these gives $\lan r_{\sf pair}\ran\lan \r\ran = 2p$,
	which is strictly less than 1 except when $p=1/2$. This heuristically suggests that a decoder which moves anyons towards their nearest neighbors has the potential of attaining $p_c = 1/2$. 
	
	fig.~\ref{fig:1d_tdec} shows the expected decoding time $\tdec$ as a function of $L$, with the data exhibiting an excellent fit to $\tdec = \ct(\log L)$. This shows that the scalings predicted by Conjecture~\ref{conj:local_noise} hold for i.i.d noise in 1D (we have also investigated several non-i.i.d noise models with strictly local spatial correlations, and all are observed to have the same scaling as i.i.d noise).

	\ss{Two dimensions} 
	
	We now repeat the above analysis for the 2D toric code. The left panel of fig.~\ref{fig:2d_offline} plots $\plog$ as a function of $p$ for different system sizes. Since our simulations are done on the torus, where logical errors occur when anyon strings wrap around either one or both of the two nontrivial cycles, $\plog \ra 3/4$ in the limit $p\ra 1/2$. The crossing points for different values of $L$ are consistent with a threshold of $p_c \approx 7.3\%$, which unlike the 1D case is (a few percent) smaller than the optimum value of $\approx 11\%$ set by the phase transition of the 2D random bond Ising model on the Nishimori line~\cite{Dennis_2002}, but is a bit higher than the $\approx 6.1\%$ of the ``constant-$c$'' field-based decoder of Ref.~\cite{herold2015cellular}. The right panel of fig.~\ref{fig:2d_offline} plots $\plog$ against $L$ and confirms the expected scaling $\plog = e^{-\ct(L)}$ for $p$ not too far below threshold. Finally, fig.~\ref{fig:2d_tdec} shows $\tdec$ as a function of $L$ for several values of $p < p_c$, which again display an excellent fit to the predicted scaling $\tdec = \ct(\log L)$ (as in 1D, introducing local spatial correlations into the noise does not appear to affect the results). 
	
	\begin{figure}
		\hfill \includegraphics[width=.45\tw]{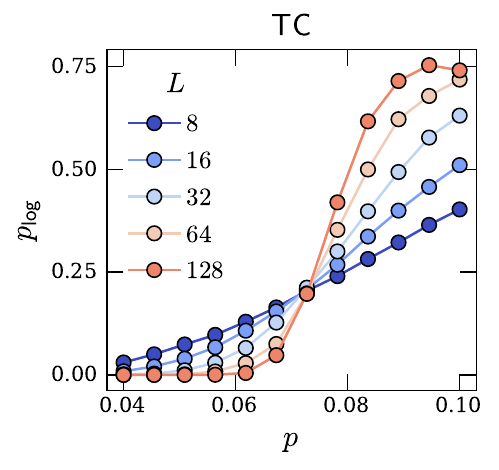} \hfill 
		\includegraphics[width=.45\tw]{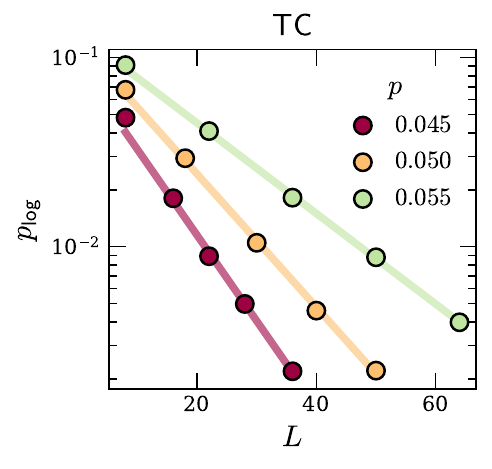}
		\hfill 
		\caption{\label{fig:2d_offline} Performance of the offline decoder for the 2D toric code. {\it Left:} logical failure rate $\plog$ as a function of $p$, with the crossing point indicating $p_c \approx 7.3\%$. {\it Right:} Scaling of the logical failure rate with $L$ below threshold. The solid lines are fits of the data to the form $\plog = C(p/p_c)^{\g L}$ with $\g \approx 1/5$.}
		
	\end{figure}
	
	\begin{figure}
		\centering 
		\includegraphics[width=.45\tw]{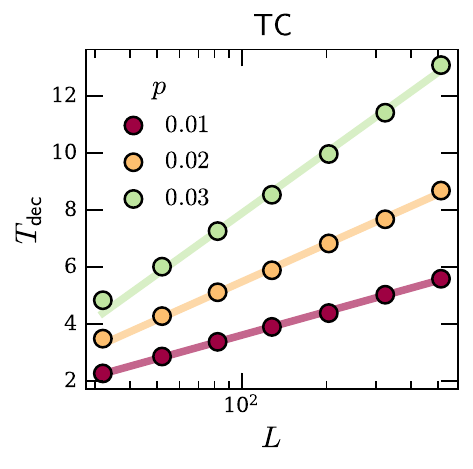}
		\caption{\label{fig:2d_tdec} Expected decoding time $\tdec$ below threshold as a function of $L$ for the 2D toric code. The solid lines are fits of the data to the form $\tdec = A \log(L)  +B$.} 
	\end{figure}
	
	We can give an extremely crude estimation of $p_c$ for i.i.d noise using the following simplistic argument. Let $\sfP_{\sf nearest\, \, paired}$ be the probability that a given anyon is paired with one of its nearest neighbors. Heuristically, our decoder should succeed when this probability is $>1/2$; in this case, moving each anyon towards its nearest neighbor is likely to induce a pairing leading to an error-free perfect matching. We compute $\sfP_{\sf nearest\, \, paired}$ numerically in fig.~\ref{fig:p_nearpair}, which shows that $\sfP_{\sf nearest\, \, paired} = 1/2$ at $p \approx 6.5\%$, which is slightly smaller than the numerically computed value of the threshold probability. It would be interesting if a refinement of this naive argument could be used to get a sharp bound on $p_c$. 
	\begin{figure}
		\centering \includegraphics[width=.45\tw]{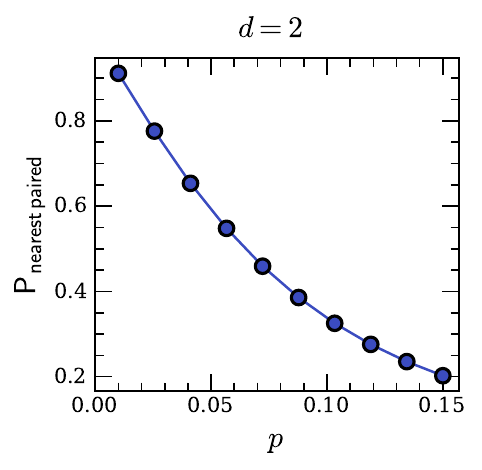} 
		\caption{\label{fig:p_nearpair} The probability for anyons in 2d created by i.i.d noise to be paired with one of their nearest neighbors. $\sfP_{\sf nearest\, \, paired} = 1/2$ around $p\approx 6.5\%$.}
	\end{figure}

	\ss{Uncoordinated desynchronization}  \label{ss:uncoord}
	
	The continuous time decoders studied in sec.~\ref{sec:lind} were constructed by augmenting the classical control architecture with additional bits that implement the marching soldiers scheme. This scheme was proven to produce Lindbladian decoders with the same values of $p_c$ and $\plog$ as their synchronous counterparts, along with the same scaling of $\tdec$. It is natural however to wonder about the performance of Lindbladian decoders constructed with a simpler descynhronization scheme, where we simply independently implement anyon motion and cellular automaton updates on each site in an uncoordinated way. We will call this type of Lindbladian construction the {\it uncoordinated desynchronization}.
	
	For all cellular automata capable of highly structured computations known to the author---such as the game of life and many of the 1D automata defined by Wolfram---performing uncoordinated desynchronization completely destroys the automaton's ability to perform meaningful computations. However, for Toom's rule, which like the present decoder is designed only to remember $O(1)$ bits of information, uncoordinated desynchronization is provably innocuous, changing the value of $p_c$ but not the existence of a threshold \cite{gray1999toom,bennett1985role,ray2024protecting}. While it is possible to construct local robust classical memories storing $O(1)$ bits for which uncoordinated desynchronization fails \cite{squeezing}, it is still reasonable to conjecture that it may work for the present message-passing decoder under consideration. 
	
	We will not attempt a rigorous analysis of uncoordinated desynchronization in this work, and will content ourselves with a numerical study of uncoordinated Lindbladian decoders under i.i.d noise. 
	
	fig.~\ref{fig:1d_uncoord_plog} studies the uncoordinated desynchronization of the message passing decoder in 1D by plotting $\plog$ as a function of $p$ for different values of $L$. Without the marching soldiers architecture, the error rate $\plog^{\sf un}$ for the uncoordinated decoder will generically not agree with its value under synchronous decoding; in general we expect $\plog^{\sf un} < \plog$. Our numerics are consistent with this, but nevertheless again indicate a threshold at the optimal value of $p_c = 1/2$. 
	
	Finally, fig.~\ref{fig:2d_uncoord_plog} shows the performance of the uncoordinated Lindbladian decoder in 2D. We again observe the existence of a threshold and that $\plog^{\sf un} < \plog$, but unlike in 1D observe that the threshold value of $p$ is decreased relative to the synchronous case, with $p_c^{\sf un} \approx 5.2\%$. The scaling of $\plog^{\sf un}$ as a function of $L$ is observed to scale as $p^{-\ct(L)}$, as in the synchronous case.

	\begin{figure}
		\centering 
		\includegraphics[width=.45\tw]{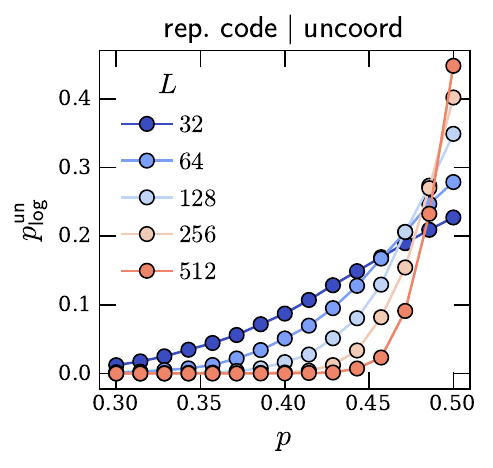}
		\caption{ \label{fig:1d_uncoord_plog} Scaling of $\plog^{\sf un}$ vs $p$ for the uncoordinated asynchronous 1D repetition code decoder under i.i.d noise of strength $p$; the evolution of the crossing points are consistent with a threshold at $p^{\sf un}_c = 1/2$.  } 
		
	\end{figure}

	\begin{figure}
		\hfill 
		\includegraphics[width=.45\tw]{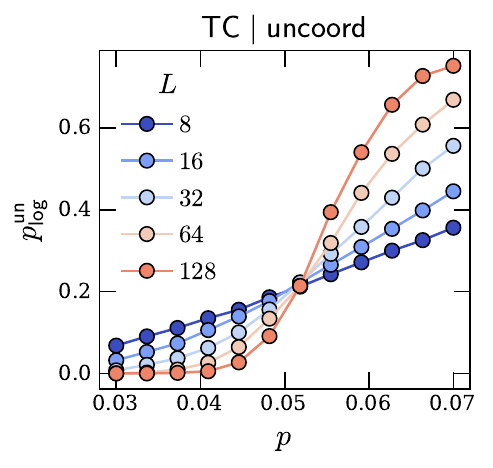}
		\hfill 
		\includegraphics[width=.45\tw]{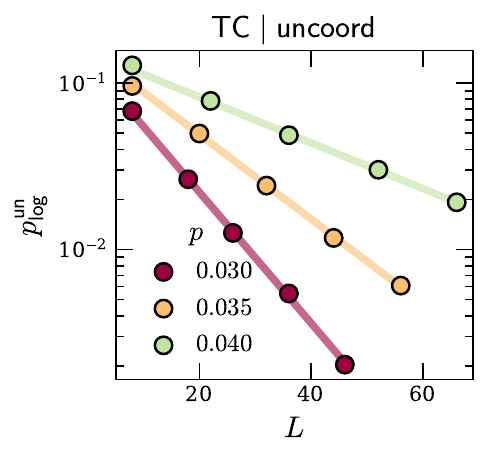}\hfill 
		\caption{ \label{fig:2d_uncoord_plog} Performance of the uncoordinated asynchronous 2D toric code decoder under i.i.d noise of strength $p$. {\it Left:} $\plog^{\sf un}$ vs $p$, consistent with a threshold around $p^{\sf un}_c \approx 5.2\%$.  {\it Right:} $\plog^{\sf un}$ vs $p$; the solid lines are fits to $\plog = C(p/p^{\sf un}_c)^{\g L}$ with $\g \approx 1/10$. } 
	\end{figure}

	\section{Acknowledgments}
	
	I thank Ehud Altman, Nicholas O'Dea, Ruihua Fan, Sarang Gopalakrishnan, Jeongwan Haah, Yaodong Li, Louis Paletta, Pablo Sala, Charles Stahl, Kazuaki Takasan, Zijian Wang, Hayata Yamasaki, and Mike Zaletel for discussions and feedback, Shankar Balasubramanian and Margarita Davydova for helpful discussions about Gacs' sparsity methods, synchronization, and collaboration on related work, Larry Chen and Noah Goss for insights into hardware constraints, and Nathaniel Selub for finding an egregious number of typos. I am supported by a Miller research fellowship. 
	
	\begin{appendix}

		\section{Sparsity theorems and hierarchical noise}\label{app:sparsity}

		In this appendix we adopt the strategy in Gacs and Capuni \cite{ccapuni2021reliable} to show how noise distributions can be organized in a hierarchy of increasingly large-sized clusters, with noise events rapidly becoming rare at large levels in the hierarchy (large cluster sizes). Our proof will closely follow that of Gacs and Capuni, with the following differences:
		\begin{itemize}
			\item For our applications to error correction, we will need to organize the noise into well-separated clusters rather than collections of isolated points;
			\item our proof will assume one-dimensional noise, with a generalization to arbitrary dimensions given as a corollary at the end;
			\item our definition of level-$k$ error rates will be a bit different;
			\item we will get slightly better bounds on some of the $O(1)$ constants that control the growth of cluster sizes.		
		\end{itemize}
		
		The discussion in this appendix is completely self-contained, and at some points will employ notation that differs slightly from the main text. 
		
		\ss{Basic notions, and definitions of level-$k$ noise and error sets}
		
		As mentioned above, to simplify the notation we will mostly consider the case where our noise is defined on $\zz$ (the case that would be relevant for offline decoding in the 1D repetition code). The generalizations of the results below to $\zz^{D>1}$ and $\zz_L^D$ is simple: the former will be presented in a corollary at the end, and the latter follows essentially immediately from the former. 
		
		As is common in the study of error correcting codes, our assumptions about the noise model will be rather mild, with us needing only to assume that the marginal probability for a certain set of locations to have errors decays exponentially in the size of that set. More precisely: 
		
		\ms\begin{definition}[$p$-bounded noise] A {\it noise realization} $\sfN$ is a random subset of $\zz$. We will say that the noise is $p$-bounded if, for all finite subsets $A \subset \zz$, 
			\be \label{pbound} \sfP (A \subset \sfN) \leq p^{|A|}.\ee 
		\end{definition}\ms
		In the following, we will always assume the existence of a finite $0 < p < 1$ for which the above holds.

		Before we begin, some needed vocabulary: 
		\ms\begin{definition}[isolation and linkedness] 
			The set of points within a distance $r$ of a point $x$ will be denoted $B_x(r)$. 
			Two points $x,y$ are said to be {\it $(r,R)$-isolated} if $y \not\in B_x(R) \setminus B_x(r)$; otherwise, $x,y$ are said to be {\it $(r,R)$-linked}. 
			
			Let $\sfN \subset \zz$ be a noise realization. 
			A point $x \in \sfN$ is {\it $(r,R)$-isolated} if $x$ is $(r,R)$-isolated from all points in $\sfN$. $\sfN$ is said to be $(r,R)$-isolated if all points in $\sfN$ are $(r,R)$-isolated. 
		\end{definition}\ms

		Isolation is the concept used by Gacs for organizing the points in $\sfN$. For our applications to the decoding problem, the following notion will be more useful: 
		\ms\begin{definition}[clustering]
			Let $\sfN \subset \zz$ be a noise realization. A $(W,B)${\it  -cluster} $C_{(W,B)}$ is a subset of $\sfN$ such that 
			\begin{enumerate}
				\item all points in $C_{(W,B)}$ are contained in an interval of width $W$, and 
				\item this interval is separated from other points in $\sfN$ by a buffer region of size at least $B$: 
				\be \dist(\sfN \setminus C_{(W,B)}, C_{(W,B)}) \geq B,\ee 
				where we have defined the distance between two sets $A_1,A_2$ as 
				\be \dist(A_1,A_2) = \min_{x\in A_1,y\in A_2} |x-y|.\ee 
			\end{enumerate}
			A point $x \in \sfN$ is called {\it $(W,B)$-clustered} if it is a member of a $(W,B)$ cluster, and a set is said to be $(W,B)$ clustered if all points it contains are $(W,B)$-clustered. 
		\end{definition}\ms
		\begin{figure}
			\centering \includegraphics[width=.5\tw]{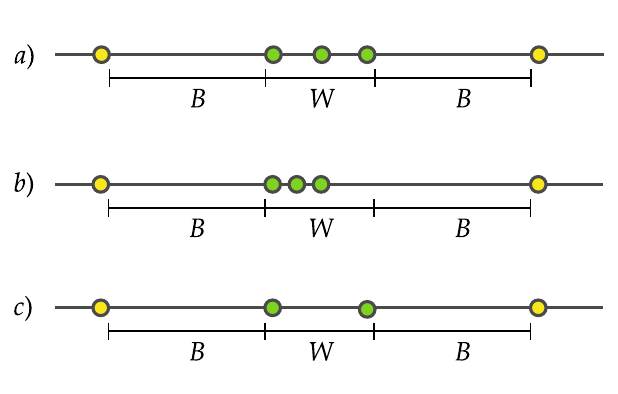}
			\caption{\label{fig:clustering_example} Isolation and clustering in one dimension. Circles represent noise points. In all panels, the green points form a $(W,B)$ cluster. In $a)$, only the central green point is $(W/2,W/2+B)$-isolated. In $b)$, all three green points are $(W/2,W/2+B)$-isolated, while in $c)$ neither of the green points are $(W/2,W/2+B)$-isolated.} 
		\end{figure}

		It follows from the above definition that any $(W/2,W/2+B)$-isolated point defines a $(W,B)$-cluster. However, a given cluster $C_{(W,B)}$ may contain anywhere between $0$ and $|C_{(W,B)}|$ $(W/2,W/2+B)$-isolated points; examples are shown in fig.~\ref{fig:clustering_example}. The partition of a (subset of a) subset of $\zz$ into clusters is as stated not necessarily unique in general, but it is if $B>W$: 
		\ms\begin{proposition}\label{prop:cluster_uniqueness}
			Let $\sfN \subset \zz$ be a noise realization, and $\{C_{(W,B)}^{(i)}\}$ be a collection of $(W,B)$-clusters in $\sfN$. If $B> W$, these clusters are disjoint: $C_{(W,B)}^{(i)} \cap C_{(W,B)}^{(j)}$ if $i\neq j$. 
		\end{proposition}
		\begin{proof}
			Fix $i, j$, assume $B > W$, and suppose $x \in C_{(W,B)}^{(i)} \cap C_{(W,B)}^{(j)}$. Consider a point $y \in C_{(W,B)}^{(i)}$. Then $|x-y| < W$. This implies $\dist(C_{(W,B)}^{(j)},\{y\}) < W$. Thus $\dist(C_{(W,B)}^{(j)}, \sfN \setminus C_{(W,B)}^{(j)}) < B$ since $B > W$, and hence we must have $i=j$. 
		\end{proof}

		The reason why we prefer to analyze things in terms of clustering rather than isolation is because we prove our threshold result for a decoder $\mcd$ by showing that anyons in a given cluster are always annihilated with other anyons from the same cluster. The same cannot be done as easily for isolated points: an isolated point in $\sfN$ may create anyons quite close to many other anyons, and this precludes us from similarly showing that the decoder eliminates all appropriately isolated points without causing an error.

		\ss{The sparsity theorem in one dimension}
		
		In the following we will fix a triple of constants $(\b,\g,n)$ chosen to satisfy the inequalities 
		\be  \label{paramineqs} \frac{2\g}{1-1/n} < \b < \g n.\ee 
		We will also define the numbers 
		\be \label{rbdefs} w_k = 2\b n^k, \qq b_k = \g n^{k+1} - \b n^k\ee 
		which accordingly must satisfy 
		\be 0<b_0<\frac{n-3}4w_0.\ee  
		
		These parameters control the scales involved in different levels of the hierarchical structure we use to coarse-grain the noise, in the manner made precise in the following definitions: 	
		\ms\begin{definition}[level-$k$ noise and error sets] 		
			Let $\sfN_0, \mcn_0 = \sfN$. The {\it $(k+1)$th level clustered noise set $\sfN_{k+1}$} is defined as the set obtained by deleting from $\sfN_k$ all $(w_k, b_k)$-clustered points. The {\it $(k+1)$th level isolated noise set $\mcn_{k+1}$} is defined as the set obtained by deleting from $\mcn_k$ all points which are $(\b n^k, \g n^{k+1})$-isolated.
			
			The subset of $\sfN$ that is deleted when passing from level $k$ to level $k+1$ defines the {\it level-$k$ clustered error set}, which we write as $\sfE_k$: 
			\be \sfE_k = \sfN_{k} \setminus \sfN_{k+1}.\ee 
			We similarly define the {\it level-$k$ isolated error set} $\mce_k$ as $\mce_k = \mcn_{k} \setminus \mcn_{k+1}$. 
		\end{definition}\ms 
		
		A $(w_k,b_k)$ cluster $C_{(w_k,b_k)}$ will be referred to as a {\it k-cluster} in what follows, and will often be written more succinctly as $C_k$. 
		
		\ms\begin{remark}
			If $b_k \geq w_k$, then by proposition~\ref{prop:cluster_uniqueness} the clusters removed when passing from level $k$ to level $k+1$ are disjoint: 
			\be \sfE_k = \bigsqcup_i C_k^{(i)} .\ee 
			Since $\sfN = \bigsqcup_k \sfE_k$, this gives a decomposition of $\sfN$ into disjoint clusters as 
			\be \sfN = \bigsqcup_{k,i} C_k^{(i)}.\ee 
			We need to assume this for our proof of Lemma~\ref{lem:const_slowdown}, but will not make this assumption in the remainder of this section. 
		\end{remark}
		
		Our goal in the following will be to show that level-$k$ clustered errors become increasingly rare as $k$ gets larger, viz. to show that if $x$ is sampled randomly from $\sfN$, it is very unlikely to be part of $\sfN_k$ as $k$ gets large. 	
		We quantify this by defining the level-$k$ error rate as follows:\footnote{Note that this definition differs from that in \cite{ccapuni2021reliable} both in the fact that it involves only a single point $x$ and that it pertains to $\sfN_k$, not $\mcn_k$.}
		\ms\begin{definition} The {\it level-$k$ error rate} $p_k$ is the largest probability with which a particular point belongs to $\sfN_k$:
			\be p_k = \max_{x \in \zz} \sfP\(x\in \sfN_k\),\ee 
			with $p_0 = p$ the number appearing in \eqref{pbound}. 
		\end{definition}\ms 
		Note that this definition differs in a few ways from the analogous definitions in \cite{ccapuni2021reliable,gacs2001reliable}. 
		
		When proving that level-$k$ errors become increasingly rare at large $k$, it will prove to be slightly more convenient to work with the isolated noise sets $\mcn_k$, rather than the clustered noise sets $\sfN_k$. Results about the sparsity of isolated noise can then be transferred to results about the sparsity of clustered noise using the following simple fact: 
		\ms\begin{proposition}
			$\sfN_k\subset \mcn_k$ for all $k$. 
		\end{proposition}
		\begin{proof}
			We argue by induction. The base case is trivial since $\sfN_0 = \mcn_0$ by definition. Now suppose $\sfN_k \subset \mcn_k$ for some $k>0$.  All $(\b n^k,\g n^{k+1})$-isolated points of $\sfN_k \cap \mcn_k$ are contained in $\sfE_k$ (as each of them defines a $k$-cluster due to the definitions in \eqref{rbdefs}), and hence are deleted from $\sfN_k$ when passing to $\sfN_{k+1}$. These isolated points are the only elements of $\sfN_k \cap \mcn_k$ that are deleted from $\mcn_k$ when passing to $\mcn_{k+1}$. Thus $\mce_k \cap \sfN_k \subset \sfE_k$, implying $\sfN_{k+1} \subset \mcn_{k+1}$. 
		\end{proof}
		An immediate consequence of this is that the level-$k$ error rate is upper bounded by the largest probability for a particular point to be contained in $\mcn_k$: 
		\be p_k \leq \max_{x\in\zz}\sfP\( x \in \mcn_k \).\ee

		Our goal in the following is to prove that the $p_k$ decay very rapidly with $k$, as long as $p$ is small enough: 
		\ms\begin{theorem}[sparsity bound]\label{thm:sparse}
			If $(\b,\g,n)$ satisfy the inequalities \eqref{paramineqs},	
			then the level-$p$ error rate satisfies 
			\be\label{simplethresh} p_k \leq \frac1{2b_0n^{k+1}}\(\frac p{p_*}\)^{2^k},  \ee 
			where the parameter $p_*$ is bounded as 
			\be \label{pstarbound} p_* \geq \frac1{2b_0n}.\ee 
		\end{theorem}
		
		Since $\sfN$ is $p$-bounded, the scaling $p_k \sim p^{2^{k}}$ suggests that to prove the above theorem, we should show that for a fixed point $x \in \sfN$, each event (viz. each subset of $\sfN$) implying $x \in \sfN_k$ must contain at least $2^k$ points. For technical reasons, it will be more convenient to show this for $\mcn_k$ instead of $\sfN_k$: 
		\ms\begin{lemma}\label{lem:minAsize}
			Fix $x \in \sfN$, suppose that $(\b,\g,n)$ satisfy the inequalities \eqref{paramineqs}, and let $M_k(x)$ denote the set of minimally-sized subsets of $\sfN$ which imply $x \in \mcn_k$.\footnote{In other words: if $x \in \mcn_k$, then there must exist a certain number of other points in $\sfN$ (for example, $\mcn_{k-1}$ must contain more points than just $x$). $M_k(x)$ is the set of minimally-sized sets such that $x \in \mcn_k$ implies that at least one element of $M_k(x)$ is contained in $\sfN$.} 
			Then for all $A_k(x) \in M_k(x)$, 
			\be |A_k(x)|  = 2^{k}.\ee 
		\end{lemma} 
		
		To prove this, we first need a proposition putting an upper bound on the size of the minimally-sized events that can imply $x \in \mcn_k$: 
		\ms\begin{proposition} 
			For all events $A^{\min}_k(x) \subset M_k(x)$, 
			\be A^{\min}_k(x) \subset B_x( \a n^k),\ee 
			where $\a$ is any number satisfying $\g/(1-1/n) < \a \leq \b /2$. 
		\end{proposition} 
		
		\begin{proof} 
			The claim is obviously true for $k=0$, since there is only one minimal event in $M_0(x)$, viz. the set $\{ x\}$. Suppose now that the claim holds for all $m \leq k$, and consider a minimal set $A_{k+1}(x) \subset M_{k+1}(x)$. Now if $x \in \mcn_{k+1}$, then $x \not\in \mce_k$. Thus there must be some point $y \in \mcn_k$ such that $x$ and $y$ are $(\b n^k, \g n^{k+1})$-linked.\footnote{This is where working with $\mcn_k$ becomes simpler than with $\sfN_k$: to shown $x \in \sfN_{k+1}$, we would need to both exhibit a point $y$ that is $(\b n^k, \g n^{k+1})$-linked with $x$, as well as demonstrate the {\it absence} of any $(\b n^k, \g n^{k+1})$-isolated points within a distance $\b n^k$ from $x$. } 
			A minimal set $A_{k+1}(x)$ must contain both $x$ and $y$, and contain enough points to imply that $y\in \mcn_k$ and $x \in \mcn_k$. Therefore $A_{k+1}(x)$ must be of the form $A_k(x) \cup A_k(y)$ for $A_k(r) \in M_k(r)$ ($r=x,y$). Since $A_k(y) \subset B_y(\a n^k)$ by assumption, we have 
			\be A_{k+1}(x) \subset B_x(\g n^{k+1} + \a n^k)\subset B_x(\a n^{k+1}),\ee 
			where the second inclusion follows from our assumption on $\a$. This reasoning is summarized as follows: 
			\be \igpfc{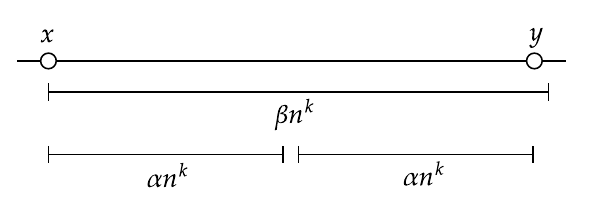} \ee 
		\end{proof} 
		
		In the context of Thm.~\ref{thm:sparse}, the stated assumption that $\b n > 2\g n + \b $ (equation \eqref{paramineqs}) allows us to fix $\a = \b/2$ in what follows. 
		
		We now use this result to prove lemma~\ref{lem:minAsize}: 
		\begin{proof} 
			The logic proceeds in the same manner as in the proof of the previous proposition. If $x \in \mcn_k$ then there must again be a point $y \in \mcn_{k-1}$ such that $x,y$ are $(\b n^{k-1}, \g n^k)$-linked, and a minimal event implying $x \in \mcn_k$ must show that both $x$ and $y$ are members of $\mcn_{k-1}$. Now if $2\a \leq \b$, then by the previous lemma, all members of $M_{k-1}(x), M_{k-1}(y)$ are disjoint. Therefore any member $A^{\min}_k(x)$ of $M_k(x)$ is expressible as a disjoint union $A^{\min}_{k-1}(x) \cup A^{\min}_{k-1}(y)$, so that if $|A^{\min}_k|$ is the size of a minimal subset in $M_k$ (now dropping the $x$-dependence, since $|A^{\min}_k(x)|$ is independent of $x$), we have 
			\be |A^{\min}_k| = 2|A^{\min}_{k-1}|.\ee 
			The claim then follows from $|A^{\min}_0| = 1$. Graphically, this logic is summarized as 
			\be \igpfc{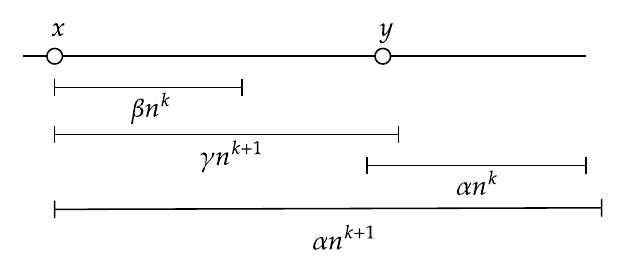} \ee 
		\end{proof} 
		
		We may now complete the proof of Thm.~\ref{thm:sparse}. 
		
		\begin{proof}
			Consider a point $x \in \sfN$. Then the probability that $x \in \sfN_k$ is bounded as 
			\bea \sfP(x \in \sfN_k) & \leq \sfP(x \in \mcn_k) \\
			&  =\sfP \( \bigcup_{A^{\min}_k(x) \in M_k(x)} A^{\min}_k(x)\)  \\ & \leq |M_k| \max_{ M_k} \sfP(A_k^{\min}(x)) \\ 
			& \leq |M_k| p^{2^{l}},\eea 
			where the last inequality follows from our assumption about the noise set $\sfN$ being $p$-bounded, and $|M_k|$ is the number of distinct minimal sets that imply $x \in \mcn_k$ for a given $x \in \zz$. This quantity is easily bounded inductively by again following the logic in the proof of Prop.~\ref{lem:minAsize}: to show that $x \in \mcn_k$ we must show that $x \in \mcn_{k-1}$, together with the existence of a $y \in \mcn_{k-1}$ that is $(\b n^{k-1}, \g n^k)$-linked with $x$. Then we have  
			\be \label{msizebound} |M_k| \leq 2 b_{k-1} |M_{k-1}|^2,\ee 
			and so for $k>0$, 
			\be |M_k| \leq \prod_{j=1}^k (2b_{k-j})^{2^{j-1}},\ee 
			where for $k=0$ we define $|M_0| = 1$. Finally, from the definition of the level-$k$ error rate, we have 
			\be p_k \leq \max_{x \in \zz}\sfP(x\in \mcn_k) \leq p^{2^k }  \prod_{j=1}^k (2b_{k-j})^{2^{j-1}}.\ee 
			The bound in \eqref{simplethresh} then follows from $b_k = b_0n^k$ and the sums $\sum_{j=1}^k 2^{j-1} = 2^k-1$ and $\sum_{j=1}^k (k-j)2^{j-1} = 2^k-k-1$. 
		\end{proof}

		\ss{General dimensions} 
		
		The generalization of the sparsity theorem to an arbitrary number of dimensions is straightforward. We first generalize the notion of isolation and clustering to $D$ dimensions. 
		\ms\begin{definition}[isolation in general dimensions] 
			Define the ball of radius $r$ centered at $\bfx \in \zz^D$ by 
			\be B^{(\sfp)}_\bfx(r) = \{ \bfy \in \zz^D \, : \, |\bfx-\bfy|_\sfp < r\},\ee 
			where the value of $1 \leq \sfp \leq \infty$ will be left unspecified for now (we use $\sfp$ instead of $p$ here to distinguish $\sfp$ from the error rate), and we will omit the basepoint $\bfx$ below when it is unimportant. 
			
			Two points $\bfx,\bfy \in \zz^D$ are {\it $(r,R)$-isolated} if $\bfy \not\in B^{(\sfp)}_\bfx(R) \setminus B^{(\sfp)}_\bfx(r)$; if $\bfy \in  B^{(\sfp)}_\bfx(R) \setminus B^{(\sfp)}_\bfx(r)$ then $\bfx,\bfy$ are {\it $(r,R)$-linked}. 
		\end{definition}\ms

		\ms\begin{definition}[clustering in general dimensions] 
			Let $\sfN\subset\zz^D$ be a noise realization. A $(W,B)$-cluster $C_{(W,B)}$ is a subset of $\sfN$ such that 
			\begin{itemize}
				\item all points in $C_{(W,B)}$ are contained in a ball $B		^{(\sfp)}(W/2)$,
				\item this ball is separated from other points in $\sfN$ along all coordinate directions by a buffer region of size at least $B$: 
				\be \dist(\sfN \setminus C_{(W,B)}, C_{(W,B)}) \geq B,\ee 
				where the distance between two sets $A_1,A_2\subset \zz^D$ is defined as 
				\be \dist(A_1,A_2) = \min_{\bfx\in A_1,\bfy\in A_2}  |\bfx-\bfy|_\sfp.\ee 
			\end{itemize}
		\end{definition}\ms

		The definitions of clustered and isolated noise sets $\sfN_k,\mcn_k$ are then made using these notions of clustering and isolation, and the same parameters $(\g,\b,n)$. The level-$k$ error rate is also defined analogously as 
		\be p_k = \max_{\bfx \in\zz^D} \sfP(\bfx\in \sfN_k).\ee 
		
		The $D$-dimensional analogue of theorem~\ref{thm:sparse} is then an easy generalization of the one dimensional case: 
		\ms\begin{theorem}
			Provided the constraints in \eqref{paramineqs} are satisfied, the level-$k$ error rate satisfies 
			\be p_k \geq  \frac{p_*}{n^{Dk}} \(\frac p{p_*}\)^{2^k},\ee  
			where 
			\be p_* \geq \frac1{(2b_0+w_0)^D}.\ee  
		\end{theorem}
		\begin{proof}
			The proof of the generalization of lemma~\ref{lem:minAsize} goes through without modification compared with the $D=1$ case. The only place where dimensionality enters is in the bound in \eqref{msizebound}, which is replaced by 
			\be |M_k| \leq (|B^{(\sfp)}(\g n^k)| - |B^{(\sfp)}(\b n^{k-1})|) |M_{k-1}|^2,\ee 
			where the prefactor counts the number of places at which to locate a point that establishes the $(\b n^{k-1},\g n^k)$-linkedness of a given point in $\mcn_k$. For general $\sfp$, we may replace the above bound by 
			\be | M_k| \leq |B^{(\infty)}(\g n^k)| |M_{k-1}|^2 = (2\g n^k)^D | M_{k-1}|^2.\ee 		
			Since $p_k \leq \max_{\bfx \in \zz^D}(\bfx \in \mcn_k)$, for general $\sfp$ we then have 
			\be p_k \leq p^{2^k}  \prod_{j=1}^k (2 \g n^{k-j})^{D \cdot 2^{j-1}}, \ee 
			which upon taking the product and using $2\g n= 2b_0+w_0$ gives the desired result. 		
		\end{proof}

		The sparsity of errors at different values of $k$ is easy enough to verify numerically; the result of doing this is shown in fig.~\ref{fig:sparsity_numerics}. One may also check that the bound on $\sfP(\mcn_k\neq\emp)$ is asymptotically rather tight for strength-$p$ i.i.d noise. As an illustration of this, define $P_{\sf max}$ as the probability to have a system-spanning cluster: 
		\be P_{\sf max} = \sfP(\mcn_{\lfloor \log_n(L/2\g)\rfloor} \neq \emp).\ee 
		If the scaling given in the sparsity bound is saturated, we expect 
		\be \label{pmaxguess} P_{\sf max} = (p/p_*)^{\ct(L^{\log_n(2)})}\ee 
		at small enough $p < p_*$. This expectation is numerically seen to be satisfied to a reasonable degree; see fig.~\ref{fig:pmaxfig} for details. 
		
		\begin{figure}
			\centering 
			\includegraphics[width=.48\tw]{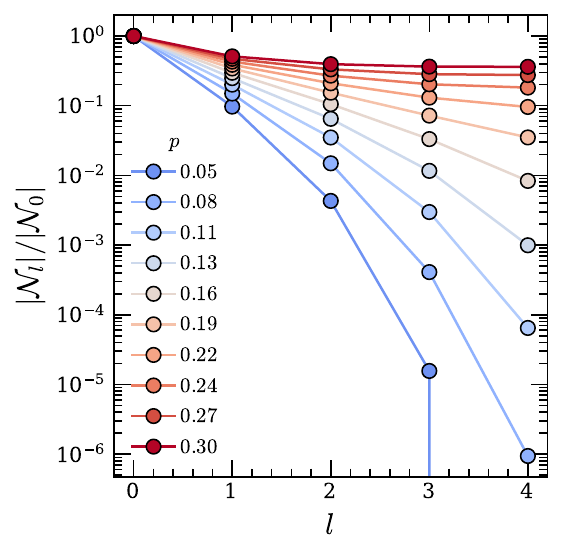} \hfill 	\includegraphics[width=.48\tw]{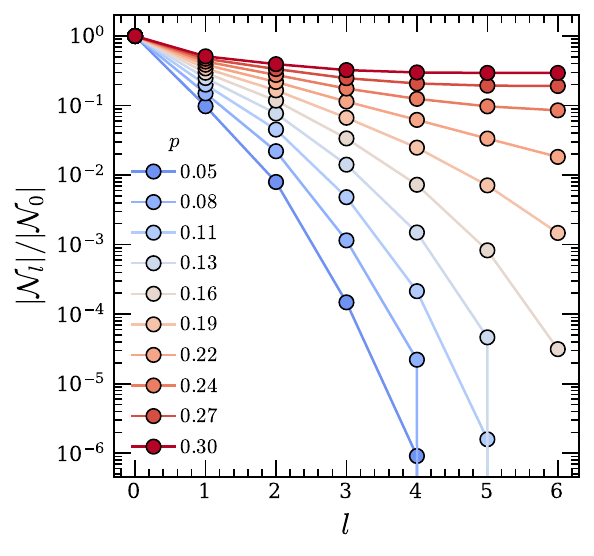} 
			\caption{\label{fig:sparsity_numerics} The size of level-$l$ isolated error sets $\mcn_l$ as a function of $l$, plotted for various values of $p$. In the left panel we have set $n=4,\g = 1,\b = 8/3$, values for which the sparsity theorem promises a threshold, while in the right we have taken $n=3,\g=1,\b=1.5$, which is beyond the regime of validity of the theorem but which is still observed to produce super-exponential decay in $|\mcn_l|/|\mcn_0|$. }
		\end{figure}
	
	\begin{figure}
		\centering \includegraphics[width=.5\tw]{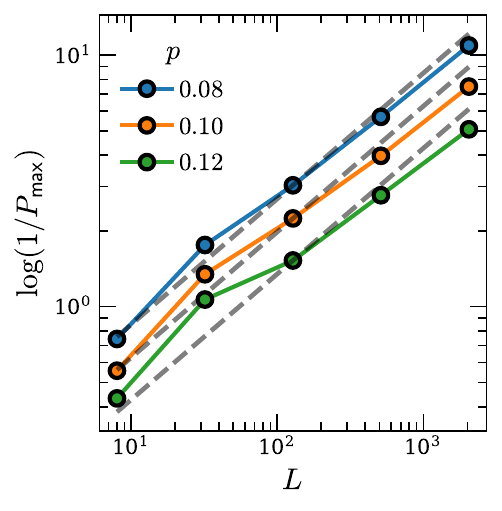} 
		\caption{\label{fig:pmaxfig} $\log (1/P_{\sf max})$ against system size for 1D i.i.d noise of strength $p$. The clustering parameters are chosen as $(\b,\g,n) = (3,1,4)$, values for which the sparsity theorem guarantees a threshold, and all three values of $p$ are taken to be sub-threshold. The dashed lines are fits to \eqref{pmaxguess}.}
	\end{figure}
		
		\section{Open boundary conditions}\label{app:open}
		
		In this appendix, we discuss how the message-passing decoders studied in the main text can be adapted to work with open boundary conditions (e.g. for the surface code). 
		
		Following standard terminology, we will say a segment $S$ of a boundary is {\it rough} if anyon strings may terminate on $S$ without violating any stabilizers.\footnote{We are continuing to specify to only a single sector of anyons (here those associated with violations of $Z$ stabilizers), but the same analysis clearly applies to any logical sector whose excitations are created by string operators.} To keep the discussion simple, we will assume that $\L\cong \zz_L^D$, and will let $\mcb\subset\p\L$ denote the set of sites on rough boundaries.  
		In order for the system to store at least one qubit (or a bit in 1D), $\mcb$ must contain at least two disjoint spatial regions; for simplicity we will fix  
		\be \label{lrboundaries} \mcb = \{ \bfr \in \L \, : \, r^1 = 0\} \cup  \{ \bfr \in \L \, : \, r^1 = L-1\} = \mcb_L \cup \mcb_R\ee 
		although other choices are of course possible. 
		
		\ss{Modifications to the decoder} \label{ss:boundaries_setup}

		We begin by describing how to modify the message passing and feedback rules in the presence of open boundaries. 
		In this context, error correction is performed by allowing anyons to match either with other anyons, or with sites in $\mcb$. Since sites in $\mcb$ act as anyons for the purposes of matching, we simply modify the automaton rule $\mca$ to source messages at all such sites, with a component of $\mcb$ with outward-facing unit normal $\uva$ sourcing messages that propagate along $-\uva$. These rules lead to anyons being attracted either to the anyon nearest to them or to the nearest site in $\mcb$, whichever is closer. This modification is formalized in the following definition:

		\ms\begin{definition}[message passing automaton with open boundaries]
			
			When $\mcb \neq \emp$, the $D$-dimensional automaton rule is replaced by a rule $\mca_{\sf open}^{(v,D)}$ which incorporates the effects of open boundary conditions, defined as 
			$\mca_{\sf open}^{(v,D)} = (\mca_{\sf open}^{(D)})^v$, where 
			\be \mca^{(D)}_{\sf open}(\bfm,\bfsig)[m^{\pm a}_\bfr] = \begin{dcases} 
				1 & 
				\sfA_\ct(C^{(a)}_{\bfr\mp\uva}) = 1 \vee \(\bfr\mp\uva \in \mcb \wedge \bfr \not \in \mcb\) \\ 
				0 & \(\sfA_\ct^\mcb(C^{(a)}_{\bfr\mp\uva}) = 0 \wedge \(m^{\pm a}_{\bfr'} = 0 \, \, \forall \, \, \bfr' \in C^{(a)}_{\bfr\mp\uva}\)\) \vee 
				\bfr \in \mcb  \\ 
				\min_{\bfr' \in C^{(a)}_{\bfr\mp\uva}}\{m^{\pm a}_{\bfr'}\} + 1 & \text{else} 
			\end{dcases}\ee 
			where we take $C^{(a)}_{\bfr\mp\uva} = \emp$ if $\bfr\mp\uva \not\in \L$. 
			
			The feedback operators $\mcf_\bfr^{(D)}$ remain unchanged from the ones used in the case of periodic boundary conditions. 
		\end{definition}\ms 
		
		This definition ensures that an anyon at $(\bfr,t)$ moves towards either the nearest anyon or nearest point of $\mcb$ in $\plf(\bfr,t)$, whichever is closer. 
		
		\ss{Erosion with open boundaries} \label{ss:open_erosion}
		
		We now show that the linear erosion results developed in sec.~\ref{sec:erosion} continue to hold with open boundary conditions. 
		We continue to use the definitions of $(W,B)$-clusters and $k$-clusters in sec.~\ref{ss:sparsity}, as they were made independently of boundary conditions.
		
		\ms\begin{lemma}[linear erosion with open boundaries]\label{lemma:erosion_with_boundaries}
			Consider a $(W,B)$ cluster $C$. 
			The statements of Lemmas~\ref{lemma:linear_erosion}, \ref{lem:const_slowdown} and Corollaries~\ref{lemma:linear_erosion_infty}, \ref{cor:infarb} about the erosion of $C$ continue to hold with open boundary conditions, in the following sense. 
			
			If $C$ is centered at $\bfr$ and $B_\bfr(W/2+B) \cap \mcb = \emp$, erosion of $C$ proceeds identically to the case of periodic boundary conditions. If this is not the case, erosion of $C$ occurs with the following modifications: 
			\begin{enumerate} 
				\item If $t_C^{\sf per}$ is the time by which the anyons initially in $C$ are guaranteed to be annihilated under periodic boundary conditions, then with open boundary conditions, the analogous time $t_C^{\sf open}$ satisfies 
				\be t_C^{\sf open} <d_Ct_C^{\sf per},\ee 
				where $d_C$ is an $O(1)$ constant.
				\item All anyons initially contained in $C$ either annihilate with other anyons initially in $C$, or else are absorbed into a site in $\mcb$. 
			\end{enumerate}
			
		\end{lemma}
		\begin{proof}
			Consider the setup of Lemma~\ref{lemma:linear_erosion}, with a $(W,DW)_1$ cluster $C$ centered at $\bfr$ initialized in a state with trivial messages. If $B^{(1)}_\bfr(3W/2) \cap \mcb = \emp$, the anyons initially in $C$ are guaranteed to all annihilate with other anyons initially in $C$ and we are guaranteed to have $t_C^{\sf open} = t_C^{\sf per}$, for the same reason as in Lemma~\ref{lemma:linear_erosion} (as every anyon in $C$ is closer to another anyon in $C$ than to any site in $\mcb$). Suppose instead that $B^{(1)}_\bfr(3W/2) \cap \mcb \neq \emp$. For simplicity, we will only consider the case when $B^{(1)}_\bfr(DW/2) \cap \mcb$ is contained within a single boundary component of $\mcb$, as this is the only case that will be relevant in subsequent sections. 
			
			\begin{figure}
				\centering 
				\includegraphics[width=.5\tw]{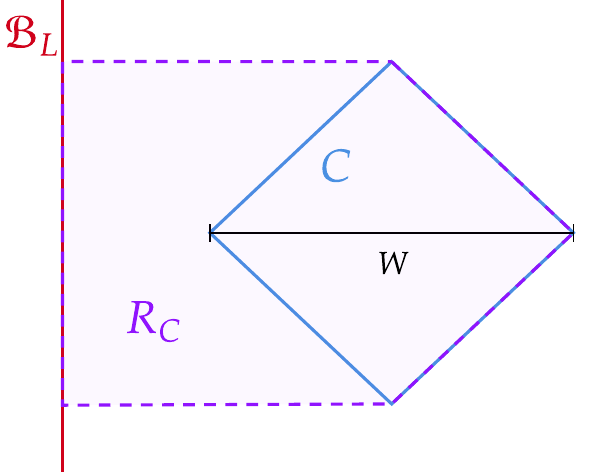} 
				\caption{\label{fig:boundary_erosion} Erosion of a cluster $C$ in 2D near a section of the left rough boundary $\mcb_L$. Anyons are initially contained in the 1-ball of radius $W/2$ indicated in blue, and at all times are confined to the region $R_C$ bounded by the purple dashed lines. The analogous region for clusters near $\mcb_R$ is obtained by reflection.  } 
			\end{figure}
			
			Consider first the case where $B^{(1)}_\bfr(3W/2) \cap \mcb \subset \mcb_L$, and choose coordinates where $C$ is centered at $x\uvx$. We claim that the anyons initially in $C$ are confined at all times to the region (see fig.~\ref{fig:boundary_erosion})
			\be R_C = \bigcup_{\d = 0}^x B^{(1)}_{(x-\d)\uvx}(W/2),\ee 
			and that all anyons in $C$ are eliminated by time $t^{\sf open}_C < x+W(1/2+1/v)$. This follows from an argument very similar as the one in the proof of Lemma~\ref{lemma:linear_erosion}: after time $W/v$, the nearest anyon worldline or point in $\mcb_L$ in the past lightfront of any anyon at $\bfr\in R_C$ is guaranteed to be in $R_C$, and at an $\infty$-norm distance no more than $W$ from $\bfr$. The size $DW$ of the buffer region then ensures that when $t > W/v$, the anyons initially in $C$ are contained in the region 
			\be R_C \cap \{ \bfr \, : \, r^1 < x+W/2-(t - W/v)\} \ee 
			by the same argument as in the proof of \eqref{shrinking}. 
			Since $x < 3W/2$ by assumption, 
			\be t^{\sf open}_C < \frac{2v+1}{v+1}t^{\sf per}_C.\ee 
			
			Now we consider the case where $B^{(1)}_\bfr(3W/2) \cap \mcb \subset \mcb_R$, the analysis of which is unfortunately more tedious since anyon motion is biased away from the boundary. Let us choose coordinates where $C$ is centered at $(L-1-x)\uvx$. We claim that anyons initially in $C$ are confined to the region 
			\be R_C= \bigcup_{\d = 0}^x B^{(1)}_{(L-1-x+\d)\uvx}(W/2),\ee 
			and that all anyons in $C$ are eliminated by time $t^{\sf open}_C < x + W(3/2+1/v)$.
			We show this using a modification of the argument in the proof of Lemma~\ref{lemma:linear_erosion}. 
			To cut down on notation, we define $\sfB = B^{(1)}_{(L-1-x)\uvx}$. 
			
			We first claim that if an anyon enters $\L\setminus \sfB$, it moves along $+\uvx$ until annihilating on a site in $\mcb_R$. To show this, consider the first anyon $a$ from $C$ to enter $\L \setminus \sfB$ (if multiple anyons enter $\L \setminus \sfB$ at the same time, we may resolve the degeneracy by picking the anyon closest to $\mcb_R$ [or one of them, if there are multiple closest anyons]).  
			If $\bfr_t$ is the spatial coordinate of anyon $a$ at time $t$, then if $t$ is the time when $a$ first moves into $\L \setminus \sfB$, we must have $\sfA_\ct(\plf^{+x}(\bfr_t,t)) = 0$ and 
			\be ||\bfr_t - \mcb_R||_\infty < ||\bfr_t - \sfA(\plf(\bfr_t,t))||_\infty,\ee 
			where we have defined the distance from $\bfr$ to a set $S$ as $||\bfr- S|| = \min_{\bfr'\in S}||\bfr-\bfr'||$ (if $S$ is a set of spacetime points, then the minimum is taken by discarding the time coordinates). 
			At time $t+1$, $a$ has moved in the $+\uvx$ direction, hence $||\bfr_{t+1}-\mcb_R||_\infty = ||\bfr_t - \mcb_R||_\infty - 1$. Given this, and given the modified feedback rules in Def.~\ref{def:modified_feedback}, the only way for 
			\be ||\bfr_{t+1} - \sfA(\plf(\bfr_{t+1},t))||_\infty < ||\bfr_{t} - \sfA(\plf(\bfr_{t},t))||_\infty-1\ee 
			is for there to be an anyon in $\plf^{+x}(\bfr_t,t)$ that moves along $-\uvx$ at time $t$, which contradicts $\sfA_\ct(\plf^{+x}(\bfr_t,t)) = 0$. Therefore the distance from $a$ to $\mcb_R$ decreases at least as much as the distance between $a$ and any point in $\sfA(\plf(\bfr_t,t))$. This argument holds for any time $t'>t$, and $a$ therefore moves along $+\uvx$ until annihilating on a site in $\mcb_R$. 
			
			Now consider an anyon $a'$ which leaves $\sfB$ but which is not among the first anyons in $C$ to do so. That $a'$ must also move along $+\uvx$ until annihilating at $\mcb_R$ can be shown by induction using the result of the above paragraph as a base case. Suppose the claim is true for the first $k$ anyons to leave $\sfB$. If $a'$ is the $(k+1)$st anyon to leave $\sfB$ (with degeneracies broken as before) and it does so at $(\bfr'_{t'},t')$, then either  $||\bfr'_{t'}-\sfA(\plf(\bfr'_{t'},t'))||_\infty = ||\bfr'_{t'}-\sfA(\plf^{+x}(\bfr'_{t'},t'))||_\infty$, or $||\bfr'_{t'} - \mcb_R||_\infty < ||\bfr'_{t'}-\sfA(\plf(\bfr'_{t'},t'))||_\infty$. If the latter case holds, then the inductive hypothesis and the properties of the feedback rules in Def.~\ref{def:modified_feedback} mean that 
			\be ||\bfr'_{t'+1}-\sfA(\plf^{+x}(\bfr'_{t'+1},t'+1))||_\infty = ||\bfr'_{t'+1}-\sfA(\plf^{+x}(\bfr'_{t'+1},t'+1))||_\infty\ee 
			provided $\sfA_\ct(\plf^{+x}(\bfr'_{t'+1},t'+1))=1$ (and if not, then we must have $||\bfr'_{t'+1} - \mcb_R||_\infty < ||\bfr'_{t'+1} - \plc(\bfr'_{t'+1},t'+1)||_\infty$, which reduces to the latter case considered above), and that 
			\be ||\bfr'_{t'}-\sfA(\plf^{+x}(\bfr'_{t'},t'))||_\infty \leq ||\bfr'_{t'+1}-\sfA(\plf(\bfr'_{t'+1},t'+1))||_\infty.\ee  
			Therefore $a'$ must continue to move along $\uvx$ at the next time step; since this argument holds for all $t'$, it shows the claim.

			To give a (loose) upper bound on the erosion time, we use the same strategy as in Lemma~\ref{lemma:linear_erosion}. Consider anyons at the slice $\sfS_k = \sfB \cap \{ \bfr \, : \, r^1 = L-1-x+W/2-k\}$ (so that $\sfS_0$ is the single point of $\sfB$ closest to $\mcb_R$, and $\sfS_W$ is the single point furthest from $\mcb_R$). At times $t>t_{\sf move} = W/v$,\footnote{This is a strict overestimate of the first-move time since e.g. $\sfS_0$ is closer to $\mcb_R$ than to $\sfS_W$, on account of our assumption that $\sfB_{(L-1-x)\uvx}(3W/2) \cap \mcb_R\neq\emp$.} the bias in the feedback rules of Def.~\ref{def:modified_feedback} mean that any anyon in $\sfA((\sfS_k,t))$ must move out of $\sfS_k$ at time $t+1$. If an anyon in $\sfS_k$ at time $t > t_{\sf move}$ moves along $\uvx$, then it must move along $\uvx$ until annihilating at $\mcb_R$. Therefore by time $t_{\sf move} + W$, every anyon initially in $C$ must have either been annihilated, or begun moving along $+\uvx$. Since $x<3W/2$ by assumption, this gives a (crude) upper bound on the erosion time of 
			\be t^{\sf open} < x + W(3/2 + 1/v) <\frac{3v+1}{v+1} t_C^{\sf per}.\ee 	
			
			The extension of this result to the scenarios considered in Lemma~\ref{lem:const_slowdown} and Corollaries~\ref{lemma:linear_erosion_infty}, \ref{cor:infarb} proceeds in an identical(ly tedious) fashion, and the details will be omitted. 
		\end{proof}

		\ss{Thresholds with open boundaries} \label{ss:open_threshold}
		
		The fact that erosion proceeds in essentially the same way regardless of boundary conditions means that the scaling of $\ploge,\tdece$ with $p,L$ is also independent of boundary conditions: 
		\ms\begin{corollary}[threshold with open boundary conditions]
			The message-passing decoder defined with open boundary conditions possesses a threshold, and has scalings of $\ploge,\tdece$ subject to the same bounds as in \eqref{plogthm}, \eqref{tdecthm}. 
		\end{corollary}
		\begin{proof}
			The result follows as a direct consequence of Theorem~\ref{thm:precise_offline} and Lemma~\ref{lemma:erosion_with_boundaries}. The latter result means that a logical error will occur only if two anyons $a_1,a_2$ from the same cluster are moved by the feedback operations in such a way that $a_1$ annihilates against a site in $\mcb_L$ and $a_2$ annihilates against a site in $\mcb_R$, which is only possible if the  noise creates at least one $k_L$ cluster. Since the sparsity theorem was formulated independently of boundary conditions,\footnote{With open boundary conditions, the bounds of \eqref{simplethresh},\eqref{pstarbound} are actually very slightly looser than they are with periodic boundary conditions, since for open boundary conditions not all points counted by the RHS of \eqref{msizebound} can actually appear in a cluster if that cluster is sufficiently close to $\mcb$.} we again obtain $\ploge \leq \sfP(\sfN_{k_L}\neq \emp) = (p/p_c)^{O(L^\b)}$, with the same $O(1)$ constant $\b>0$ as in the case of periodic boundary conditions. The scaling of $\tdece$ is similarly $O((\log L)^\eta)$ at $p<p_c$ with the same $O(1)$ constant $\eta >0$: while Lemma~\ref{lemma:erosion_with_boundaries} gave slower erosion than in the periodic case for clusters sufficiently close to $\mcb$, the largest cluster in a noise realization is w.h.p at a distance from $\mcb$  much larger than its size, and hence $\tdece$ is independent of boundary conditions. 
		\end{proof}

		\section{Constant slowdown under desynchronization} \label{app:slowdown} 
		
		In this appendix we provide the proof of Lemma~\ref{lem:slowdown_lemma}. Our general strategy closely follows that of \cite{berman1988investigations}, the differences being that we derive bounds dependent on $t$ (instead of the simulation time), and provide an analogous bound on $\tsimmax$. 
		
		\begin{proof} 
			
			We begin with the bound on $\tsimmax(t)$, which is simpler. For a fixed site $\bfr$ and time $t$, consider the event $\tsim(\bfr,t)= m$. The number of attempted updates at site $\bfr$ by time $t$ must be at least $m$. Let $\d_k$ denote the time interval between the $(k-1)$th and $k$th updates at site $\bfr$, with $k = 1,\dots,m$. The event $\tsim(\bfr,t) = m$ implies the event  $S_m = \sum_{k=1}^m \d_k \leq t$. Fixing a constant $\l>2e$ and letting $m_t$ denote the smallest integer larger than $\l t/\mu$, 
			\bea \sfP(\mu\tsim(\bfr,t) > \l t) = \sum_{m=m_t}^\infty \sfP(\tsim(\bfr,t) = m) \leq \sum_{m=m_t}^\infty \sfP(S_m \leq t).\eea 
			By definition, the $\d_k$ are i.i.d exponentially-distributed random variables with mean $\mu$. The Chernoff bound then gives, assuming $m > t/\mu$, 
			\bea \sfP(S_m \leq t) & \leq \inf_{\a > 0} \exp\(\a t + m \ln \EE[e^{-\a \d_k}]\) = \exp\(-\frac t\mu - m \(\ln \frac{m \mu}t - 1\)\). \eea 
			Since $m > \l t/\mu$ for all terms in the sum and $\l > 2e$, have $\ln(m\mu/t) - 1 > \ln(\l) - 1 > \ln2$, we have 
			\be \sfP(\mu\tsim(\bfr,t) > \l t) < e^{-t/\mu} \sum_{m=m_t}^\infty 2^{-m} = 2 e^{-m_t \ln 2 - t/\mu} \leq 2 e^{-2c_\l t},\ee 
			where we defined 
			\be c_\l = \frac{1+\l\ln2}{2\mu}.\ee 
			Finally, consider the event $\mu\tsimmax(t) > \l t$. This implies that there is at least one site $\bfr$ with $\mu\tsim(\bfr,t) > \l t$, and so the union bound over the $L^D$ different sites in the system gives 
			\be \sfP(\mu\tsimmax(t) > \l t) \leq 2e^{-2c_\l t + D \ln L}.\ee 
			If we assume $t \geq \frac{2D}{c_\l} \ln L$, we then have (dropping the prefactor of $2$ since we are interested in large $L$), 
			\be \label{maxbound} \sfP(\mu\tsimmax(t) > \l t) \leq e^{-c_\l t},\ee 
			which is what we wanted to show. 
			
			We now show the bound on $\tsimmin(t)$, which we do following the general strategy of \cite{berman1988investigations}. As before, fix a spacetime point $(\bfr,t)$ and consider the event $\tsim(\bfr,t) = m$. For the bound on $\tsimmax(t)$, we showed that for large $m$, this event implied that some sum of i.i.d random variables had to be significantly smaller than its mean. For the bound on $\tsimmin(t)$, we will show that for small $m$, this event implies the existence of a sum of i.i.d random variables which must be significantly larger than its mean. 		
			We will do this by forming a set of spacetime points $\sfD(\bfr,t)$ called a ``delay chain''. $\sfD(\bfr,t)$ consists of $\tsim(\bfr,t)$ spacetime points we label as 
			\be \label{delchain} \sfD(\bfr,t) = \{ (\bfr_n,t_n) \, : \, n = 1,\dots,\tsim(\bfr,t)\}.\ee
			To define these points, we need some additional notation. Let $\tau_{\bfr,k}$ denote the $k$th time at which the system attempts an update at site $\bfr$, and define $\d_{\bfr,t}$ as the time interval containing $t$ between attempted updates at $\bfr$: 
			\be \d_{\bfr,t} = \tau_{\bfr,k_{\bfr,t}+1} - \tau_{\bfr,k_{\bfr,t}},\qq k_{\bfr,t} = \max \{ k\, : \, t > \tau_{\bfr,k}\}.\ee 
			
			The points in $\sfD(\bfr,t)$ are chosen in the following way. First, we fix $(\bfr_1,t_1) = (\bfr,t)$.  Given $(\bfr_n,t_n)$, we fix $(\bfr_{n+1},t_{n+1})$ as follows. Let 
			\be \tau' = \tau_{\bfr_n,k_{\bfr_n,t_n}}-0^+\ee  
			be right before the last attempted update time at site $\bfr_n$. Then either the proposed update at $(\bfr_n,\tau')$ was accepted, in which case $\tsim(\bfr_n,t_n) = \tsim(\bfr_n,\tau') +1$, or it was rejected, in which case $\tsim(\bfr_n,t_n) = \tsim(\bfr_n,\tau')$. In the first case, we fix $(\bfr_{n+1},t_{n+1}) = (\bfr_n,  \tau')$. In the latter case, there must have been at least one site $\bfr'$ neighboring $ \bfr_n$ such that $\tsim(\bfr',\tau') = \tsim(\bfr_n,\tau')-1$, whose smaller simulation time prevented the proposed automaton update at $(\bfr_n,\tau')$ from succeeding (by virtue of having $\lag_{\bfr\ra\bfr'} = +1$). We then choose one such site $\bfr'$---with different choices of $\bfr'$ defining different delay chains---and then fix $(\bfr_{n+1},t_{n+1}) = (\bfr', t_n)$.
			
			The event $\mu\tsim(\bfr,t) < \g t$ implies the existence of at least one delay chain with $(\bfr_1,t_1) = (\bfr,t)$. Let ${\sf Paths}(m)$ denote the set of length-$m$ nearest-neighbor paths on the lattice, and let $\sfD(\bfr,t)$ have path $P\in {\sf Paths}(m)$ if the spatial coordinates of $\sfD(\bfr,t)$ is equal to $P$ as an ordered set. Letting $z$ be an upper bound on the number of nearest neighbors of any given site, we have $|{\sf Paths}(m)|  = z^m$. Defining $m_t$ as the largest integer less than $\g t / \mu$, we use the above facts to write  
			\bea  \sfP(\mu \tsim(\bfr,t) < \g t) & = \sum_{m=0}^{m_t} \sfP(\tsim(\bfr,t) = m) \\ 
			&  \leq \sum_{m=0}^{m_t}\sfP\( \bigcup_{P\in {\sf Paths}(m)} \text{$\exists$ delay chain with path $P$} \) \\ &
			\leq \sum_{m=0}^{m_t} z^m \sfP\(\text{$\exists$ some length-$m$ delay chain $\sfD(\bfr,t)$}\).
			\eea 
			
			Recall that the time delays $\d_{\bfr,t}$ are i.i.d random variables with mean $\mu$. If a delay chain $\sfD(\bfr,t)$ of length $m$ exists, the manner in which the temporal points $\{t_n\}$ of a delay chain were chosen then implies  
			\be \label{tbound} t = \sum_{n=1}^{m} (t_n - t_{n+1}) \leq  \sum_{n=1}^{m} \d_{\bfr_n,t_n} =S_m, \ee 
			where as before $S_m$ is the sum of $m$ exponentially distributed i.i.d random variables with mean $\mu$, and the inequality is saturated only if $\bfr_n = \bfr$ for all $n$ and $t = \tau_{\bfr,k_{\bfr,t}+1}-0^+$. Thus 
			\be \sfP(\mu\tsim(\bfr,t) < \g t) \leq \sum_{m=0}^{m_t} z^m \sfP(S_m\geq t).\ee 
			The remaining steps are the same as for the bound on $\tsimmax(t)$. Again using the Chernoff bound, we obtain 
			\bea \sfP(\mu\tsim(\bfr,t) < \g t) & \leq e^{-t/\mu} \sum_{m=0}^{m_t} \exp\( m \[1 + \ln\(\frac{zt}{m\mu}\)\]\) . \eea 
			Since the expression in the exponent is monotonically increasing in $m$ for $m \leq m_t < zt/\mu$, we may write 
			\bea  \sfP(\mu\tsim(\bfr,t) < \g t) & < m_t e^{-t/\mu + m_t(1 + \ln(zt/m_t\mu))} \leq \exp\(-\frac t\mu \[1 - \g\( 1 + \ln\frac z\g \) - \frac{\ln (\g t / \mu)}{t/\mu} \]\). \eea 
			Using $\g \ln 1/\g \leq 1/e$ for $\g \in [0,1]$ and defining  
			\be d_\g = \frac{1 - e^{-1} - \g ( 1 + \ln z)}{2\mu},\ee 
			which is positive and $O(1)$ as long as e.g. $\g < \g_*  = (1-e\inv)/(2(1+\ln z))$, 
			we have, after a union bound over all $L^D$ sites, 
			\be \sfP(\mu\tsimmin(t) < \g t) < e^{-2d_\g t + D \ln L - \ln(\g t/\mu)}.\ee 
			Thus if $t \geq \frac{2D}{d_\g} \ln L$, we can drop the last $\ln (\g t / \mu)$ term and write 
			\be\label{minbound} \sfP(\mu\tsimmin(t) < \g t) < e^{-d_\g t}.\ee 
			
			The (weaker) bounds of the Lemma statement then follow from \eqref{minbound}, \eqref{maxbound} upon choosing $c = \min(c_{2e},d_{\g_*})$. 
		\end{proof}

		Although the bounds derived above are sufficient for our purposes, without much extra work one may obtain a stronger bound on $\tsimmax(t)$: since $\lag_{\bfr'\ra\bfr}$ is at most $1$, the event $\mu\tsimmax(t) > \a\l t$, $\a>1$ implies the existence of at least $\sim [\l(\a-1)t]^D$ sites with $\tsim(\bfr,t) > \l t$. This then gives 
		\be \sfP(\mu\tsimmax(t) > \a \l t) < e^{-\wt ct^{D+1}} \ee 
		for some constant $\wt c$, as long as $t = \O((\ln L)^{1/(1+D)})$. An analogous improvement for $\tsimmin(t)$ is less straightforward since the delay chains at the additional $\sim[\l(\a-1)t]^D$ sites needn't necessarily be independent.

		\section{Comments on PDE-based approaches}\label{app:no_pdes}
		
		In this appendix we briefly discuss PDE-based automata similar to the field-based decoders of Herold et al.~\cite{herold2015cellular,herold2017cellular}, and provide a few comments on the self-interaction problem identified by the above authors.

		Let us briefly recall the general construction of \cite{herold2015cellular,herold2017cellular}. We write the classical variables on each site as $\mch_{\sf cl} = {\sf span}\{ |\phi_\bfr,\s_\bfr\ran \}$, where $\phi_\bfr$ is a rotor field valued in $[0,2\pi)$,\footnote{Although a truncation to a field which only stores $\log \log L$ different values is permissible for the same reasons as in the main text.}  and take $\phi$ to be updated according to some appropriate discretization of a differential equation in which anyons act as sources. Letting $\phi_{\bfr,t}$ be the value of $\phi_\bfr$ at time step $t$, the automaton rule $\mca$ is then designed to ensure that 
		\be \mcl \phi_{\bfr,t} = \sfA(\bfr,t),\ee 
		where $\mcl$ is an appropriate discretization of a linear differential operator and $\sfA$ is defined in Def.~\ref{def:lightfronts2d} ($\mcl$ is restricted to be linear in order for signals received  by different anyons to superimpose). The anyons are then updated with feedback operators that move an anyon at $\bfr$ in the direction of the strongest field gradient: 
		\be  \mcf_\bfr = \bot_{a = 1}^D X_{\lan \bfr,\bfr+\uva\ran}^{\s_\bfr f_\bfr^a \vee \s_{\bfr+\uva} (-f_{\bfr+\uva}^a)}\ee 	
		as before, where now 
		\be f^a_\bfr  = \d_{a,a_\bfr} {\sf sgn}( \D_{a_\bfr} \phi_\bfr),\qq a_\bfr = {\sf argmax}_a \{ |\D_a \phi_\bfr|\}.\ee

		Refs.~\cite{herold2015cellular,herold2017cellular} take $\mcl= \p_t - \D^2$ to describe overdamped Langevin dynamics. As explained in these works, this leads to several undesirable features. First, $\mcl$ generates an attractive force between anyons scaling as $|\D \phi| \sim 1/r^{D-1}$. The techniques of sec.~\ref{sec:offline} applied to this problem strongly suggest that these decoders lack thresholds (as was identified numerically in $D=2$ in Ref.~\cite{herold2015cellular}), while a force weaker than $1/r^{D}$ is essentially equivalent to nearest-anyon interactions, and produces a threshold. Since the decay of the force is fixed by $D$ and does not improve when more spatial derivatives are added, attaining a threshold requires that $\phi$ be allowed to propagate in one or more additional dimensions (albeit dimensions whose size need diverge only polylogarithmically in $L$).
		
		From the perspective of this work, a more serious problem with choosing $\mcl = \p_t - f(\D)$ (for any function $f$) is the self-interactions of anyons with the fields they themselves produce. When an anyon moves, it leads to a violation of $\mcl \phi = \sfA$ when it ``leaves behind'' the divergence of $\phi$ that it sourced. If $\phi$ updates at an $O(1)$ speed $v$---meaning that $v$ updates of $\phi$ occur for each application of the feedback $\mcf$---this anyon will then be attracted to the lattice site it most recently occupied, preventing anyon matching from occurring. The detailed analysis of Ref.~\cite{herold2015cellular} showed that this problem can be overcome only by taking $v$ to diverge at least polylogarithmically with $L$, so that $\phi$ very rapidly re-equilibrates after each anyon movement.  
		
		Our message-passing decoders solve this problem by ensuring that the ``field'' sourced by an anyon travels from it along a ballistically-propagating wavefront, which moves at a velocity larger than the speed of anyon motion (see the discussion around Def.~\ref{def:1d_msg_passing}). This property can also be attained by field-based decoders by modifying the choice of $\mcl$, but only in certain dimensions.  To see this, note that the self-interaction problem is avoided only if the Greens function for $\mcl$ has support only on the surface of the past lightfront, and not on its interior.  Our only option is then to take $\mcl = \p_t^2 - \D^2$, as wave equations are the only linear PDEs with this property \cite{friedlander1975wave}.\footnote{Taking $\mcl$ to be nonlinear in $\phi$ is undesirable for other reasons, and in any case the existence of nonlinear PDEs with the desired lightfront property seems rather unlikely.} However, this property in fact holds only for wave equations in {\it odd} dimensions. Using the wave equation to update $\phi$ therefore allows us to remove the self-interaction problem in the setting where an extended 3D field architecture is used to error correct a 2D toric code, but not in more general scenarios.\footnote{Architectures which embed multiple fields evolving under the 1D wave equation into a higher dimensional space do not offer a way out, since the solutions to the 1D wave equation involve signals whose strength does not decay in time.}

		\section{Decoding with a constant number of bits per site}\label{app:constant_bits} 
		
		In this appendix we give a schematic description of decoders that use only an $O(1)$ number of bits per site, as opposed to the $O(\log \log L)$ bits required by the construction in the main text. We will not prove that the decoders defined below possess a threshold, and leave a careful analysis of their performance for future work. 
		
		It is clear that we cannot simply cap the message fields $m^{\pm a}_\bfr$ of Def.~\ref{def:Dd_msg_passing} off at a maximum value of $\ell = O(1)$: this does not work because anyons will not be able to distinguish between signals coming from anyons of different distances beyond $\ell$, which is needed for decoding to work. Instead, our approach will be to find a way of encoding decaying message strengths  by ensuring that the probability for a site $\bfr$ to be occupied by a message decays as a function of distance between $\bfr$ and the anyon nearest to $\bfr$. We first explain the construction in 1D, where it is simpler. 
		
		\ss{One dimension} 
		
		First let us understand how to set up decaying message strengths within a discrete state space. We will parametrize the control variables at site $r$ as
		\be m_r = \{ n^+_r, n^-_r\}, \qq n^\pm_r \in \{0,1\}.\ee 
		Messages with $n^\pm_r = 1$ will be viewed as particles undergoing a certain type of annihilating asymmetric exclusion process defined with the following probabilistic automaton rule: 
		\ms\begin{definition}[probabilistic binary message passing, 1D]\label{def:1d_asep}
			Fix a number $0<q<1$ and in integer $m>1$. The probabilistic automaton rule $\mca_{m,q}^{(1)}$ uses a collection of i.i.d random variables $\xi_r$, each drawn from the uniform distribution on $[0,1]$ and refreshed at each time step, and is defined as 
			\be \mca_{m,q}^{(1)}(\bfm,\bfsig)[n^\pm_r] = \begin{dcases} 1 &  n^\pm_{r\mp1}(1-n_r^\pm) = 1  \wedge \xi_r < q\\
				0  &  \(n_r (1-n^\pm_{r\pm1})= 1 \wedge \xi_r < q  \) \vee \(\prod_{r' = r}^{r\pm m}n^\pm_{r'} = 1\) \\ 
				n_r^\pm & \text{else} 
			\end{dcases} . 
			\ee 
		\end{definition}\ms 
		Under this automaton rule, $n^\pm$ messages move in the $\pm$ direction at each step independently with probability $q$, and are annihilated if they find themselves at the left (right) end of a chain of particles on $m$ consecutive sites. 
		As in sec.~\ref{sec:msg_passing}, the decoding rules are then defined using $\mca^{(v,1)}_{m,q} = (\mca^{(1)}_{m,q})^v$ and feedback operators defined (for now) analogously to those of Def.~\ref{def:1d_msg_passing}. 
		
		Consider a single well-isolated anyon fixed at the origin $r=0$, and examine the expected particle density $\r(r>0) = \lan n^+_r\ran$ at long times $t \gg r$. 
		The annihilation in Def.~\ref{def:1d_asep} leads to an equilibrium density that decays with $r$ in a manner fixed by $m$. For all choices of $m>0$, numerics show that $\r(r) \sim r^{-\a_m}$, where the exponent $\a_m$ decreases with $m$; see fig.~\ref{fig:asep} for several choices of $m$.\footnote{This can be anticipated from mean field: there we would (schematically) have $\p_t \r \sim \p_r( v   +D \p_r)\r -\r^m$, with a power law scaling suggesting $\r \sim r^{-\a_m}$ with $\a_m = 1/(m-1)$; the real power law is (as expected) slower than this.}
		For example, at $m=2$ one finds $\r(r) \propto \frac1{\sqrt r}$. 
		This can in fact be shown rigorously using methods developed for analyzing systems undergoing 1D ballistic annihilation \cite{biswas2021ballistic}, but since it is rather tangential to the main point of this paper we will not go through the details. 
		
		\begin{figure} 
			\centering \includegraphics[width=.45\tw]{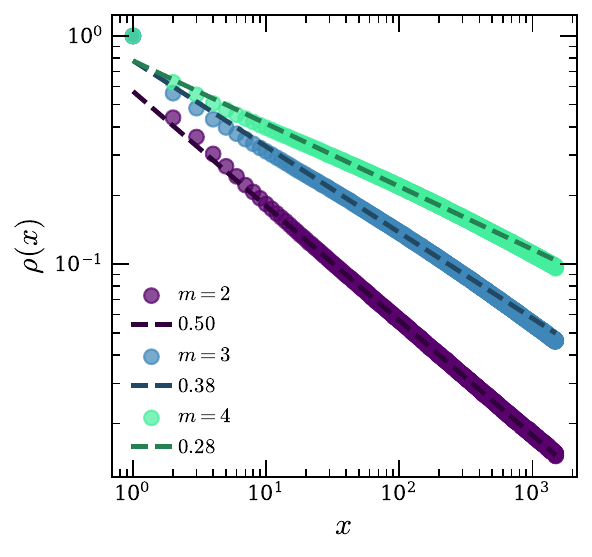}
			\caption{\label{fig:asep} Steady state particle density $\r(x)$ at a distance $x$ from an unmoving particle source for ASEP with different $m$-body annihilating interactions. } 
		\end{figure}

		This approach allows us lets us set up power-law interactions between anyons of the form $1/r^{\a}$ with $\a \leq 1/2$. Getting larger values of $\a$ can be done by coupling multiple copies of this chain together: in an $\sfn$-species model where each species independently undergoes the annihilating exclusion process defined by $\mca_{m,q}$, 
		we have (in schematic notation) 
		\be \prod_{a=1}^\sfn \r_a(r) \sim \frac1{r^{\sfn\a_m}},\ee 
		so that by using the product of the messages across all flavors to control anyon motion, the interaction between anyons can be designed to yield an arbitrarily quickly-decaying power law (exponentially decaying interactions are of course trivial to engineer, simply by letting each message ``decay'' with a constant probability at each time step). 
		
		While we did not explicitly discuss models with power-law interactions in sec.~\ref{sec:offline}, one may use the same logic to show that decoders with retarded attractive interactions of $1/r^{\a > D}$ have thresholds, provided the interactions are set up using the multi-flavor framework of sec.~\ref{ss:power_law_construction}. If we use the same feedback operators as in Def.~\ref{def:1d_msg_passing}---so that an anyon at $r$ moves only when $n^\pm_r\neq0$---a proof along the same lines does not directly go through, as in this case the erosion of a $(W,\infty)$ cluster takes time $t_W = W^\b$ with $\b > 1$. This can be heuristically seen as follows. Suppose a message source is fixed at the origin, and consider a well-isolated anyon at position $W>0$. The anyon undergoes biased motion with a velocity $-v(r)$, where $v(r) \sim \lan n^+_r\ran$. If the latter decays as $1/r^\a$, then 
		\be t_W \sim \int_0^W \frac{dr}{|v(r)|} \sim W^{1+\a},\ee 
		which is indeed super-linear in $W$. This problem can be circumvented by adding an additional variable $d_r \in \{\pm1\}$ at each site with $d_r\neq 0$ only if $\s_r = -1$, with $d_r$ keeping track of the direction in which the anyon at $r$ last moved. Linear erosion can then be restored by taking the feedback operators to move an anyon at $r$ in the direction of $d_r$ in the case where $n_r^\pm = 0$. In the present work we will not go through all the details required to make this fix rigorous.

		\ss{Higher dimensions}
		
		We now generalize the above scheme to higher dimensions. The hardest part is finding a biased annihilating exclusion process that gives a signal which decays uniformly as a function of distance in all directions; naive generalizations of Def.~\ref{def:1d_asep} which attempt to bias message motion in the radial direction are hard to define in a translation invariant way. Instead, our strategy will be to make use of the 1D result, employing multiple flavors of 1D random-walking messages along different directions and making them interact in a way which produces a symmetric expected message density. 
		
		Unlike in sec.~\ref{sec:msg_passing}, we will find it helpful to use message passing rules with a lightfront determined by the $1$-norm. To simplify the notation slightly, we will specify to $D=2$; the extension to higher dimensions is essentially trivial. The message variables on each site will be written as 
		\be m_\bfr = \{n^{\sfd,c}_r \, : \, \sfd\in\{1,2,3,4\}, \, c\in \{1,2\}\}\ee 
		where each $n^{\sfd,c}_\bfr \in \{0,1\}$. Messages with $\sfd=1$ are designed to propagate in the $x\geq0,y>0$ quadrant, those with $\sfd=2$ to propagate in the $x<0,y\geq 0$ quadrant, and so on. 
		
		\ms\begin{definition}[probabilistic binary message passing, 2D]
			Fix an integer $m>0$ and a number $0<q<1$. The message fields $n^{\sfd,c}_\bfr$ are emitted from anyons, and are designed to propagate along $\uvc$ in the $\sfd$ quadrant of the plane. For $\sfd=1$, the update rules are 
			\bea \mca_{m,q}^{(2)}(\bfm,\bfsig)[n^{1,x}_{(x,y)}] & = \mca^{(1)}_{m,q}(\bfm|_{L_{y}},\bfsig|_{L_y})[n^+_x] \\ 
			\mca_{m,q}^{(2)}(\bfm,\bfsig)[n^{1,y}_{(x,y)}] & = \begin{dcases} 
				\mca^{(1)}_{m,q}(\bfm|_{L_{x}},\bfsig|_{L_x})[n^+_y] & n^{1,x}_{(x,y-1)} = 0 \\ 
				1 & \text{else} \end{dcases}.
			\eea
			where $L_x = \{ \bfr \, : \, r^1 = x\}$ and similarly for $L_y$. 
			In words: $n^{1,c}$ messages are sourced at anyon locations and move along $\uvc$ according to the exclusion process of Def.~\ref{def:1d_asep}, and $x$-traveling walkers source $y$-traveling walkers. $\mca_{m,q}^{(2)}$ updates the fields in the remaining three quadrants analogously, and the full automaton rule is obtained by taking $v$ repetitions of $\mca_{m,1}^{(2)}$. 
		\end{definition}\ms 
		
		The properties of the 1D automaton rule ensure that if an anyon source is fixed at the origin, the time-averaged total message density scales as 
		\be \r(\bfr) = \sum_{a,b,c} \lan n^{ab,c}_\bfr \ran \sim \frac1{||\bfr||_1^{\a_m}}.\ee 
		Since the messages propagate according to the $1$-norm, the feedback operators of Def.~\ref{def:Dd_msg_passing} are modified to move anyons along the directions normal to the surfaces of $1$-balls. As in 1D, superlinear erosion times can be avoided by adding additional variables to $\mch_{\sf cl}$ which keep track of the directions along which anyons have moved. A detailed analysis is left for future work.

	\end{appendix} 
	
	\bibliographystyle{abbrev}
	\bibliography{offline_decoders}


\begin{thebibliography}{64}
\ifx \bisbn   \undefined \def \bisbn  #1{ISBN #1}\fi
\ifx \binits  \undefined \def \binits#1{#1}\fi
\ifx \bauthor  \undefined \def \bauthor#1{#1}\fi
\ifx \batitle  \undefined \def \batitle#1{#1}\fi
\ifx \bjtitle  \undefined \def \bjtitle#1{#1}\fi
\ifx \bvolume  \undefined \def \bvolume#1{\textbf{#1}}\fi
\ifx \byear  \undefined \def \byear#1{#1}\fi
\ifx \bissue  \undefined \def \bissue#1{#1}\fi
\ifx \bfpage  \undefined \def \bfpage#1{#1}\fi
\ifx \blpage  \undefined \def \blpage #1{#1}\fi
\ifx \burl  \undefined \def \burl#1{\textsf{#1}}\fi
\ifx \doiurl  \undefined \def \doiurl#1{\url{https://doi.org/#1}}\fi
\ifx \betal  \undefined \def \betal{\textit{et al.}}\fi
\ifx \binstitute  \undefined \def \binstitute#1{#1}\fi
\ifx \binstitutionaled  \undefined \def \binstitutionaled#1{#1}\fi
\ifx \bctitle  \undefined \def \bctitle#1{#1}\fi
\ifx \beditor  \undefined \def \beditor#1{#1}\fi
\ifx \bpublisher  \undefined \def \bpublisher#1{#1}\fi
\ifx \bbtitle  \undefined \def \bbtitle#1{#1}\fi
\ifx \bedition  \undefined \def \bedition#1{#1}\fi
\ifx \bseriesno  \undefined \def \bseriesno#1{#1}\fi
\ifx \blocation  \undefined \def \blocation#1{#1}\fi
\ifx \bsertitle  \undefined \def \bsertitle#1{#1}\fi
\ifx \bsnm \undefined \def \bsnm#1{#1}\fi
\ifx \bsuffix \undefined \def \bsuffix#1{#1}\fi
\ifx \bparticle \undefined \def \bparticle#1{#1}\fi
\ifx \barticle \undefined \def \barticle#1{#1}\fi
\bibcommenthead
\ifx \bconfdate \undefined \def \bconfdate #1{#1}\fi
\ifx \botherref \undefined \def \botherref #1{#1}\fi
\ifx \url \undefined \def \url#1{\textsf{#1}}\fi
\ifx \bchapter \undefined \def \bchapter#1{#1}\fi
\ifx \bbook \undefined \def \bbook#1{#1}\fi
\ifx \bcomment \undefined \def \bcomment#1{#1}\fi
\ifx \oauthor \undefined \def \oauthor#1{#1}\fi
\ifx \citeauthoryear \undefined \def \citeauthoryear#1{#1}\fi
\ifx \endbibitem  \undefined \def \endbibitem {}\fi
\ifx \bconflocation  \undefined \def \bconflocation#1{#1}\fi
\ifx \arxivurl  \undefined \def \arxivurl#1{\textsf{#1}}\fi
\csname PreBibitemsHook\endcsname

\bibitem[\protect\citeauthoryear{Campbell}{2024}]{campbell2024series}
\begin{barticle}
\bauthor{\bsnm{Campbell}, \binits{E.}}:
\batitle{A series of fast-paced advances in quantum error correction}.
\bjtitle{Nature Reviews Physics}
\bvolume{6}(\bissue{3}),
\bfpage{160}--\blpage{161}
(\byear{2024})
\end{barticle}
\endbibitem

\bibitem[\protect\citeauthoryear{AI and Collaborators}{2025}]{google_qec}
\begin{barticle}
\bauthor{\bsnm{AI}, \binits{G.Q.}},
\bauthor{\bsnm{Collaborators}}:
\batitle{Quantum error correction below the surface code threshold}.
\bjtitle{Nature}
\bvolume{638}(\bissue{8052}),
\bfpage{920}--\blpage{926}
(\byear{2025})
\end{barticle}
\endbibitem

\bibitem[\protect\citeauthoryear{Reilly}{2019}]{reilly2019challenges}
\begin{bchapter}
\bauthor{\bsnm{Reilly}, \binits{D.}}:
\bctitle{Challenges in scaling-up the control interface of a quantum computer}.
In: \bbtitle{2019 IEEE International Electron Devices Meeting (IEDM)},
pp. \bfpage{31}--\blpage{7}
(\byear{2019}).
\bcomment{IEEE}
\end{bchapter}
\endbibitem

\bibitem[\protect\citeauthoryear{Brennan et~al.}{2025}]{brennan2025classical}
\begin{botherref}
\oauthor{\bsnm{Brennan}, \binits{J.C.}},
\oauthor{\bsnm{Barbosa}, \binits{J.}},
\oauthor{\bsnm{Li}, \binits{C.}},
\oauthor{\bsnm{Ahmad}, \binits{M.}},
\oauthor{\bsnm{Imroze}, \binits{F.}},
\oauthor{\bsnm{Rose}, \binits{C.}},
\oauthor{\bsnm{Karar}, \binits{W.}},
\oauthor{\bsnm{Stanley}, \binits{M.}},
\oauthor{\bsnm{Heidari}, \binits{H.}},
\oauthor{\bsnm{Ridler}, \binits{N.M.}}, et al.:
Classical interfaces for controlling cryogenic quantum computing technologies.
arXiv preprint arXiv:2504.18527
(2025)
\end{botherref}
\endbibitem

\bibitem[\protect\citeauthoryear{Liyanage et~al.}{2024}]{liyanage2024fpga}
\begin{botherref}
\oauthor{\bsnm{Liyanage}, \binits{N.}},
\oauthor{\bsnm{Wu}, \binits{Y.}},
\oauthor{\bsnm{Tagare}, \binits{S.}},
\oauthor{\bsnm{Zhong}, \binits{L.}}:
Fpga-based distributed union-find decoder for surface codes.
IEEE Transactions on Quantum Engineering
(2024)
\end{botherref}
\endbibitem

\bibitem[\protect\citeauthoryear{deMarti iOlius
  et~al.}{2024}]{demarti2024decoding}
\begin{barticle}
\bauthor{\bsnm{iOlius}, \binits{A.}},
\bauthor{\bsnm{Fuentes}, \binits{P.}},
\bauthor{\bsnm{Or{\'u}s}, \binits{R.}},
\bauthor{\bsnm{Crespo}, \binits{P.M.}},
\bauthor{\bsnm{Martinez}, \binits{J.E.}}:
\batitle{Decoding algorithms for surface codes}.
\bjtitle{Quantum}
\bvolume{8},
\bfpage{1498}
(\byear{2024})
\end{barticle}
\endbibitem

\bibitem[\protect\citeauthoryear{Ziad et~al.}{2024}]{ziad2024local}
\begin{botherref}
\oauthor{\bsnm{Ziad}, \binits{A.B.}},
\oauthor{\bsnm{Zalawadiya}, \binits{A.}},
\oauthor{\bsnm{Topal}, \binits{C.}},
\oauthor{\bsnm{Camps}, \binits{J.}},
\oauthor{\bsnm{Geh{\'e}r}, \binits{G.P.}},
\oauthor{\bsnm{Stafford}, \binits{M.P.}},
\oauthor{\bsnm{Turner}, \binits{M.L.}}:
Local clustering decoder: a fast and adaptive hardware decoder for the surface
  code.
arXiv preprint arXiv:2411.10343
(2024)
\end{botherref}
\endbibitem

\bibitem[\protect\citeauthoryear{Wu and Zhong}{2023}]{wu2023fusion}
\begin{bchapter}
\bauthor{\bsnm{Wu}, \binits{Y.}},
\bauthor{\bsnm{Zhong}, \binits{L.}}:
\bctitle{Fusion blossom: Fast mwpm decoders for qec}.
In: \bbtitle{2023 IEEE International Conference on Quantum Computing and
  Engineering (QCE)},
vol. \bseriesno{1},
pp. \bfpage{928}--\blpage{938}
(\byear{2023}).
\bcomment{IEEE}
\end{bchapter}
\endbibitem

\bibitem[\protect\citeauthoryear{Higgott and Gidney}{2025}]{higgott2025sparse}
\begin{barticle}
\bauthor{\bsnm{Higgott}, \binits{O.}},
\bauthor{\bsnm{Gidney}, \binits{C.}}:
\batitle{Sparse blossom: correcting a million errors per core second with
  minimum-weight matching}.
\bjtitle{Quantum}
\bvolume{9},
\bfpage{1600}
(\byear{2025})
\end{barticle}
\endbibitem

\bibitem[\protect\citeauthoryear{Barber et~al.}{2023}]{barber2023real}
\begin{botherref}
\oauthor{\bsnm{Barber}, \binits{B.}},
\oauthor{\bsnm{Barnes}, \binits{K.M.}},
\oauthor{\bsnm{Bialas}, \binits{T.}},
\oauthor{\bsnm{Bugdayci}, \binits{O.}},
\oauthor{\bsnm{Campbell}, \binits{E.T.}},
\oauthor{\bsnm{Gillespie}, \binits{N.I.}},
\oauthor{\bsnm{Johar}, \binits{K.}},
\oauthor{\bsnm{Rajan}, \binits{R.}},
\oauthor{\bsnm{Richardson}, \binits{A.W.}},
\oauthor{\bsnm{Skoric}, \binits{L.}}, et al.:
A real-time, scalable, fast and highly resource efficient decoder for a quantum
  computer.
arXiv preprint arXiv:2309.05558
(2023)
\end{botherref}
\endbibitem

\bibitem[\protect\citeauthoryear{Delfosse and
  Nickerson}{2021}]{delfosse2021almost}
\begin{barticle}
\bauthor{\bsnm{Delfosse}, \binits{N.}},
\bauthor{\bsnm{Nickerson}, \binits{N.H.}}:
\batitle{Almost-linear time decoding algorithm for topological codes}.
\bjtitle{Quantum}
\bvolume{5},
\bfpage{595}
(\byear{2021})
\end{barticle}
\endbibitem

\bibitem[\protect\citeauthoryear{Harrington}{2004}]{Harrington2004}
\begin{botherref}
\oauthor{\bsnm{Harrington}, \binits{J.W.}}:
{Analysis of Quantum Error-Correcting Codes: Symplectic Lattice Codes and Toric
  Codes}.
PhD thesis,
Caltech
(2004).
\url{https://thesis.library.caltech.edu/1747}
\end{botherref}
\endbibitem

\bibitem[\protect\citeauthoryear{Herold et~al.}{2015}]{herold2015cellular}
\begin{barticle}
\bauthor{\bsnm{Herold}, \binits{M.}},
\bauthor{\bsnm{Campbell}, \binits{E.T.}},
\bauthor{\bsnm{Eisert}, \binits{J.}},
\bauthor{\bsnm{Kastoryano}, \binits{M.J.}}:
\batitle{Cellular-automaton decoders for topological quantum memories}.
\bjtitle{npj Quantum information}
\bvolume{1}(\bissue{1}),
\bfpage{1}--\blpage{8}
(\byear{2015})
\doiurl{10.1038/npjqi.2015.10}
\end{barticle}
\endbibitem

\bibitem[\protect\citeauthoryear{Herold et~al.}{2017}]{herold2017cellular}
\begin{barticle}
\bauthor{\bsnm{Herold}, \binits{M.}},
\bauthor{\bsnm{Kastoryano}, \binits{M.J.}},
\bauthor{\bsnm{Campbell}, \binits{E.T.}},
\bauthor{\bsnm{Eisert}, \binits{J.}}:
\batitle{Cellular automaton decoders of topological quantum memories in the
  fault tolerant setting}.
\bjtitle{New Journal of Physics}
\bvolume{19}(\bissue{6}),
\bfpage{063012}
(\byear{2017})
\end{barticle}
\endbibitem

\bibitem[\protect\citeauthoryear{Balasubramanian
  et~al.}{2024}]{balasubramanian2024local}
\begin{botherref}
\oauthor{\bsnm{Balasubramanian}, \binits{S.}},
\oauthor{\bsnm{Davydova}, \binits{M.}},
\oauthor{\bsnm{Lake}, \binits{E.}}:
A local automaton for the 2d toric code.
arXiv preprint arXiv:2412.19803
(2024)
\end{botherref}
\endbibitem

\bibitem[\protect\citeauthoryear{Kubica and
  Preskill}{2019}]{kubica2019cellular}
\begin{barticle}
\bauthor{\bsnm{Kubica}, \binits{A.}},
\bauthor{\bsnm{Preskill}, \binits{J.}}:
\batitle{Cellular-automaton decoders with provable thresholds for topological
  codes}.
\bjtitle{Physical review letters}
\bvolume{123}(\bissue{2}),
\bfpage{020501}
(\byear{2019})
\doiurl{10.1103/PhysRevLett.123.020501}
\end{barticle}
\endbibitem

\bibitem[\protect\citeauthoryear{Breuckmann et~al.}{2017}]{breuckmann2016local}
\begin{barticle}
\bauthor{\bsnm{Breuckmann}, \binits{N.P.}},
\bauthor{\bsnm{Duivenvoorden}, \binits{K.}},
\bauthor{\bsnm{Michels}, \binits{D.}},
\bauthor{\bsnm{Terhal}, \binits{B.M.}}:
\batitle{Local decoders for the 2d and 4d toric code}.
\bjtitle{Quantum Information and Computation}
\bvolume{17}(\bissue{3-4}),
\bfpage{181}--\blpage{208}
(\byear{2017})
\doiurl{10.1109/TIT.2021.3122352}
\end{barticle}
\endbibitem

\bibitem[\protect\citeauthoryear{Coser and
  Perez-Garcia}{2019}]{coser2019classification}
\begin{barticle}
\bauthor{\bsnm{Coser}, \binits{A.}},
\bauthor{\bsnm{Perez-Garcia}, \binits{D.}}:
\batitle{Classification of phases for mixed states via fast dissipative
  evolution}.
\bjtitle{Quantum}
\bvolume{3},
\bfpage{174}
(\byear{2019})
\end{barticle}
\endbibitem

\bibitem[\protect\citeauthoryear{Rakovszky
  et~al.}{2024}]{rakovszky2024defining}
\begin{barticle}
\bauthor{\bsnm{Rakovszky}, \binits{T.}},
\bauthor{\bsnm{Gopalakrishnan}, \binits{S.}},
\bauthor{\bsnm{Von~Keyserlingk}, \binits{C.}}:
\batitle{Defining stable phases of open quantum systems}.
\bjtitle{Physical Review X}
\bvolume{14}(\bissue{4}),
\bfpage{041031}
(\byear{2024})
\end{barticle}
\endbibitem

\bibitem[\protect\citeauthoryear{Sang et~al.}{2024}]{sang2024mixed}
\begin{barticle}
\bauthor{\bsnm{Sang}, \binits{S.}},
\bauthor{\bsnm{Zou}, \binits{Y.}},
\bauthor{\bsnm{Hsieh}, \binits{T.H.}}:
\batitle{Mixed-state quantum phases: Renormalization and quantum error
  correction}.
\bjtitle{Physical Review X}
\bvolume{14}(\bissue{3}),
\bfpage{031044}
(\byear{2024})
\end{barticle}
\endbibitem

\bibitem[\protect\citeauthoryear{Dennis et~al.}{2002}]{Dennis_2002}
\begin{barticle}
\bauthor{\bsnm{Dennis}, \binits{E.}},
\bauthor{\bsnm{Kitaev}, \binits{A.}},
\bauthor{\bsnm{Landahl}, \binits{A.}},
\bauthor{\bsnm{Preskill}, \binits{J.}}:
\batitle{Topological quantum memory}.
\bjtitle{Journal of Mathematical Physics}
\bvolume{43}(\bissue{9}),
\bfpage{4452}--\blpage{4505}
(\byear{2002})
\doiurl{10.1063/1.1499754}
\end{barticle}
\endbibitem

\bibitem[\protect\citeauthoryear{Fan et~al.}{2024}]{fan2024diagnostics}
\begin{barticle}
\bauthor{\bsnm{Fan}, \binits{R.}},
\bauthor{\bsnm{Bao}, \binits{Y.}},
\bauthor{\bsnm{Altman}, \binits{E.}},
\bauthor{\bsnm{Vishwanath}, \binits{A.}}:
\batitle{Diagnostics of mixed-state topological order and breakdown of quantum
  memory}.
\bjtitle{PRX Quantum}
\bvolume{5}(\bissue{2}),
\bfpage{020343}
(\byear{2024})
\end{barticle}
\endbibitem

\bibitem[\protect\citeauthoryear{Cubitt et~al.}{2015}]{cubitt2015stability}
\begin{barticle}
\bauthor{\bsnm{Cubitt}, \binits{T.S.}},
\bauthor{\bsnm{Lucia}, \binits{A.}},
\bauthor{\bsnm{Michalakis}, \binits{S.}},
\bauthor{\bsnm{Perez-Garcia}, \binits{D.}}:
\batitle{Stability of local quantum dissipative systems}.
\bjtitle{Communications in Mathematical Physics}
\bvolume{337},
\bfpage{1275}--\blpage{1315}
(\byear{2015})
\end{barticle}
\endbibitem

\bibitem[\protect\citeauthoryear{Stahl}{2024}]{stahl2024single}
\begin{barticle}
\bauthor{\bsnm{Stahl}, \binits{C.}}:
\batitle{Single-shot quantum error correction in intertwined toric codes}.
\bjtitle{Physical Review B}
\bvolume{110}(\bissue{7}),
\bfpage{075143}
(\byear{2024})
\end{barticle}
\endbibitem

\bibitem[\protect\citeauthoryear{Kubica and Vasmer}{2022}]{kubica2022single}
\begin{barticle}
\bauthor{\bsnm{Kubica}, \binits{A.}},
\bauthor{\bsnm{Vasmer}, \binits{M.}}:
\batitle{Single-shot quantum error correction with the three-dimensional
  subsystem toric code}.
\bjtitle{Nature Communications}
\bvolume{13}(\bissue{1}),
\bfpage{6272}
(\byear{2022})
\end{barticle}
\endbibitem

\bibitem[\protect\citeauthoryear{Bomb{\'\i}n}{2015}]{bombin2015single}
\begin{barticle}
\bauthor{\bsnm{Bomb{\'\i}n}, \binits{H.}}:
\batitle{Single-shot fault-tolerant quantum error correction}.
\bjtitle{Physical Review X}
\bvolume{5}(\bissue{3}),
\bfpage{031043}
(\byear{2015})
\end{barticle}
\endbibitem

\bibitem[\protect\citeauthoryear{Lake}{2025}]{online}
\begin{botherref}
\oauthor{\bsnm{Lake}, \binits{E.}}:
Local active error correction from simulated confinement.
arXiv preprint arXiv:2510.08056
(2025)
\end{botherref}
\endbibitem

\bibitem[\protect\citeauthoryear{Berman and
  Simon}{1988}]{berman1988investigations}
\begin{bchapter}
\bauthor{\bsnm{Berman}, \binits{P.}},
\bauthor{\bsnm{Simon}, \binits{J.}}:
\bctitle{Investigations of fault-tolerant networks of computers}.
In: \bbtitle{Proceedings of the Twentieth Annual ACM Symposium on Theory of
  Computing},
pp. \bfpage{66}--\blpage{77}
(\byear{1988}).
\doiurl{10.1145/62212.62219}
\end{bchapter}
\endbibitem

\bibitem[\protect\citeauthoryear{Cook et~al.}{2008}]{gacs_synch_slides}
\begin{botherref}
\oauthor{\bsnm{Cook}, \binits{M.}},
\oauthor{\bsnm{Winfree}, \binits{E.}},
\oauthor{\bsnm{Gacs}, \binits{P.}}:
Synchronization in 3 Dimensions.
\url{https://cs-web.bu.edu/faculty/gacs/papers/3DasyncCA-talk.pdf}
\end{botherref}
\endbibitem

\bibitem[\protect\citeauthoryear{Kim and Kosterlitz}{1989}]{kim1989growth}
\begin{barticle}
\bauthor{\bsnm{Kim}, \binits{J.M.}},
\bauthor{\bsnm{Kosterlitz}, \binits{J.}}:
\batitle{Growth in a restricted solid-on-solid model}.
\bjtitle{Physical review letters}
\bvolume{62}(\bissue{19}),
\bfpage{2289}
(\byear{1989})
\end{barticle}
\endbibitem

\bibitem[\protect\citeauthoryear{Gottesman}{2024}]{gottesman2024surviving}
\begin{botherref}
\oauthor{\bsnm{Gottesman}, \binits{D.}}:
Surviving as a quantum computer in a classical world.
Textbook manuscript preprint
(2024)
\end{botherref}
\endbibitem

\bibitem[\protect\citeauthoryear{Gottesman}{2009}]{gottesman2009introductionquantumerrorcorrection}
\begin{botherref}
\oauthor{\bsnm{Gottesman}, \binits{D.}}:
An Introduction to Quantum Error Correction and Fault-Tolerant Quantum
  Computation
(2009).
\url{https://arxiv.org/abs/0904.2557}
\end{botherref}
\endbibitem

\bibitem[\protect\citeauthoryear{Knill et~al.}{1998}]{knill1998resilient}
\begin{barticle}
\bauthor{\bsnm{Knill}, \binits{E.}},
\bauthor{\bsnm{Laflamme}, \binits{R.}},
\bauthor{\bsnm{Zurek}, \binits{W.H.}}:
\batitle{Resilient quantum computation: error models and thresholds}.
\bjtitle{Proceedings of the Royal Society of London. Series A: Mathematical,
  Physical and Engineering Sciences}
\bvolume{454}(\bissue{1969}),
\bfpage{365}--\blpage{384}
(\byear{1998})
\end{barticle}
\endbibitem

\bibitem[\protect\citeauthoryear{Aliferis et~al.}{2005}]{aliferis2005quantum}
\begin{botherref}
\oauthor{\bsnm{Aliferis}, \binits{P.}},
\oauthor{\bsnm{Gottesman}, \binits{D.}},
\oauthor{\bsnm{Preskill}, \binits{J.}}:
Quantum accuracy threshold for concatenated distance-3 codes.
arXiv preprint quant-ph/0504218
(2005)
\end{botherref}
\endbibitem

\bibitem[\protect\citeauthoryear{Gacs}{2024}]{gacs2024probabilistic}
\begin{barticle}
\bauthor{\bsnm{Gacs}, \binits{P.}}:
\batitle{Probabilistic cellular automata with andrei toom}.
\bjtitle{Brazilian Journal of Probability and Statistics}
\bvolume{38}(\bissue{2}),
\bfpage{285}--\blpage{301}
(\byear{2024})
\end{barticle}
\endbibitem

\bibitem[\protect\citeauthoryear{Gacs}{2017}]{gacs_eroder_slides}
\begin{botherref}
\oauthor{\bsnm{Gacs}, \binits{P.}}:
Eroders.
\url{https://cs-web.bu.edu/faculty/gacs/papers/eroders-talk-BU-2017f.pdf}
\end{botherref}
\endbibitem

\bibitem[\protect\citeauthoryear{G{\'a}cs}{2001}]{gacs2001reliable}
\begin{barticle}
\bauthor{\bsnm{G{\'a}cs}, \binits{P.}}:
\batitle{Reliable cellular automata with self-organization}.
\bjtitle{Journal of Statistical Physics}
\bvolume{103},
\bfpage{45}--\blpage{267}
(\byear{2001})
\end{barticle}
\endbibitem

\bibitem[\protect\citeauthoryear{Capuni and Gacs}{2021}]{ccapuni2021reliable}
\begin{botherref}
\oauthor{\bsnm{Capuni}, \binits{I.}},
\oauthor{\bsnm{Gacs}, \binits{P.}}:
A reliable turing machine.
arXiv preprint arXiv:2112.02152
(2021)
\end{botherref}
\endbibitem

\bibitem[\protect\citeauthoryear{Bravyi and Haah}{2011}]{bravyi2011analytic}
\begin{botherref}
\oauthor{\bsnm{Bravyi}, \binits{S.}},
\oauthor{\bsnm{Haah}, \binits{J.}}:
Analytic and numerical demonstration of quantum self-correction in the 3d cubic
  code.
arXiv preprint arXiv:1112.3252
(2011)
\end{botherref}
\endbibitem

\bibitem[\protect\citeauthoryear{Duclos-Cianci and
  Poulin}{2013}]{duclos2013fault}
\begin{botherref}
\oauthor{\bsnm{Duclos-Cianci}, \binits{G.}},
\oauthor{\bsnm{Poulin}, \binits{D.}}:
Fault-tolerant renormalization group decoder for abelian topological codes.
arXiv preprint arXiv:1304.6100
(2013)
\end{botherref}
\endbibitem

\bibitem[\protect\citeauthoryear{Paletta et~al.}{2025}]{paletta2025high}
\begin{botherref}
\oauthor{\bsnm{Paletta}, \binits{L.}},
\oauthor{\bsnm{Leverrier}, \binits{A.}},
\oauthor{\bsnm{Mirrahimi}, \binits{M.}},
\oauthor{\bsnm{Vuillot}, \binits{C.}}:
High-performance local automaton decoders for defect matching in 1d.
arXiv preprint arXiv:2505.10162
(2025)
\end{botherref}
\endbibitem

\bibitem[\protect\citeauthoryear{G{\'a}cs et~al.}{1978}]{gacs1978one}
\begin{barticle}
\bauthor{\bsnm{G{\'a}cs}, \binits{P.}},
\bauthor{\bsnm{Kurdyumov}, \binits{G.L.}},
\bauthor{\bsnm{Levin}, \binits{L.A.}}:
\batitle{One-dimensional uniform arrays that wash out finite islands}.
\bjtitle{Problemy Peredachi Informatsii}
\bvolume{14}(\bissue{3}),
\bfpage{92}--\blpage{96}
(\byear{1978})
\end{barticle}
\endbibitem

\bibitem[\protect\citeauthoryear{Toom}{1995}]{toom1995cellular}
\begin{botherref}
\oauthor{\bsnm{Toom}, \binits{A.}}:
Cellular automata with errors: problems for students of probability.
Topics in contemporary probability and its applications
\textbf{117}
(1995)
\end{botherref}
\endbibitem

\bibitem[\protect\citeauthoryear{Guedes et~al.}{2024}]{guedes2024quantum}
\begin{barticle}
\bauthor{\bsnm{Guedes}, \binits{T.L.}},
\bauthor{\bsnm{Winter}, \binits{D.}},
\bauthor{\bsnm{M{\"u}ller}, \binits{M.}}:
\batitle{Quantum cellular automata for quantum error correction and density
  classification}.
\bjtitle{Physical Review Letters}
\bvolume{133}(\bissue{15}),
\bfpage{150601}
(\byear{2024})
\end{barticle}
\endbibitem

\bibitem[\protect\citeauthoryear{Lang and Buchler}{2018}]{lang2018strictly}
\begin{barticle}
\bauthor{\bsnm{Lang}, \binits{N.}},
\bauthor{\bsnm{Buchler}, \binits{H.P.}}:
\batitle{Strictly local one-dimensional topological quantum error correction
  with symmetry-constrained cellular automata}.
\bjtitle{SciPost Physics}
\bvolume{4}(\bissue{1}),
\bfpage{007}
(\byear{2018})
\end{barticle}
\endbibitem

\bibitem[\protect\citeauthoryear{Schotte et~al.}{2022}]{schotte2022fault}
\begin{botherref}
\oauthor{\bsnm{Schotte}, \binits{A.}},
\oauthor{\bsnm{Burgelman}, \binits{L.}},
\oauthor{\bsnm{Zhu}, \binits{G.}}:
Fault-tolerant error correction for a universal non-abelian topological quantum
  computer at finite temperature.
arXiv preprint arXiv:2301.00054
(2022)
\end{botherref}
\endbibitem

\bibitem[\protect\citeauthoryear{Cirel’son}{2006}]{cirel2006reliable}
\begin{bchapter}
\bauthor{\bsnm{Cirel’son}, \binits{B.}}:
\bctitle{Reliable storage of information in a system of unreliable components
  with local interactions}.
In: \bbtitle{Locally Interacting Systems and Their Application in Biology:
  Proceedings of the School-Seminar on Markov Interaction Processes in Biology,
  Held in Pushchino, Moscow Region, March, 1976},
pp. \bfpage{15}--\blpage{30}
(\byear{2006}).
\bcomment{Springer}
\end{bchapter}
\endbibitem

\bibitem[\protect\citeauthoryear{Chirame et~al.}{2024}]{chirame2024stabilizing}
\begin{botherref}
\oauthor{\bsnm{Chirame}, \binits{S.}},
\oauthor{\bsnm{Prem}, \binits{A.}},
\oauthor{\bsnm{Gopalakrishnan}, \binits{S.}},
\oauthor{\bsnm{Burnell}, \binits{F.J.}}:
Stabilizing non-abelian topological order against heralded noise via local
  lindbladian dynamics.
arXiv preprint arXiv:2410.21402
(2024)
\end{botherref}
\endbibitem

\bibitem[\protect\citeauthoryear{}{2025}]{code}
\begin{botherref}
See the Linked GitHub Repository.
\url{https://github.com/ethanlake/message-passing-decoders}
\end{botherref}
\endbibitem

\bibitem[\protect\citeauthoryear{Aharonov and
  Ben-Or}{1996}]{aharonov1996faulttolerantquantumcomputation}
\begin{botherref}
\oauthor{\bsnm{Aharonov}, \binits{D.}},
\oauthor{\bsnm{Ben-Or}, \binits{M.}}:
Fault Tolerant Quantum Computation with Constant Error Rate
(1996).
\url{https://arxiv.org/abs/quant-ph/9611025}
\end{botherref}
\endbibitem

\bibitem[\protect\citeauthoryear{Wootton}{2015}]{wootton2015simple}
\begin{barticle}
\bauthor{\bsnm{Wootton}, \binits{J.}}:
\batitle{A simple decoder for topological codes}.
\bjtitle{Entropy}
\bvolume{17}(\bissue{4}),
\bfpage{1946}--\blpage{1957}
(\byear{2015})
\end{barticle}
\endbibitem

\bibitem[\protect\citeauthoryear{Hutter et~al.}{2015}]{hutter2015improved}
\begin{barticle}
\bauthor{\bsnm{Hutter}, \binits{A.}},
\bauthor{\bsnm{Loss}, \binits{D.}},
\bauthor{\bsnm{Wootton}, \binits{J.R.}}:
\batitle{Improved hdrg decoders for qudit and non-abelian quantum error
  correction}.
\bjtitle{New Journal of Physics}
\bvolume{17}(\bissue{3}),
\bfpage{035017}
(\byear{2015})
\end{barticle}
\endbibitem

\bibitem[\protect\citeauthoryear{Brell et~al.}{2014}]{brell2014thermalization}
\begin{barticle}
\bauthor{\bsnm{Brell}, \binits{C.G.}},
\bauthor{\bsnm{Burton}, \binits{S.}},
\bauthor{\bsnm{Dauphinais}, \binits{G.}},
\bauthor{\bsnm{Flammia}, \binits{S.T.}},
\bauthor{\bsnm{Poulin}, \binits{D.}}:
\batitle{Thermalization, error correction, and memory lifetime for ising anyon
  systems}.
\bjtitle{Physical Review X}
\bvolume{4}(\bissue{3}),
\bfpage{031058}
(\byear{2014})
\end{barticle}
\endbibitem

\bibitem[\protect\citeauthoryear{Wootton et~al.}{2014}]{wootton2014error}
\begin{barticle}
\bauthor{\bsnm{Wootton}, \binits{J.R.}},
\bauthor{\bsnm{Burri}, \binits{J.}},
\bauthor{\bsnm{Iblisdir}, \binits{S.}},
\bauthor{\bsnm{Loss}, \binits{D.}}:
\batitle{Error correction for non-abelian topological quantum computation}.
\bjtitle{Physical review X}
\bvolume{4}(\bissue{1}),
\bfpage{011051}
(\byear{2014})
\end{barticle}
\endbibitem

\bibitem[\protect\citeauthoryear{Dauphinais and
  Poulin}{2017}]{dauphinais2017fault}
\begin{barticle}
\bauthor{\bsnm{Dauphinais}, \binits{G.}},
\bauthor{\bsnm{Poulin}, \binits{D.}}:
\batitle{Fault-tolerant quantum error correction for non-abelian anyons}.
\bjtitle{Communications in Mathematical Physics}
\bvolume{355},
\bfpage{519}--\blpage{560}
(\byear{2017})
\end{barticle}
\endbibitem

\bibitem[\protect\citeauthoryear{Schotte et~al.}{2022}]{schotte2022quantum}
\begin{barticle}
\bauthor{\bsnm{Schotte}, \binits{A.}},
\bauthor{\bsnm{Zhu}, \binits{G.}},
\bauthor{\bsnm{Burgelman}, \binits{L.}},
\bauthor{\bsnm{Verstraete}, \binits{F.}}:
\batitle{Quantum error correction thresholds for the universal fibonacci
  turaev-viro code}.
\bjtitle{Physical Review X}
\bvolume{12}(\bissue{2}),
\bfpage{021012}
(\byear{2022})
\end{barticle}
\endbibitem

\bibitem[\protect\citeauthoryear{Wootton and Hutter}{2016}]{wootton2016active}
\begin{barticle}
\bauthor{\bsnm{Wootton}, \binits{J.R.}},
\bauthor{\bsnm{Hutter}, \binits{A.}}:
\batitle{Active error correction for abelian and non-abelian anyons}.
\bjtitle{Physical Review A}
\bvolume{93}(\bissue{2}),
\bfpage{022318}
(\byear{2016})
\end{barticle}
\endbibitem

\bibitem[\protect\citeauthoryear{Breuer and
  Petruccione}{2002}]{breuer2002theory}
\begin{botherref}
\oauthor{\bsnm{Breuer}, \binits{H.-P.}},
\oauthor{\bsnm{Petruccione}, \binits{F.}}:
The theory of open quantum systems
(2002)
\end{botherref}
\endbibitem

\bibitem[\protect\citeauthoryear{Gray}{1999}]{gray1999toom}
\begin{botherref}
\oauthor{\bsnm{Gray}, \binits{L.F.}}:
Toom’s stability theorem in continuous time.
Perplexing Problems in Probability: Festschrift in Honor of Harry Kesten,
331--353
(1999)
\end{botherref}
\endbibitem

\bibitem[\protect\citeauthoryear{Bennett and Grinstein}{1985}]{bennett1985role}
\begin{barticle}
\bauthor{\bsnm{Bennett}, \binits{C.H.}},
\bauthor{\bsnm{Grinstein}, \binits{G.}}:
\batitle{Role of irreversibility in stabilizing complex and nonergodic behavior
  in locally interacting discrete systems}.
\bjtitle{Physical review letters}
\bvolume{55}(\bissue{7}),
\bfpage{657}
(\byear{1985})
\end{barticle}
\endbibitem

\bibitem[\protect\citeauthoryear{Ray et~al.}{2024}]{ray2024protecting}
\begin{barticle}
\bauthor{\bsnm{Ray}, \binits{A.}},
\bauthor{\bsnm{Laflamme}, \binits{R.}},
\bauthor{\bsnm{Kubica}, \binits{A.}}:
\batitle{Protecting information via probabilistic cellular automata}.
\bjtitle{Physical Review E}
\bvolume{109}(\bissue{4}),
\bfpage{044141}
(\byear{2024})
\end{barticle}
\endbibitem

\bibitem[\protect\citeauthoryear{Lake and Ro}{2025}]{squeezing}
\begin{botherref}
\oauthor{\bsnm{Lake}, \binits{E.}},
\oauthor{\bsnm{Ro}, \binits{S.}}:
Squeezing codes: robust fluctuation-stabilized memories.
arXiv preprint arXiv:2509.20730
(2025)
\end{botherref}
\endbibitem

\bibitem[\protect\citeauthoryear{Friedlander}{1975}]{friedlander1975wave}
\begin{botherref}
\oauthor{\bsnm{Friedlander}, \binits{F.G.}}:
The wave equation on a curved space-time
\textbf{2}
(1975)
\end{botherref}
\endbibitem

\bibitem[\protect\citeauthoryear{Biswas and
  Leyvraz}{2021}]{biswas2021ballistic}
\begin{barticle}
\bauthor{\bsnm{Biswas}, \binits{S.}},
\bauthor{\bsnm{Leyvraz}, \binits{F.}}:
\batitle{Ballistic annihilation in one dimension: a critical review}.
\bjtitle{The European Physical Journal B}
\bvolume{94}(\bissue{12}),
\bfpage{240}
(\byear{2021})
\end{barticle}
\endbibitem

\end{thebibliography}
	
\end{document}